\documentclass[american, aps,pra,reprint,superscriptaddress, floatfix,   
nofootinbib,
amsmath,amssymb,
pra,
]{revtex4-1}
\usepackage{algorithm}
\usepackage{algpseudocode}
\usepackage[T1]{fontenc}
\usepackage{inputenc}
\usepackage{color}
\usepackage{babel}
\usepackage{physics}
\usepackage{mathtools}
\usepackage{bbding}
\usepackage{pifont}
\usepackage{textcomp}
\usepackage{wasysym}
\usepackage{amsthm}
\usepackage[normalem]{ulem}
\usepackage{nicefrac}
\usepackage{wasysym}
\usepackage{dsfont}
\usepackage{amsmath}
\usepackage{amssymb}
\usepackage{graphicx}

\usepackage{notes2bib}
\bibnotesetup{
note-name = ,
use-sort-key = false
}

\usepackage{hyperref}

\hypersetup{
     colorlinks   = true,
     linkcolor    = blue,
     citecolor    = blue,
     urlcolor     = blue,
     }

\usepackage{appendix}
\renewcommand{\selectlanguage}[1]{}
\makeatletter
\@ifundefined{textcolor}{}
{%
	\definecolor{BLACK}{gray}{0}
	\definecolor{WHITE}{gray}{1}
	\definecolor{RED}{rgb}{1,0,0}
	\definecolor{GREEN}{rgb}{0,1,0}
	\definecolor{BLUE}{rgb}{0,0,1}
	\definecolor{CYAN}{cmyk}{1,0,0,0}
	\definecolor{MAGENTA}{cmyk}{0,1,0,0}
	\definecolor{YELLOW}{cmyk}{0,0,1,0}
}
\theoremstyle{plain}

\theoremstyle{plain}

\ifx\proof\undefined
\newenvironment{proof}[1][\protect\proofname]{\par
	\normalfont\topsep6\p@\@plus6\p@\relax
	\Trivlist
	\itemindent\parindent
	\item[\hskip\labelsep
	\scshape
	#1]\ignorespaces
}{%
	\endtrivlist\@endpefalse
}
\providecommand{\proofname}{\textbf{Proof}}
\fi
\theoremstyle{plain}

\renewenvironment{proof}[1][\proofname]{\noindent {\bfseries #1.} }{\qed}

\providecommand{\lemmaname}{Lemma}
\providecommand{\definitionname}{Definition}
\providecommand{\propositionname}{Proposition}

\usepackage{babel}
\usepackage{txfonts}
\usepackage{colortbl}\definecolor{myurlcolor}{rgb}{0,0,0.7}

\renewcommand{\ket}[1]{\left| #1 \right\rangle}

\renewcommand{\ketbra}[2]{\left|#1\middle\rangle\!\middle\langle#2\right|}
\newcommand{\proj}[1]{\ketbra{#1}{#1}}

\newcommand{\St}[1]{\mathsf{St} (#1)}

\DeclareMathOperator{\N}{\mathcal{N}}
\newcommand{\haH}

\DeclareMathOperator{\CM}{\mathrm{CE}_ {\max}}
\DeclareMathOperator{\Cm}{\mathrm{CE}_ {\min}}

\newtheorem{theorem}{Theorem}
\newtheorem*{thm*}{Theorem}

\newtheorem{lemma}{Lemma}
\newtheorem*{lem*}{Lemma}

\newtheorem{corollary}{Corollary}
\newtheorem*{cor*}{Corollary}
\newtheorem{proposition}{Proposition}
\newtheorem*{prop*}{Proposition}
\newtheorem{definition}{Definition}

\def\H{\mathcal{H}}

\makeatother
\usepackage{subfiles}

\begin{document}
\preprint{APS/123-QED}
\title{Maximum and minimum causal effects of physical processes}
\author{Giulio Chiribella}
\email[]{giulio@cs.hku.hk}
\affiliation{QICI Quantum Information and Computation Initiative,  School of Computing and Data Science,
The University of Hong Kong, Pokfulam Road, Hong Kong}	 
\affiliation{Department of Computer Science, Parks Road, Oxford, OX1 3QD, United Kingdom}	
\affiliation{Perimeter Institute for Theoretical Physics, Waterloo, Ontario N2L 2Y5, Canada}

\author{Kaumudibikash Goswami}\email{kaumudi@hku.hk}
\affiliation{QICI Quantum Information and Computation Initiative, School of Computing and Data Science, The University of Hong Kong, Pokfulam Road, Hong Kong}

\begin{abstract}

We introduce two quantitative measures of the strength of causal relations in quantum theory and more  general physical theories.  These two measures, called the maximum and minimum causal effect,    quantify  the maximum and minimum changes in the output of a quantum process induced by  changes in its input.  The maximum and minimum causal effect possess useful properties, such as continuity and data-processing inequality. In quantum theory, they have close connections with quantum information tasks. The maximum quantum causal effect can be used to detect quantum channels with nonzero capacity for transmitting classical information.  The minimum causal effect can be used to guarantee the recoverability of quantum information:  every quantum process with a high value of the minimum quantum causal effect can be approximately inverted. Moreover, we show that the quantum causal effects satisfy a duality relation: if the  minimum causal effect of a quantum system $A$ on another quantum system $B$ is high, then $A$ must have a low value of the maximum causal effect on any other quantum system $B'$ that is  spacelike separated with $B$.  This duality implies a monogamy relation for quantum causal effects, and represents a fundamental difference between quantum and classical causal relations. 

We illustrate the application of the maximum causal effect to the analysis of  two paradigmatic examples, the first involving a coherent  superposition of direct cause and common cause, and the second involving  a coherent  superposition of multiple quantum  processes. 
\end{abstract}

\maketitle
\section{Introduction}

Causal inference, the task of inferring cause-effect relations, is a key component of the scientific method. In the classical domain,  the theory of Bayesian networks~\cite{pearl_2009} provides a powerful tool for identifying the causal structure of a network of events.
  Recently, there has been increasing interest in the extension of causal inference to the quantum domain~\cite{costa2016, Spekkens_cause, barrett_2019, Barrett_2021}. This extension is motivated by the increasing degree of control on quantum systems, which enables the experimental study of causal relations at the quantum scale~\cite{Ried_2015,MacLean_2017, Carvacho2017, Chaves2017,  Agresti2022}, and by the variety of new causal structures emerging from quantum mechanics,  which includes   coherent quantum superpositions of common cause and direct cause~\cite{MacLean_2017, Feix_2017}, indefinite causal order~\cite{chiribella09,oreshkov12}, and indefinite input-output direction~\cite{chiribella_Indefinite2022}.

 A central task in quantum causal inference is to quantify the strength of causal relationships among quantum systems. 
The simplest instance of this scenario is to quantify the dependence between the inputs and outputs of a given quantum process~\cite{janzing_2013,henson2014theory, Chiribella_quantum_speedup_2019, chiribella2021fast,  perinotti_2021,escola_2022, Bai_2022, yi_2022, hutter2022quantifying}. The problem is illustrated in Fig. \ref{fig:settings}: an experimenter has black-box access to an uncharacterized quantum process $\N$, transforming an input quantum system $A$ into an output quantum system $B$.  The problem is to determine the strength of the causal influence of system $A$ on system $B$  by observing how the statistics of measurements on system $B$ is affected by changes in the state of system $A$.  Here, no assumption is made on the internal structure of the process $\N$: in particular, no assumption on  the structure of the interactions  producing the state of system $B$ from the state of system $A$, and on additional variables that may be involved in these interactions (later in the paper, we will extend the analysis to scenarios that involve additional variables, making connection between the black box setting described above and  causal modeling approaches  based on the structure of interactions~\cite{barrett_2019,perinotti_2021,perinotti_2023}.)

\begin{figure}
    \centering
    \includegraphics[width=0.8\columnwidth]{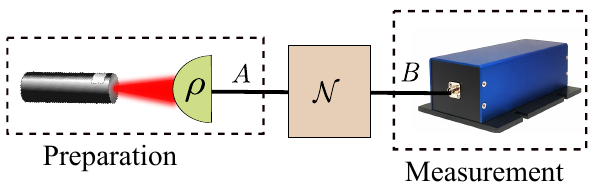}
    \caption{{\bf Quantifying  causal dependencies in the black box scenario.  } An experimenter has black box access to a physical process $\N$ transforming an input system $A$ into an output system $B$.  The process is probed by preparing states of system $A$ and performing measurements on system $B$.  The dependence of the measurement statistics  on the state preparation is used as an indicator of the causal influence of system $A$ on system $B$. 
    }
    \label{fig:settings}
\end{figure}

In classical causal inference, a commonly used measure of causal influence is the \emph{average causal effect} (ACE)~\cite{angrist_1996,balke_1997,pearl_2009}. Explicitly, the ACE quantifies the changes in the probability distribution of an output variable $B$  induced by changes in the value of an input variable $A$.
In the quantum case, an analogue of the ACE was put forward in Ref. \cite{hutter2022quantifying}, where the authors defined a ``quantum ACE''  as the average distance between two output states of system $B$ corresponding to two randomly chosen orthogonal pure states of system $A$. The randomization over the possible input states, however,   represents a departure from the original definition of ACE, which included a maximization over the  possible pairs of input states, rather than an average. As a consequence, the quantum ACE  of Ref. \cite{hutter2022quantifying}, while providing a valuable indicator of causal influence, does not reduce to the classical ACE when the process $\N$ is of the classical form ({\em i.e.} when $\N$ can be realized by   measuring system $A$ in an orthonormal basis and by re-preparing   system $B$ in states of another orthonormal basis.)

In this paper, we propose two   measures of the strength of causal relations, applicable to quantum theory and to more general probabilistic theories  \cite{hardy2001quantum, Barrett2007,chiribella2010probabilistic, Barnum2011,chiribella2016quantum,DAriano2016, Plvala2023}.   Our first measure, called the {\em maximum causal effect} $\CM$, is defined as  the maximum ratio  
  between the distance of two  output states and the distance of the corresponding input states, maximized over all possible pairs of input states.  This quantity provides  the direct extension of the ACE to quantum theory and other physical theories, and reduces to the classical ACE in the case of classical theory.   
   The second measure,  called the {\em minimum causal effect} $\Cm$, is 
   defined as the minimum ratio between the distance of output states and the distance of the  input states.  In quantum theory,   the minimum causal effect  witnesses causal relations  that cannot be realized by measuring the input system in an orthonormal basis.
Both quantities are continuous in their argument,  and satisfy a data-processing inequality, stating that the subsequent application of physical processes cannot increase the causal effect.

The maximum and minimum quantum causal effects have close  connections with quantum information tasks.   ${\CM}$ provides a lower bound on the classical capacity \cite{Holevo73, Holevo1998,Schumacher1997},  the rate at which the given process can transfer classical messages.  $\Cm$ is closely  connected with the problems of minimum error state discrimination \cite{helstrom_quantum_1976} and approximate recovery of quantum information \cite{fawzi2015quantum,Junge_2018}. Specifically, we prove that the minimum causal effect of a quantum channel is close to one if and only if the channel approximately preserves the distinguishability of quantum states, which in turn happens  if and only if the channel is approximately correctable. Our finding provides  the robust  version of a result by Blume-Kohout,  Ng, Poulin, and Viola \cite{Blume-kohut_IPS_2010}, who showed that a quantum channel is perfectly correctable if and only if it perfectly preserves the distinguishability of quantum states.  

Using the connection with recoverability, we then prove a duality between the minimum and maximum quantum causal effects:  if system $A$ has a high value of the minimum causal effect on quantum system $B$, then it must have a low value of the maximum causal effect  on every  quantum system $B'$ that is spacelike separated with $B$. In turn, this duality implies a monogamy property: a quantum system  cannot have a high value of the minimum causal effect on multiple systems simulateously.    Monogamy of the minimum  causal effect marks a fundamental difference between quantum and classical causal inference:  while the causal influence of a classical variable  can be arbitrarily high on multiple variables simultaneously, due to the possibility of broadcasting information, in quantum physics high values of the minimum causal effect on a system preclude causal effects on other spacelike systems.  This property is closely connected with the monogamy of entanglement  \cite{coffman2000distributed,osborne2006general},  the  no-cloning  \cite{wootters1982single,dieks1982communication}  and no-broadcasting \cite{barnum1996noncommuting,barnum2007generalized} theorems, and the information-disturbance tradeoff \cite{fuchs1996quantum,kretschmann2008information}. 

We then use our notions of causal effects to analyze  two paradigmatic examples of quantum causal relations.  In the first example, we use  the maximum and minimum quantum causal effect to quantify the amount of direct cause relation present in   a quantum setup  exhibiting  a coherent superposition of common cause and direct cause \cite{MacLean_2017}.      The second example involves a coherent superposition of multiple quantum processes \cite{aharonov1990superpositions,oi2003interference,aaberg2004operations,gisin2005error,Chiribella_2019,Abbott2020communication,witnessing_latent_time}. 
In this case, we provide an upper bound on  the  enhancement of the maximum causal effect due to the quantum superposition, and 
we show that the bound is tight when the original process has zero causal effect.  All channels with zero causal effect can be realized by discarding the input system and initializing the output system in a fixed state.  The set of channels with zero causal effect includes in particular completely depolarizing channels, which have been previously studied in the literature on quantum communication with superpositions of noisy channels \cite{Abbott2020communication,witnessing_latent_time}.

Finally, we develop a variational algorithm for estimating a channel's maximum quantum causal effect ${\CM}$. The algorithm optimizes over pairs of  input states for a process and over measurements performed at the process' output.    At every iteration, it   
outputs the settings for the next state preparation and measurement.    This approach can be used both for estimating  the maximum causal effect either through numerical computation or using  experimental data.  We test the performance of the algorithm for numerical calculation of $\CM$, benchmarking the numerical values with the analytical expressions obtained in the two examples of common cause/direct cause process and superposition of quantum processes.

The paper is organized as follows. In Section~\ref{sec:theory}, we introduce the maximum and minimum causal effects in the framework of general probabilistic theories. In Section~\ref{sec:quantum}, we specialize the notions to the case of quantum theory, making connections with quantum information tasks, and establishing the monogamy of quantum casal relations.  In  Section~\ref{sec:applications}, we apply the notions of maximum and minimum quantum causal effect to two examples of quantum causal relations. Finally, in Section~\ref{sec:algorithm}, we provide the variational algorithm for the estimation of the maximum quantum causal effect.  Finally, in Section~\ref{sec:discussion}, we provide the conclusions.

\section{Causal effects} \label{sec:theory}

\subsection{Average causal effect for classical random variables}

In classical causal inference, a popular  measure  of the strength of the causal influence is  the  \emph{average causal effect} (ACE)~\cite{Holland1988,pearl_2009}.  
For two binary variables $A$ and $B$,   with values $\{a_0,a_1\}$ and $\{b_0,b_1\}$, respectively,  the ACE is defined as
\begin{align}
    \mathrm{ACE}_{A\to B}: =   \abs{P\big(b_1|\mathrm{do}(a_1)\big)-P\big(b_1|\mathrm{do}(a_0)\big)}\, ,\label{Eq:ACE_define_classical}
\end{align}
where $p(b|{\rm do}(a))$  is the probability distribution of output variable $B$ when the input variable $A$ has been reset to the value $a$.  In this case, the ACE can be equivalently written as 
\begin{align}\label{ACEnorm}
    \mathrm{ACE}_{A\to B}    =  \frac{  \|   P_{{\rm do}  (a_1)}   -  P_{{\rm do}  (a_0)}   \|_1}2 \, , 
\end{align}
where $\|  P_1-  P_2\|_1  :  =  \sum_b  |  P_1(b)-  P_2(b)|$  is the total variation distance of two generic probability distributions  $P_1(b)$ and $P_2(b)$, and we used the notation $ P_{{\rm do}  (a)}  (b):  =  P(  b| {\rm do}(a))$.

For non-binary variables $A$ and $B$ with a finite set of possible values, a natural extension is to define the ACE as  
\begin{align}\label{ACEnormnonbinary}
    \mathrm{ACE}_{A\to B}:= \max_{a, a'}     \frac{ \| P_{{\rm do}  (a)}  -  P_{{\rm do}  (a')} \|_1   }2 \, ,
\end{align}
where the maximum is over all possible values $a$ and $a'$ of the random variable $A$ \cite{hutter2022quantifying}.    In other words, the ACE quantifies the statistical distinguishability of the probability distributions for variable $B$ obtained by varying the values of variable $A$.

We now provide an alternative expression of the ACE that facilitates the transition from classical theory to quantum theory and, more broadly, to general probabilistic theories.  First, we observe that the conditional probability distribution $p(b|  {\rm do}  (a))$  can be interpreted as a stochastic process that generates the value of variable $B$ from the value of variable $A$ set by the do-intervention.  We  denote this stochastic   process  by $\cal N$.
In general,  process $\N$ will transform probability distributions for  variable $A$ into probability distributions for  variable $B$: specifically, an input  probability distribution $P(a)$  will be transformed into an output  probability distribution  $Q(b)  $, defined by    
\begin{align}
\nonumber  Q(b)   & = \sum_a   P(b|{\rm do}  (a)  ) \,   P\big({\rm do} (a) \big)  \\
 &   =:\left[{\cal N}   (P)\right]\, (b)\, . 
\label{classicalchannel}
\end{align}

Using this notation, we can now provide an alternative expression of the ACE:  
\begin{proposition}\label{prop:classicalACE}
Let  $A$ and $B$ be two random variables with finite number of values, and let $\N$ be the process described by the conditional probability distribution $p(b|{\rm do}  (a))$.  Then,    the following equality holds
\begin{align}
\label{almostGPT}
{\rm ACE}_{A\to B}     =   \sup_{P \not  =    P'}  \frac{  \|   {\cal N}  (P)   -  {\cal N}   (P')\|_1}{\|  P  -  P'\|_1 \, ,}
\end{align}
where the supremum is over all pairs of probability distributions  $P$ and $P'$ for variable $A$,  and $\N  (P)$ and $\N(P')$ are  the corresponding probability distributions for variable $B$, as defined  in Eq. (\ref{classicalchannel}).    
\end{proposition}
The proof is provided in Appendix \ref{app:classicalexpression}.  The above proposition shows that the  ACE can be equivalently interpreted as a measure of a process' ability  to preserve the distinguishability of probability distributions. This interpretation establishes a link between causality and information theory, in particular state discrimination, communication, and error correction, as we will see in more detail later in the paper.   Moreover, Eq. (\ref{almostGPT}) directly suggests a way to generalize the causal effect to physical theories beyond classical theory, as we will do in the next section.

\subsection{From classical theory to general probabilistic theories}  

We now provide the extension of the ACE to general probabilistic theories.   We will replace the random variables $A$ and $B$ in  Eq. \eqref{almostGPT} with physical systems, and will  model the causal relation between $A$ and $B$ as a physical process $\cal N$, transforming states of system $A$ into states of system $B$ \cite{henson2014theory,Chiribella_quantum_speedup_2019}.   The distinguishability of states in a general probabilistic theory will be quantified a suitable notion of norm, which generalizes the total variation distance, and coincides with the trace norm in the quantum case. 

\subsubsection{Framework and notation}
In the framework of general probabilistic theories  \cite{hardy2001quantum, Barrett2007,chiribella2010probabilistic, Barnum2011,chiribella2016quantum,DAriano2016, Plvala2023}, the states of a physical system are  described by  elements of a convex set, called the {\em state space}. For a  generic system $S$, we will denote the state space by ${\sf St}  (S)$.   For simplicity, we will restrict our attention to systems with finite-dimensional state spaces and we will make the standard assumption that, for every system $S$,  
the state space ${\sf St}  (S)$ is compact.

Classical and quantum theory are canonical examples of general probabilistic theories. In classical theory,   physical systems are associated to  random variables, and the  state space of a system is  the convex set of all possible probability distributions for the corresponding random variable. In quantum theory,  physical systems are associated to Hilbert spaces, and the  state space of a system $S$ is the set of all possible density matrices on the corresponding Hilbert space $\H _S$.

In a general probabilistic theory, the composite system made of systems $A$ and $B$ will be denoted by $A\otimes B$.  For states $\rho  \in {\sf St} (A)$ and $\sigma  \in {\sf St} (B)$, the notation $\rho\otimes \sigma$ will be used to denote the state obtained by an  independent preparation of systems $A$ and $B$ in the states $\rho$ and $\sigma$, respectively.   In quantum theory, this notation coincides with the standard tensor product notation.

The set of physical processes transforming an input system $S$ into an output system $S'$ is represented by a set of affine maps  ${\cal P}:  {\sf St} (S) \to {\sf St} (S')$.  In classical theory, a process $\cal P$ is described by  a conditional probability distribution $P(s'|  s)$.     In quantum theory, a process $\cal P$ is described by a completely positive trace-preserving linear map transforming linear operators on the Hilbert space $  \H _S $ into linear operators on the Hilbert space ${\cal  H}_{S'}$.

Both in classical and quantum theory,  there is a unique way to discard physical systems.  In classical theory,  the discarding process is implemented by marginalization:  for a joint probability distribution $P_{AB}  (a,b)$ describing the state of a composite system $A\otimes B$, discarding system $B$ leaves system $A$ in the state described by the marginal probability distribution $P_A  (a):  =\sum_{b}P_{AB}  (a,b)$. In quantum theory, the discarding process is mathematically described by the partial trace:  for a composite system $A\otimes B$ in the state $\rho_{AB}$,  discarding system $B$  leaves system $A$ in the state described by the reduced density matrix  $\rho_A:  =  \Tr_B  [\rho_{AB}]$.  

In general probabilistic theories beyond classical and quantum theory,  the uniqueness of the discarding operation is a non-trivial axiom,  known as the Causality Axiom  \cite{chiribella2010probabilistic,chiribella2016quantum,DAriano2016}.    In the following we will restrict our attention to theories satisfying the Causality Axiom. For a generic system $B$,  we will denote the discarding process by $\Tr_B$, following the quantum notation.

\subsubsection{The operational norm}  
To extend  Eq. (\ref{almostGPT}) from classical theory to a general probabilistic theory we need a notion of distance that quantifies the distinguishability of states. Operationally, distinguishing among the states of a given system requires performing a measurement.
In a general probabilistic theory, a measurement  is described by a collection of effects \cite{hardy2001quantum, Barrett2007,chiribella2010probabilistic, Barnum2011,chiribella2016quantum,DAriano2016, Plvala2023}, that is, positive affine functionals  acting on the system's state space. For example, an $N$-outcome measurement on system $S$ is described by a collection $(e_i)_{i=1}^N$ where each $e_i  : {\sf St} (  S) \to \mathbb R$ is an affine functional and satisfies the positivity condition $e_i (\rho)  \ge 0  \, , \forall \rho  \in {\sf St}  (S)$.   When the measurement is performed, the probability of the $i$-th outcome is given by 
\begin{align}
p (i|\rho)  :=  e_i  (\rho)\, ,\qquad \forall i\in  \{1,\dots, N\}
\end{align}
where $\rho$ is the state of the system immediately before the measurement.  The normalization of the probability distribution  then amounts to the condition $\sum_i  e_i  ( \rho)  =1 \, ,\forall \rho \in  {\sf St}  (S)$.  

The distinguishability of two states $\rho$ and $\rho'$ can then be defined in terms of the distinguishability of the resulting probability distributions $p (i|\rho)$ and $p(i|\rho')$, optimized over all possible measurements allowed by the theory.  The resulting quantity

\begin{align} 
\| \rho  -  \rho'\|_1  :  =  \max_N\max_{  (e_i)_{i=1}^N}  \,  \sum_{i=1}^N  \big|   e_i  (\rho)  -  e_i  (\rho')\big|   \label{Eq:operational_norm2}
\end{align}
is known as the  {\em operational norm} \cite{chiribella2010probabilistic}.      It is easy to see that the operational norm satisfies the bounds $0\le \|  \rho -  \rho'\|_1 \le 2$. The equality $\|\rho-  \rho'\|_1  =  0$ holds if and only if $\rho =  \rho'$, while the equality   $\|\rho-  \rho'\|_1  =  2$ holds if and only if the states are perfectly distinguishable, that is, if there is a measurement that gives rise to probability distributions  $p(i|\rho)$ and $p(i|\rho')$ with disjoint supports.

As the name suggests, the operational norm is indeed a norm on the vector space ${\sf St}_{\mathbb R}  (S)$ consisting of linear combination of the form $\sum_j  c_j  \,  \rho_j$, with  $\rho_j\in{\sf St}  (S)$ and $c_j\in \mathbb R$. For a generic element of the vector space ${\sf St}_{\mathbb R}  (S)$, the operational norm is given by \cite{chiribella2010probabilistic}  
\begin{align}\label{opnorm}
\|   \delta  \|_1  :  =  \max_{e \in  {\sf Eff} (S)}  e  (\delta)  -  \min_{f  \in {\sf Eff}  (S)}   f(\delta)\, ,
\end{align}
where   ${\sf Eff} (S)$ is  the space of all effects for system $S$. 

In quantum theory,  the vector space $ {\sf St}_{\mathbb R}  (S)$ is the space of all possible (trace-class) Hermitian  operators on the system's Hilbert space, and  the set ${\sf Eff}  (S)$ consists of all functionals of the form $e  (  \rho): = {\sf Tr} [  P \rho]$, where $ P$  is an operator satisfying the condition $0\le P\le I$.    Hence, the operational norm of a Hermitian operator $\delta$ is  $\|  \delta\|_1    =  \max_{0\le P \le I}   \Tr [P \rho]  -  \min_{0\le Q \le I}   \Tr [Q \rho]$, and  coincides with the trace norm $\|  \delta\|_1 :  =  \Tr [\delta_+]  + \Tr[\delta_-]  $, where $\delta_+$ and $\delta_-$ are the positive and negative parts of the operator $\delta$.

\subsection{The maximum causal effect} \label{Sec:CE_GPT}

We are now ready to extend  the notion of ACE to general probabilistic theories. Our definition refers to the scenario where an experimenter has black-box access to a physical process $\N$, transforming an input system $A$ into an output system $B$, as in Fig. \ref{fig:settings}.    To quantify the strength of the causal influence induced by process $\N$, we use the direct extension of Eq. (\ref{almostGPT}):
\begin{definition}[Maximum causal effect]\label{def:GCEmax}
In a general probabilistic theory, the {\em maximum causal effect} induced by a process ${\cal N}:  {\sf St }  (A)  \to {\sf St}  (B)$ is 
\begin{align}\label{sup} 
      \CM  (\N):  = \sup_{\rho,\, \rho'  \in  {\sf St}(A)  ,\,  \rho \not =  \rho'} \frac{  \|  {\cal N}  (\rho)  -  {\cal N}  (\rho')  \|_1  }{  \|   \rho  -   \rho'  \|_1}   \,.
\end{align}     
\end{definition}
Later in this subsection, we will see that, under an assumption  satisfied both by classical and quantum theory,  the supremum is actually a maximum and can be obtained by maximizing over pairs of perfectly distinguishable pure states.  

Intuitively, $\CM$  quantifies the maximum change  in the state of system $B$ due to a change in the state of system $A$.   In the special case of classical theory,  $\CM$  reduces to the ACE, as one can see from Proposition \ref{prop:classicalACE}. 

An interpretation of the  maximum causal effect  in terms of do-operations, analogous  to the do-operations in classical causal inference, is provided in Appendix  \ref{app:foundations}. There,  we also show that $\CM$  is an upper bound  to the maximum ACE between a classical variable describing the preparation of system $A$ and a classical variable associated with the outcomes of a measurement on system $B$, with the preparation and measurement optimized over all possible settings permitted by the theory. For a class of theories, including  quantum theory, the equality sign holds, thereby providing an alternative characterization of $\CM$ in terms of the maximum ACE between state preparation variables on system $A$ and measurement outcome variables on system $B$.

A few elementary properties of $\CM$  are provided by the following lemma:  
\begin{proposition}\label{prop:propertiesmax}
 For every probabilistic theory satisfying the Causality Axiom, every pair of systems $A$ and $B$ and every process ${\cal N}:  {\sf St }  (A)  \to {\sf St}  (B)$, the maximum causal effect has the  following properties: 
 \begin{enumerate}
 \item {\em Range:}  $0\le  \CM   ({\cal N}) \le 1$,
 \item {\em Faithfulness:}  $\CM   ({\cal N})=0$ if and only if $\cal N$ is a discard-and-reprepare process, that is, if and only if  it is of the form ${\cal N}   =  \sigma_0 \Tr_A$ for some fixed state $\sigma_0 \in {\sf St}  (B)$, 
 
 \item {\em Convexity:}   $ \CM    (\sum_i \, p_i \, {\cal N}_i) \le  \sum_i \,  p_i  \, \CM    ( {\cal N}_i)$ for every collection of processes $({\cal N}_i)_i$ with input $A$ and output $B$, and for every probability distribution $(p_i)_i$.
 \item {\em Data-processing inequality:}   $  \CM    ( {\cal  B}  \circ {\cal N}   \circ {\cal A})  \le  \CM    (  {\cal N} )$ for arbitrary processes ${\cal A}  : {\sf St}  ( A')  \to  {\sf St}  ( A),$  ${\cal N}  : {\sf St}  ( A)  \to  {\sf St}  ( B)$,  and ${\cal B}  : {\sf St}  ( B)  \to  {\sf St}  ( B')$, and arbitrary systems $A,A',B,$ and $B'$. 
 \end{enumerate}
\end{proposition}

The proof is provided in Appendix \ref{app:basicproperties_Max}.   An immediate consequence of the data-processing inequality  property is that every correctable process has maximum causal effect equal to one:  
\begin{corollary}\label{cor:maxforcorrectable}
Let ${\cal N}  : {\sf St}  ( A)  \to  {\sf St}  ( B)$ be a correctable process, meaning that there exists a recovery process ${\cal R}  : {\sf St}  ( B)  \to  {\sf St}  ( A)$ such that $ {\cal R}  \circ {\cal N}   (\rho)  = \rho$ for every state $\rho \in  {\sf St}  (A)$.    Then,  $\CM   ({\cal N})  =1$.
\end{corollary}
{\bf Proof.} The correctability condition implies $\CM   ({\cal R}  \circ {\cal N})=1  $, while the data-processing inequality implies   $\CM   ({\cal R}  \circ {\cal N}) \le \CM   ( {\cal N})$. Hence, $\CM   ({\cal N})\ge 1$. Since the maximum value of $\CM$ is 1, the equality holds. \qed 

\medskip 
It is worth stressing that the converse of Corollary \ref{cor:maxforcorrectable} does not hold in general: the condition  $\CM  ({\cal N}) = 1$ only implies that there exists a pair  of states $\rho$ and $\rho'$ whose distinguishability is preserved by $\cal N$. This condition does not set any requirement on the action of $\N$ on the remaining  states, and therefore  is too weak to guarantee correctability.

Proposition \ref{prop:propertiesmax} and Corollary \ref{cor:maxforcorrectable}   hold generically, with minimal requirements on the probabilistic theory under consideration.  On the other hand, classical and quantum theory  exhibit much stronger properties, including the following property, which   will play a crucial role in the rest of this paper: 
\begin{definition}\label{def:generated}
A state space ${\sf St}  (S)$ is {\em generated by  perfectly distinguishable states}  if  every element $\delta$  of the vector space ${\sf St}_{\mathbb R}  (S)$ can be decomposed as  
\begin{align}
\delta  =   a   \,    \rho    -   b  \,    \rho'  \,,
\end{align}
where $a$ and $b$ are nonnegative coefficients, and $\rho$ and $\rho'$ are perfectly distinguishable states.   
\end{definition}
All state spaces in  classical and quantum theory are generated by perfectly distinguishable states.  Indeed, every  (trace-class)  Hermitian operator $H$ can be decomposed as $  H=  a  \rho  -  b  \rho'$, where  $\rho$ and $\rho'$ are density matrices with orthogonal support, while $a$ and $b$ are non-negative real numbers.  

For compact state spaces generated by distinguishable states, the supremum in Eq. (\ref{sup}) is guaranteed to be a maximum and we have the following: 
\begin{proposition}\label{prop:pdmax}
If the state space ${\sf St}  (A)$ is compact and  generated by perfectly distinguishable states, then one has the equality
\begin{align}\label{distinguishablerhorho'}
\CM   ({\cal N})     =  \max_{\rho \, {\rm p.d.} \, \rho'  }\frac {\|{\cal N}  (\rho)  -  {\cal N} (\rho')\|_1}{2} \,,
\end{align}
where the notation $\rho ~  {\rm p.d.} ~\rho'$ means that the states $\rho$ and $\rho'$ are perfectly distinguishable. Moreover, the maximization can be restricted without loss of generality to pure states, namely
\begin{align}\label{distinguishablepsipsi'}
\CM   ({\cal N})     =  \max_{\psi ~ {\rm p.d.} ~  \psi'}\frac {\|{\cal N}  (\psi)  -  {\cal N} (\psi')\|_1}{  2 }\,,
\end{align}
where the maximization is restricted to perfectly distinguishable pure states $\psi$ and $\psi'$.   
\end{proposition}

The proof is provided in Appendix \ref{app:pdstuff}.

 We conclude this subsection by observing that the decomposition into distinguishable states implies a  continuity property of the maximum causal effect:
  \begin{proposition}\label{prop:continuity_max}   Let  ${\sf St}  (A)$ be a  compact state space   generated by perfectly distinguishable states, and let $\cal  N$ and $\cal N'$ be two processes with input $A$ and output $B$. If  $\cal  N$ and $\cal N'$ are $\epsilon$-close on all possible input states, namely $ \|  {\cal N}  (\rho)   -  {\cal N'}  (\rho)\|_1 \le \epsilon \, , \forall \rho  \in  {\sf St}  (A)$, then  then their maximum causal effects are $\epsilon$-close, namely  $|\CM   ({\cal N})  -  \CM   ({\cal N'})|  \le \epsilon$   
  \end{proposition}
The proof can be found in Appendix~\ref{app:CE_max_continuity}.

\subsection{The minimum causal effect}
While maximum causal effect quantifies the maximum influence that a change of  input   can have on the output of a physical process, for  certain applications it is useful to quantify the minimum influence: 

\begin{definition}[Minimum causal effect]
In a general probabilistic theory, the {\em minimum causal effect} induced by a process ${\cal N}:  {\sf St }  (A)  \to {\sf St}  (B)$ is 
 \begin{align}
 \Cm   ({\cal N})    &:  = \inf_{\rho , \,  \rho'  \in  {\sf St}  (A) ,\,  \rho \not = \rho'}   \,  \frac{  \|  {  \cal N}  (\rho)  -  {\cal N}  (\rho')  \|_1}{  \|  \rho -  \rho'\|_1} \, . \label{Eq:GCE_min}
 \end{align}   
\end{definition}
Later in this subsection, we will see that, under an assumption satisfied by both classical and quantum theory,  the infimum is actually a minimum, and can be attained by minimizing over perfectly distinguishable states.  Intuitively, $\Cm$ quantifies the minimum change in the state of system $B$ due to a change in the state of system $A$.  

Like the maximum causal effect, the minimum causal effect satisfies basic inequalities:

\begin{proposition}\label{prop:propertiesmin}
     For  every pair of systems $A$ and $B$, and every process ${\cal N}:  {\sf St }  (A)  \to {\sf St}  (B)$, the minimum causal effect
      has the  following properties: 
 \begin{enumerate}
 \item {\em Range:}  $0\le  \Cm   ({\cal N}) \le 1$,

 \item {\em Data-processing inequality:}   $  \Cm    ( {\cal  B}  \circ {\cal N} )  \le  \Cm    (  {\cal N} )$ for arbitrary processes  ${\cal N}  : {\sf St}  ( A)  \to  {\sf St}  ( B)$, and ${\cal B}  : {\sf St}  ( B)  \to  {\sf St}  ( B')$, and arbitrary systems $A,B,$ and $B'$. 
  \end{enumerate}
\end{proposition}

The proof is provided in Appendix \ref{app:basicproperties_Min}.  As in the case of the maximum causal effect, an immediate consequence of the data-processing inequality  property is that every correctable process has minimum causal effect equal to 1. 
\begin{corollary}\label{cor:maxforcorrectable2}
Let ${\cal N}  : {\sf St}  ( A)  \to  {\sf St}  ( B)$ be a correctable process.    Then,  $\Cm   ({\cal N})  =1$.
\end{corollary}
{\bf Proof.} The correctability condition implies $\Cm   ({\cal R}  \circ {\cal N})=1  $, while the data-processing inequality implies   $\Cm   ({\cal R}  \circ {\cal N}) \le \Cm   ( {\cal N})$. Hence, $\Cm   ({\cal N})\ge 1$. Since the maximum value of $\Cm$ is 1, the equality holds. \qed 

\smallskip

If state space of the input system is compact and  generated by perfectly distinguishable states, in the sense of  Definition \ref{def:generated}, then  the infimum in Eq. \eqref{Eq:GCE_min} is a minimum and the minimization can be  restricted to perfectly distinguishable states:  

\begin{proposition}\label{prop:pdmin}
If the state  space ${\sf St}  (A)$ is compact and  generated by perfectly distinguishable states, then one has the equality
\begin{align}\label{minimumondistinguishablestates}  
\Cm   ({\cal N})     =  \min_{\rho \, {\rm p.d.} \, \rho'  }\frac {\|{\cal N}  (\rho)  -  {\cal N} (\rho')\|_1}{2} \,,  
\end{align} 
where the notation $\rho ~  {\rm p.d.} ~\rho'$ means that the states $\rho$ and $\rho'$ are perfectly distinguishable.
\end{proposition}

The proof is provided in Appendix~\ref{app:property_min_perfectly_disting}.  Note that, unlike in the case of the maximum causal effect, the minimization in Eq.  (\ref{minimumondistinguishablestates}) cannot, in general,  be restricted to pure states.      For example, consider a classical  process $\N$ transforming a two-bit input system $A$ into a single-bit output system $B$ according to the following rules: 
\begin{itemize}
\item if the input bits are 00,  the output bit is $0$,
\item if the input bits are 11, the output bit $1$,  
\item if the input bits are 01, the output bit is $0$ with probability $1/3$ and $1$ with probability $2/3$,
\item if the input bits are 10, the output bit is $0$ with probability $2/3$ and $1$ with probability $1/3$.
\end{itemize}
For this process, the minimum causal effect is zero,  as one can see by choosing the states $\rho$  and $\rho'$ in Eq. (\ref{minimumondistinguishablestates}) to be a uniform mixture of $00$ and $11$, and a uniform mixture of $01$ and $10$, respectively.  Instead,  the ratio $\|  \N (\rho)  -  \N(\rho')\|_1$ is non-zero for every pair of distinct pure states $\rho$ and $\rho'$.

We conclude the section with  the following continuity property of $\rm{CE}_{\min}$:  

\begin{proposition}[Continuity of the minimum causal effect] \label{prop:continuity_minimum}
    If the state  space ${\sf St}  (A)$ is compact and  generated by perfectly distinguishable states,  then the condition $ \|  {\cal N}  (\rho)   -  {\cal N'}  (\rho)\|_1 \le \epsilon \, , \forall \rho  \in  {\sf St}  (A)$  implies  $|\Cm   ({\cal N})  -  \Cm   ({\cal N'})|  \le \epsilon$. 
\end{proposition}
The proof is provided in Appendix~\ref{app:CE_continuity_min}.

\subsection{Relation to minimum error state discrimination}
We now establish a relation between the minimum causal effect and   the problem of minimum error state discrimination.  Suppose that system $A$ is prepared in one of  two possible states $\rho$ and $\rho'$, with prior probabilities  $p$ and $1-p$.    In this setting, the minimum error probability in distinguishing between $\rho$ and $\rho'$ is given by \cite{chiribella2010probabilistic}  
\begin{align}\label{perr}
p_{\rm err}^{\min}  =  \frac{  1 -  \| p\rho  -(1-p) \,\rho' \|_1}2 \, .
\end{align}
In the special case of quantum theory, this expression coincides with  Helstrom's formula for the minimum error probability \cite{helstrom_quantum_1976}.

Eq. (\ref{perr}) shows that the operational norm  quantifies the deviation of the minimum error probability from the error probability of a random guess $p_{\rm err}^{\rm random}  =  1/2$.      Hence, the ratio  between the  norms $\|  p  \, {\cal N}  (\rho)  -  (1-p)  \,   {\cal N}   (\rho')\|_1$ and $\|  p  \, \rho  -  (1-p)  \,   \rho' \|_1$  is a measure of how much the process $\cal N$ preserves  the distinguishability of two input states provided with a given prior probabilities $p$ and $1-p$, respectively.    In  the worst case over all possible prior probabilities and over all possible states, we have the following quantity: 
\begin{definition}\label{def:dp}
The {\em  minimum distinguishability preservation} of channel $\cal N$ is  
\begin{align}
{\rm DP}_{\min}   ({\cal N}) :  =  \min_{0\le p\le 1}  ~\inf_{\rho, \,  \rho'  \in {\sf St}  (A), \, \rho\not =  \rho'}   \frac{ \|  p  \, {\cal N}  (\rho)  -  (1-p)  \,   {\cal N}   (\rho')\|_1   }{\|   \, p\, \rho  -    \, (1-p)\, \rho'\|_1} \, .       \end{align}
\end{definition}
In general, the data-processing inequality for  the operational norm   \cite{chiribella2010probabilistic} implies the inequality    ${\rm DP}_{\min}   ({\cal N})   \le 1$ for every process $\cal N$.  

Eq. (\ref{perr}) implies the following  observation: 
\begin{proposition}
 For every process $\N  :  {\sf St} (A) \to {\sf St} (B)$, every pair of states $\rho , \rho'  \in  {\sf St}  (A)$, and every probability $p\in  [0,1]$, the minimum error probability in distinguishing between the states $\N  (\rho)$ and $\N  (\rho')$, given with prior probabilities $p$ and $1-p$, respectively, is upper bounded as 
 \begin{align}
 p_{\rm err}^{\min}  \le  \frac{1  -  {\rm DP}_{\min}  \, \|  p\,  \rho  - (1-p)\,  \rho' \|_1}2 \, .
 \end{align}
\end{proposition}

Preservation of distinguishability is an important property of physical processes. In quantum theory, it is  linked to the tasks of error correction and information recovery by the work of Blume-Kohout,  Ng, Poulin, and Viola \cite{Blume-kohut_IPS_2010}, who showed that a quantum process  is perfectly correctable if and only if it  preserves the  trace distance  for all pairs of states and for all values of the prior probabilities.  
In our terminology,  their result can be stated as:  a quantum process $\cal N$ is perfectly correctable if and only if ${\rm DP}_{\min}  ({\cal N})=1$.  
Later in this paper,  we will extend the result  of Ref. \cite{Blume-kohut_IPS_2010} to approximate correctability, showing  a quantitative relation between the approximate correctability of a quantum process and its  minimum distinguishability preservation. 

Crucially,  the minimum distinguishability preservation is lower bounded in terms  of the minimum causal effect: 
\begin{theorem}\label{theo:ddworst}
 If the state space ${\sf St}  (A)$ is  compact and generated by perfectly distinguishable states, then the bound
 \begin{align}
   {\rm DP}_{\min}   ({\cal  N})  \ge  2  \Cm  ({\cal N})  -1 
 \end{align}
 holds for   every process ${\cal N} :  {\sf St}  (A)  \to  {\sf St}  (B)$ and for every system $B$. 
\end{theorem}

The theorem is proven in Appendix~\ref{app:property_min_perfectly_disting}. In particular, the theorem implies that the distinguishability preservation is close to 1 whenever the minimum causal effect is close to 1: $\Cm  ({\cal N})  \ge 1-\epsilon$  implies      ${\rm DP}_{\min}   ({\cal  N}) \ge 1-2 \epsilon$.  Building on this result, we will later show that quantum processes with near-unit values of the minimum causal effect are approximately correctable.

\subsection{Other notions of causal effect}

We conclude this section by briefly mentioning other possible notions of causal effect in general probabilistic theories. 

One possibility is to consider the maximum and minimum of the ratio $\|  {\cal N}  (\rho)-  {\cal N}  (\rho')\|_1  / \|  \rho-  \rho'\|_1 $ over restricted sets of states, such as states satisfying  symmetries \cite{Marvian2013, Marvian2014, Marvian2014_modes, Marvian2016},   or states with  bounded  expectation value of the energy  \cite{winter2017energyconstrained, Shirokov2017, Sharma2018,Becker2021}. 
  
Moreover, instead of taking the maximum or minimum, one could consider taking the expectation value of the ratio $\|  {\cal N}  (\rho)-  {\cal N}  (\rho')\|_1  / \|  \rho-  \rho'\|_1 $ with respect to some prior probability distribution $\pi({\rm d} \rho, {\rm d} \rho')$ over the possible pairs of input states.  
   The resulting quantity, denoted by ${  \rm CE}_{\pi}  ({\cal N})$, is given by    
\begin{align}
\nonumber {  \rm CE}_{\pi}  ({\cal N})    & :  =  \int  \pi({\rm d} \rho, {\rm d} \rho')  \,    \,  \frac{   \|  {\cal N}  (\rho)   -  {\cal N}  (\rho') \|_1  } { \| \rho  -  \rho'\|_1  }   \, .    
 \end{align}
Notice that, by definition, the $\pi$-averaged causal effect is bounded as  
\begin{align}
\Cm  ({\cal N})  \le  {\rm CE}_{\pi}  ({\cal N})  \le  \CM  ({\cal N}) \, .
\end{align}
For quantum theory, a $\pi$-averaged causal effect was defined by Hutter {\em et al.} \cite{hutter2022quantifying}, who considered  the case where $\pi$ is the unitarily invariant  distribution over all pairs of pure orthogonal states.  

Finally, one can consider   weighted generalizations of the maximum and minimum  causal effects, defining
 \begin{align}\label{DPmaxp} \CM (\N  ,  p):=\sup_{\rho,  \, \rho'  \in  {\sf St}  (A), \, \rho\not  = \rho'}\frac{\norm{p\, \N(\rho)-  (1-p)\,  \N  ( \rho')}_1}{\norm{p\,   \rho-(1-p)\rho'}_1}  
\end{align} 
 and 
 \begin{align}\label{DPminp}\Cm(\N  ,  p):=\inf_{\rho,  \, \rho'  \in  {\sf St}  (A), \,  \rho\not  = \rho'}\frac{\norm{p\, \N(\rho)-  (1-p)\,  \N  ( \rho')}_1}{\norm{p\,   \rho-(1-p)\rho'}_1}   \, .
\end{align}  
These quantities  can be regarded as  a measure of how well a given process $\cal N$   preserves the   distinguishability of pairs of states given with prior probabilities $p$ and $1-p$, in the best and worst case scenarios, respectively.

In quantum theory, an analogue of $\CM(\N  ,  p)$ was defined by Hirche {\em et al.}  using a different distance metric known as the hockey-stick divergence \cite{hirche2023quantum}.

\subsection{Including additional input systems into the picture}  

Like the classical ACE,  the notions of causal effect defined in this paper use signaling as a  proxy for causal influence. A downside of the signaling-based approach is that certain processes may not allow for signaling,  even if there exists an actual  causal relation between the input and the output \cite{janzing_2013, Spekkens_cause}. A simple example of this situation arises from the  one-time-pad cryptographic protocol, in which the sender encrypts the value of a bit $a  \in \{0,1\}$ by computing its modulo-2 sum  with another bit $k\in  \{0,1\}$, serving as the key.  
 
The value of the key is chosen uniformly at random, which implies that  the encoded bit $b=  (a +  k ) \mod 2$ is also  uniformly random. Hence, the effective channel from $a$ to $b$ does not permit any signaling.   However, there is a clear causal dependence of $b$ on $a$, since $b$ is generated from $a$ by evaluating the deterministic function $f(a,k)=  (a +  k ) \mod 2$ \cite{schmid2019initial}.         

The one-time-pad  and similar examples indicate that, to capture causal relations present at  a more fundamental level, it is necessary to inspect the underlying physical mechanism that generates variable $B$ from variable $A$, rather than just considering the effective process from $A$ to $B$ ~\cite{barrett_2019,perinotti_2021,perinotti_2023}.   

One way to bridge the gap between signaling-based and interaction-based approaches is to consider physical processes involving not only the variables $A$ and $B$ of interest, but also additional variables that could in principle be accessed by the experimenter.  For example, consider a two-input  process $\N:  {\sf St} (A\otimes K) \to  {\sf St} (B)$, where $K$ is an additional physical system involved in the mechanism that generates $B$ from $A$ (for example, $K$ could be the key in the one-time pad example.)   In this case, the amount of signaling from $A$ to $B$ generally depends on the initial state of system $K$. It is then useful to consider  the amount of signaling from $A$ to $B$ in the best and worst case over all possible states $K$, which give rise to the   following variants of the minimum and maximum causal effects:  
\begin{align*}
{\rm CE}^{A\to B|  K}_{\max|\max}  (\N)  &=  \sup_{\sigma  \in  {\sf St}  (K)}   ~\sup_{\rho,  \rho'  \in  {\sf St} (A) \, , \rho \not  =  \rho'}   \frac{  \|  \N_\sigma (\rho)  -\N_\sigma  (\rho')\|_1}{\|  \rho  - \rho'\|_1}  \\
 {\rm CE}^{A\to B|  K}_{\min|\max}  (\N)    &=  \inf_{\sigma  \in  {\sf St}  (K)}   ~\sup_{\rho,  \rho'  \in  {\sf St} (A) \, , \rho \not  =  \rho'}   \frac{  \|  \N_\sigma (\rho)  -\N_\sigma  (\rho')\|_1}{\|  \rho  - \rho'\|_1}  \\
  {\rm CE}^{A\to B|  K}_{\max|\min}  (\N)    &=  \sup_{\sigma  \in  {\sf St}  (K)}  ~ \inf_{\rho,  \rho'  \in  {\sf St} (A) \, , \rho \not  =  \rho'}   \frac{  \|  \N_\sigma (\rho)  -\N_\sigma  (\rho')\|_1}{\|  \rho  - \rho'\|_1}  \\
   {\rm CE}^{A\to B|  K}_{\min|\min}  (\N)    &=  \inf_{\sigma  \in  {\sf St}  (K)}  ~ \inf_{\rho,  \rho'  \in  {\sf St} (A) \, , \rho \not  =  \rho'}   \frac{  \|  \N_\sigma (\rho)  -\N_\sigma  (\rho')\|_1}{\|  \rho  - \rho'\|_1} \, ,
\end{align*}
where $\N_\sigma  :  {\sf St}  (A) \to {\sf St}  (B)$ is the effective channel defined by 
\begin{align}
\N_\sigma  (\rho) :  =  \N (\rho \otimes \sigma) \qquad  \forall \rho  \in  {\sf St}  (A) \, ,\forall \sigma \in  {\sf St}  (K) \,.
\end{align}

In the one-time-pad example,  systems $A, B, $ and $K$ are classical and correspond to the input bit, the output bit, and the key, respectively. The overall channel $\N$ is described by the conditional probability distribution  $  p(  b|  a,k)   =  \delta_{b,  (a+  k)\mod 2}$, while the effective channel $\N_\sigma$  is described by the conditional probability distribution $p_\sigma (b|a)  =  \sum_k  \,  p(  b|  a,k)  \, \sigma(k)$, where $\sigma(k)$ is the probability distribution for the value of the key.     In the worst case over all probability distributions $\sigma(k)$, no causal effect can be detected through signaling from $A$ to $B$:   ${\rm CE}^{A\to B|  K}_{\max|\min}  (\N)  =  0$.  In contrast,    the best case over all probability distributions $\sigma(k)$ reveals the highest possible value of the causal effect: ${\rm CE}^{A\to B|  K}_{\min|\max}  (\N) =1$.  

The above notions can be extended to multipartite processes, thereby providing a way to characterize how a group of variables can control the visibility of causal influences between two or more systems.

\section{Causal effects in quantum theory}\label{sec:quantum}

From this section on, we will focus our attention to quantum theory, highlighting the  applications of the maximum and minimum causal effects to  quantum information tasks, such as quantum communication and  approximate information recovery.   All the general  properties established in Propositions \ref{prop:propertiesmax}  and \ref{prop:propertiesmin} will apply. Moreover,  since all state spaces in quantum theory are generated by perfectly distinguishable states,  Propositions \ref{prop:pdmax} and \ref{prop:pdmin} apply too, and the maximum and minimum causal effects can be computed by restricting the maximization and minimization to perfectly distinguishable states.

\subsection{Maximum and minimum quantum causal effects}

For finite-dimensional quantum systems,  the maximum causal effect  of a channel $\cal N:  {\sf St} (A) \to {\sf St} (B) $ is given by 
\begin{align}
 \nonumber  \CM   ({\cal N})   &=  \max_{\rho, \, \rho'  \in {\sf St (A)} \, ,\rho \not =  \rho'}  \frac{\norm{\mathcal{N}(\rho)-\mathcal{N}(\rho')}_1  }{\|  \rho -  \rho'  \|_1}  \\
 &\equiv  \max_{|\psi\rangle,{|\psi_{\perp}\rangle   \in  {\cal H}_A , \,   \langle \psi|  \psi_\perp\rangle  = 0 }  }\frac{\norm{\mathcal{N}(|\psi\rangle\langle \psi|)-\mathcal{N}(|\psi_\perp\rangle \langle \psi_\perp|)}_1  }{2} \,, \label{maxquantum}
  \end{align}
  where $\|  \cdot \|_1  $ is  now the trace distance, ${\cal H}_A$ is the Hilbert space of  system $A$, and the second equality follows from    Proposition \ref{prop:pdmax}.

Similarly,  the minimum quantum causal effect is
\begin{align}
 \nonumber  \Cm   ({\cal N})   &=  \min_{\rho , \, \rho' \in {\sf St}  (A), \, \rho \not =  \rho'}  \frac{\norm{\mathcal{N}(\rho)-\mathcal{N}(\rho')}_1  }{\|  \rho -  \rho'  \|_1}  \\
 &\equiv  \min_{\rho , \, \rho_\perp \in {\sf St}  (A), \,    \Tr[\rho\rho_\perp]  = 0 }  \frac{\norm{\mathcal{N}(\rho)-\mathcal{N}(\rho_\perp)}_1  }{2} \,,  \label{minquantum}
  \end{align}
  where the second equality follows from    Proposition \ref{prop:pdmin}.  When the input system is two-dimensional, the orthogonality condition $\Tr[\rho \rho']=  0$ implies that the density matrices $\rho$ and $\rho'$ are rank-one, and therefore  the minimization in Eq. (\ref{minquantum}) can be restricted to orthogonal pure states.    In general, however, a minimization over mixed states is needed.  

An interesting special case arises for quantum channels of the classical-to-classical form 
\begin{align}\label{classicalN}
\N  ( \rho )   =   \sum_{a=1}^{d_A} \sum_{b= 1}^{d_B}    q(  b|  a)  \,  \langle a| \rho  |a\rangle ~  |b  \rangle \langle b| \qquad\forall \rho  \in  {\sf St}  (A)   \, ,
\end{align}
where $\{|a\rangle\}_{  a=1}^{D_A}$ and $\{|b\rangle\}_{  b=1}^{d_B}$ are two fixed orthonormal bases for systems $A$ and $B$, respectively, and $q(b|a)$ is a conditional probability distribution. 
In this case, the maximum causal effect coincides with the ACE of the associated classical channel:
\begin{proposition}\label{prop:quantumeffectclassicalchan}
 For every quantum channel $\N$ of the classical-to-classical form (\ref{classicalN}), one has 
 \begin{align}\label{maxclass}
 \CM  (\N)   =  {\sf ACE}_{X\to Y} \,,
 \end{align}
 where $X$ and $Y$ are  random variables with values in the sets  $\{1,\dots, d_A  \}$ and  $\{1,\dots, d_B  \}$, respectively, and with do-probability $p(b|  {\rm do}  (a))  :  = q  (b|a)$. 
  \end{proposition}
The proof is provided in Appendix \ref{app:quantumeffectclassicalchan}.

Proposition \ref{prop:quantumeffectclassicalchan}  provides further evidence that the maximum causal effect is the correct quantum version of the ACE. On the other hand,  every channel of the classical form (\ref{classicalchannel}) has zero value of the minimum quantum causal effect.  More generally,  the minimum quantum causal effect is zero for every channel that decoheres the input  on a given orthonormal basis, that is, any channel $\N$ of the classical-to-quantum form 
\begin{align}\label{decoherence}
\N (\rho)   =  \sum_{a=1}^{d_A}   \,  \langle a|  \rho  |a\rangle\,  \,\rho_a
\end{align}
where $(\rho_a)_{a=1}^{d_A}$ is a set of density matrices for the output system. 

\begin{proposition}
For every channel $\N$ of the classical-to-quantum form (\ref{decoherence}), one has  \begin{align}\label{minclass}
 \Cm(\N)  =  0 \,.        
 \end{align}
\end{proposition}
{\bf Proof.}   Let $\{  |e_k\rangle\}_{k=1}^{d_A}$ be the Fourier basis for system $A$, defined as $|e_k\rangle:  =  \sum_{x=1}^{d_A}   \exp  [  2\pi i kx/d_A] \,  |x\rangle/\sqrt {d_A}$ for every $k$.   Setting $  \rho :  =  |e_1\rangle \langle e_1|$ and $\rho'=  |e_2\rangle\langle e_2|$ we then have $\|  \rho-  \rho'\|_1  =  2$ and $\N (\rho)  = \N(\rho')  = \sum_a  \rho_a/{d_A}$.  Hence,  $\|\N (\rho)  - \N(\rho')\|_1  = 0 $. By definition of minimum causal effect, this proves Eq. (\ref{minclass}). 
\qed
  \medskip 
  
This result shows that the minimum quantum causal effect  witnesses a special type of causal influence which is absent in all channels of the classical-to-quantum form, and therefore also in all channels of the classical-to-classical form.  Later in this section, we will see that a given quantum channel has high value of the minimum quantum causal effect  if and only if it is approximately correctable, thereby linking the minimum quantum causal effect to the preservation of quantum information.

\subsection{Application to classical communication through quantum noisy channels}

The maximum causal effect provides a lower bound on the ability of quantum communication channels to transmit classical information.   For a quantum channel ${\cal N}  :{\sf St}   (A) \to {\sf St}   (B)$,  the classical capacity $C(\N)$ is the maximum rate at which information can be reliably transmitted per channel use, in the limit of asymptotically many uses   \cite{Holevo73,Schumacher1997,Holevo1998}.    
  The classical capacity can be lower bounded in terms of the maximum causal effect, as in the following:   

\begin{proposition}\label{lem:classical_capacity}
For every quantum channel ${\cal N}$, the classical capacity $C(\mathcal{N})$ satisfies the lower bound 
\begin{align}\label{capacitybound1}
C(\mathcal{N}) \ge 1 - h_2\left(\frac{1- \CM  (\cal N)  }2 \right)  \,, 
\end{align}
where  $h_2(x):=-x\log_2(x)-(1-x)\log_2(1-x)$ is the binary  entropy.    
\end{proposition}
The lower bound (\ref{capacitybound1}), proven in Appendix~\ref{app:classical_capacity},  is obtained by considering communication protocols  that use the channel a single time and encode information into  two orthogonal states.  For this reason, the bound  (\ref{capacitybound1}) is generally not tight: in general, it is known that {\em (i)} achieving the capacity requires multiple uses of the channel \cite{Holevo73,Schumacher1997,Holevo1998},  {\em (ii)} orthogonal states may not be optimal for achieving the capacity \cite{Fuchs1997}, and {\em (iii)}  for quantum systems of dimension $d\ge 2$,  achieving the capacity generally requires more than two input states.    Despite the lack of tightness,  the bound (\ref{capacitybound1}) is useful because it provides a way to detect a non-zero classical capacity from a quantity that can be estimated from experiments.  

In addition, the maximum causal effect can be used to identify channels that are poor for the transmission of classical information.  From Proposition \ref{prop:propertiesmax}, we know that    $\CM   ({\cal N})=  0$ if and only if $\cal N$ is a discard-and-reprepare channel, that is, a channel that produces a fixed output state, independently of the input state.   In quantum theory,    this statement is valid also in the approximate case:    
\begin{proposition}[Faithfulness]\label{Lem:faithful}
    For every quantum channel ${\cal N} :  {\sf St}  (A)  \to  {\sf St}  (B)$,   there exists a channel ${\cal N}_0 :  {\sf St}  (A)  \to  {\sf St}  (B)$,  of the  discard-and-reprepare form $\mathcal{N}_0(\rho)  =  \sigma_0  {\sf Tr} $ for some state $\sigma_0 \in  {\sf St} (B)$, such that
    \begin{align}
     \norm{\mathcal{N}-\mathcal{N}_0}_\diamond \le  4  d_A  \,    \CM  ({\cal N}) \, ,  
    \end{align}
    where $d_A$ is the dimension of system $A$'s Hilbert space,  and $\|  \cdot \|_{\diamond}$ is the diamond norm, defined as $\|  {\cal M} \|_{\diamond}   :  =  \sup_E  \max_{\rho  \in {\sf St}  (  A \otimes  E)}    \|  ({\cal M} \otimes {\cal I}_E)   (\rho) \|_1$ for a Hermitian-preserving map ${\cal M}:  L(\H _A) \to  L(\H _B)$.    
\end{proposition}

The proof is provided in Appendix~\ref{App:faithful_proof}. 
The above proposition implies that a channel with a small  $\CM$  (compared to the input dimension) is necessarily close to a discard-and-reprepare channel.  The converse is also true due to the continuity of $\CM$  (cf. Proposition \ref{prop:continuity_max}).

\subsection{Application to approximate recoverability of quantum information}

We now show that the minimum causal effect plays an important role the  approximate recovery of quantum information \cite{fawzi2015quantum,Junge_2018}.   The connection to recovery proceeds through the connection to minimum error state discrimination established in Theorem \ref{theo:ddworst}. There, we showed that the minimum causal effect provides a lower bound to to the minimum distinguishability preservation ${\rm DP}_{\min}  (\N)$,  a measure of the channel's ability to preserve the distinguishability of arbitrary pairs of states given with arbitrary prior probabilities.  
Now, we establish a relation between  minimum distinguishability preservation and the recovery of quantum information:

\begin{theorem}\label{theo:correctability}
    For every quantum channel $\N  : {\sf St}  (A)  \to  {\sf St}  (B)$,  there exists a recovery channel  $\cal R: {\sf St}  (B)  \to  {\sf St}  (A)$ such that the upper bound
    \begin{align}
  \nonumber  & \|  {\cal R}  \circ  \N  (\rho)    - \rho\,  \|_1 \\
   &  \qquad \le  4  \sqrt{\ln 2} \,       \sqrt{ \frac{1-{\rm DP}_{\min}  (\N)}2}    \,  \sqrt{    \frac {d_A}2  + \ln (d_A+1)  }  \label{prima}
   \end{align}
    holds for every state $\rho  \in {\sf St} (A)$.  
\end{theorem}

The proof is provided in Appendix \ref{app:correctability}.  Theorem \ref{theo:correctability} shows that approximately preserving the distinguishability of {\em pairs} of states, given with arbitrary prior probabilities, is a sufficient condition for approximate correctability: 
\begin{corollary}\label{cor:approximatecor}
If channel $\N$ approximately preserves the distinguishability of quantum states with arbitrary prior probabilities, namely  ${\sf DP}_{\min}  (\N)  \ge 1-\epsilon$, then   it can be approximately corrected, up to an error  $2 \sqrt{3\ln 2  \,  \epsilon \, d_A}$.      \end{corollary}  
{\bf Proof.} Immediate from Theorem \ref{theo:correctability} and from the bound $\ln (d_A+1) \le d_A$, valid for every $d_A \ge 2$.  \qed 

\medskip 

Note that a converse result also holds:  
\begin{corollary}\label{cor:approximatecor2}
  If  channel  $\N$ is approximately correctable ({\em i.e.,} if there exists a channel ${\cal R}$ such that    $\|  {\cal R  \circ \N}  (\rho)  -  \rho\|_1  \le \epsilon$ for every state $\rho  \in {\sf St}  (A)$), then  it approximately preserves the distinguishability of states; quantitatively,  ${\sf DP}_{\min}  (\N)  \ge 1-2\epsilon$.   
\end{corollary}
{\bf Proof.} The continuity of the minimum causal effect (Proposition \ref{prop:continuity_minimum}), applied to channel ${\cal R} \circ \N$ implies $\Cm( {\cal R} \circ \N)  \ge 1-\epsilon$.  Then, the data-processing inequality of the minimum causal effect (Proposition \ref{prop:propertiesmin})   implies $\Cm(  \N)  \ge 1-\epsilon$. Finally,  Theorem \ref{theo:ddworst} implies ${\rm DP}_{\min}  (\N)   \ge  1- 2\epsilon$. \qed 

\medskip

Corollaries \ref{cor:approximatecor}  and \ref{cor:approximatecor2}  provide the robust version of a result by Blume-Kohout, Ng, Poulin, and Viola \cite{Blume-kohut_IPS_2010}, who showed that a quantum channel can be perfectly corrected if and only if it perfectly preserves the distinguishability of arbitrary pairs of states, given with arbitrary prior probabities.

The above results can be equivalently seen as a relation between causal effects and approximate correctability:  
\begin{corollary}
 If channel $\N$ has a high value of the minimum causal effect, namely  $\Cm  (\N)  \ge 1-\epsilon$, then   it can be approximately corrected, up to an error  $2 \sqrt{3\ln 2  \,  \epsilon \, d_A}$.    Conversely, if channel $\N$ can be approximately corrected, up to error $\epsilon$, then it has a high value of the minimum causal effect; quantitatively, $\Cm (\N)  \ge  1-\epsilon$.  
\end{corollary}
\begin{proof}  
The direct part follows from Theorems \ref{theo:correctability}
 and \ref{theo:ddworst} and using $\ln(d_A+1)<d_A$.  The converse part was already proven in  the proof of Corollary \ref{cor:approximatecor2}. \end{proof}

 \medskip  

The above results can be summarized in a single qualitative statement:  
\begin{corollary}
 For a quantum channel $\N$, the following are equivalent: 
 \begin{enumerate}
 \item $\N$ has high value of the minimum causal effect,
 \item $\N$ approximately preserves the distinguishability of quantum states,
 \item $\N$ is approximately correctable. 
 \end{enumerate}
 \end{corollary}

 We conclude this subsection by observing that all the above results can be extended to the case in which  channel $\N$ acts locally  on entangled input states involving an additional quantum system $E$. For example, Theorem \ref{theo:correctability}  becomes:  
 \begin{corollary} \label{cor:correctability}
    For every quantum channel $\N  : {\sf St}  (A)  \to  {\sf St}  (B)$,  there exists a recovery channel  $\cal R$ such that the upper bound
    \begin{align}
 \nonumber    & \|  {\cal R}  \circ  \N      - {\cal I}_A   \|_{\diamond} \\
   &  \qquad \le  8  \sqrt{\ln 2} \,       \sqrt{ \frac{1-{\rm DP}_{\min}  (\N )}2}    \,d_A\,   \sqrt{    \frac {d_A}2 +  \ln (d_A+1) } \label{cor:prima}
    \end{align}  
    holds for every state $\rho  \in  {\sf St}  (A\otimes E)$ and for every quantum system $E$. 
 \end{corollary}

{\bf Proof.}  Immediate from Theorem \ref{theo:correctability}
and from Lemmas \ref{lem:norm1} and \ref{lem:diamondtrace} in Appendix \ref{App:faithful_proof}. \qed

\subsection{Monogamy of the minimum  quantum causal effects}  

Theorem \ref{theo:correctability} implies  a fundamental property of quantum causal relations, which is analogous to the monogamy of quantum entanglement \cite{coffman2000distributed,osborne2006general}:   if the minimum  causal effect of a variable $A$ on a variable $B$ is high, then variable $A$ must have a low causal effect on any other variable $B'$ that is spacelike separated with $B$.

The monogamy of quantum causal relations follows from a fundamental duality in quantum theory,  namely the duality between a quantum channel and its {\em complementary channel}~\cite{kretschmann2008information,kretschmann_continuity_2007,Hayden_2020}. The complementary channel describes the flow of information to the environment, and is defined in terms of a unitary realization of the channel $\N$: 
\begin{align}
\N  (\rho)    = \Tr_{A'} [   U  \, (\rho  \otimes  |\psi_0\rangle \langle \psi_0| )\,  U^\dag] \qquad \forall \rho  \in  {\sf St}  (A)  \, ,     
\end{align}
where $A'$ is an environment,    $|\psi_0\rangle  \in  {\cal H }_{A'}$ is a pure  initial state of the environment,  $U: {\H _A} \otimes {\H _{A'}}  \to {\H _{B'}} \otimes {\H _{B'}}$ is a unitary operator, and $B'$ is the environment for system $B$.   The complementary channel ${\cal N}^{\rm c}:  {\sf St}  (A)  \to {\sf St} (B') $ is  defined as 
\begin{align}
\N^{\rm c}   (\rho):  =     \Tr_{B} \Big[   U  \, (\rho  \otimes  |\psi_0\rangle \langle \psi_0| )\,  U^\dag\Big] \qquad \forall \rho  \in  {\sf St}  (A)  \, .    \end{align}

A fundamental result in quantum information  is that a  channel $\N$ is approximately correctable if and only if its complementary channel $\N^{\rm c}$ is approximately a discard-and-reprepare channel  ~\cite{kretschmann2008information,Hayden_2020}.   This result implies a duality between the maximum and minimum causal effects:  
\begin{proposition}[Duality] \label{prop:max_min_duality}
For every quantum channel   $\N  :    {\sf St}  (A) \to  {\sf St}  (B)$, the following implications hold:
\begin{enumerate}
\item  $\CM  (\N)  \le \epsilon$ implies $\Cm   (\N^{\rm c}) \ge 1-  4 \sqrt{ d_A  \epsilon}$.
\item $\Cm  (\N)\ge 1-\epsilon$ implies $\CM  (\N^{\rm c})\le 2\sqrt{f(\epsilon, d_A)}$, with $f(\epsilon, d_A)=8d_A\sqrt{\epsilon \ln 2}  \sqrt{d_A/2 + \ln(d_A + 1)}$.
\end{enumerate}
\end{proposition}

  See Appendix~\ref{App:max_min_duality} for the proof.   
  
  Proposition \ref{prop:max_min_duality}  implies a monogamy  
   relation for causal effects:  if a quantum system $A$ has has a large value of the minimum causal effect on another system $B$,  then it has a small value of the maximum causal effect on any other system $B'$ that is spacelike separated with $B$:  
   \begin{corollary}[Monogamy of minimum quantum causal effects]\label{cor:monogamy}
    Let $\N:  {\sf St}  (A) \to  {\sf St}  (B\otimes B')$ be a quantum channel, and let $\N_B  : =  \Tr_{B'} \circ \N$   ($\N_{B'}  : =  \Tr_{B} \circ \N$) be the      quantum channel describing the transformation from system $A$ to  system $B$  (from system $A$ to system $B'$).  If $\Cm  (\N_B)$ is high, then $\CM  (\N_{B'})$ is necessarily low; quantitatively, $\Cm (\N_B) \ge 1-\epsilon$ implies $\CM  (\N_{B'})  \le 2\sqrt{f(\epsilon, d_A)}$, with $f(\epsilon, d_A)=8d_A\sqrt{\epsilon \ln 2}  \sqrt{d_A/2 + \ln(d_A + 1)}$.    
   \end{corollary}
{\bf Proof.}  Immediate from the fact that a unitary realization of channel $\N$ with environment $E'$  is automatically a unitary realization of channel $\N_B$ with environment $B'\otimes E'$.  Hence, $\N_{B'}  =   \Tr_{E'} \circ \N_B^{\rm c}$,  and $\CM (\N_{B'})  \le \CM (\N^{\rm c})$ by the data-processing inequality of the maximum causal effect. The result then follows from Proposition \ref{prop:max_min_duality}. \qed  

\medskip

   The monogamy of causal effects highlights an important difference between quantum and classical causal relations.  Classically,    a  physical  system $A$ can have arbitrarily high causal effects on two spacelike systems $B$ and $B'$.  For example, consider the classical process with conditional probability distribution $p(b,b'|a)=  \delta_{b,b'}\delta_{b', f(a)}$, for some invertible function $f:  A \to B$: this process   has minimum causal effect equal to 1 both from $A$ to $B$ and from $A$ to $B'$.  In stark contrast,  Corollary \ref{cor:monogamy}  dictates  that    a high causal effect of a quantum system $A$ on system $B$ implies a low  causal effect on any other system $B'$ that is spacelike separated with $B$.    This phenomenon is 
  closely related to the no-cloning \cite{wootters1982single,dieks1982communication} and to no-broadcasting \cite{barnum1996noncommuting,barnum2007generalized} theorems, and the information-disturbance tradeoff \cite{fuchs1996quantum,kretschmann2008information}, which prevent a leakage of quantum information from system $A$ to system $B'$.

\section{Quantum causal effects in two examples}\label{sec:applications}

Here we apply the notions of maximum and minimum causal effect   to the analysis of two paradigmatic examples of quantum causal relations.  In the first example,  we quantify the amount of direct cause relation in a setup exhibiting a quantum superposition of direct cause and common cause structures.    In the second example, we examine the increase of the causal effect due to the quantum interference of multiple processes, and provide an upper bound on the amount of increase in terms of the amount of coherence among different paths.  

\subsection{Estimating the amount of direct cause in a  coherent superposition of a common and direct cause}

Here we analyze a setup,  introduced by MacLean {\em et al.}~\cite{MacLean_2017}, that exhibits a coherent quantum superposition of common cause and cause effect relations.  The setup is based on the quantum circuit illustrated in Fig.~\ref{fig:process_direct_common_cause}(a).    The circuit involves three quantum systems $A $, $B$, and $\Lambda$,  associated to Hilbert spaces of the same dimension $d$, namely $\H_A  \simeq \H_B\simeq \H_\Lambda  \simeq   {\mathbb C}^d$.  Here, systems $A$ and $B$ are  assumed to be accessible to an experimenter, while system $\Lambda$ is unaccessible. 

Systems $A$ and $\Lambda$ start off in a joint state $\sigma_{A\Lambda}$. Then, they are transformed into systems $B$ and $\Lambda$ through  the  partial swap  evolution  \begin{align}
U_{\theta} = \cos{\theta}\,   I+i\sin{\theta} \,   {\tt SWAP}  \label{Eq:partial_swap} , 
\end{align}
where  ${\tt SWAP}$ is the SWAP operator, uniquely defined by the  relation ${\tt SWAP} (|\psi\rangle \otimes |\phi\rangle)  := |\phi\rangle \otimes |\psi\rangle $ for every pair of vectors $|\phi\rangle, |\psi\rangle \in  {\mathbb C}^d$, and $\theta \in [0,2\pi)$ is a parameter of the evolution.    Finally, system  $\Lambda$ is discarded.    

MacLean {\em et al.} considered the correlations between measurements performed on system $A$ before the action of the gate $U_{\theta}$ and  measurements performed on system $B$ after the action of the gate $U_{\theta}$, as in Fig.\ref{fig:process_direct_common_cause}(b).     In this setting, the correlations  exhibit a combination of common cause and direct cause relations, depending on the value of the parameter $\theta$. For example, $\theta =  0 $ implies that the measurements are performed in parallel on the bipartite state $\sigma_{AB}$, and therefore are of the common-cause type (the common cause being the preparation of the state $\sigma_{AB}$.)  On the other hand, $\theta=  \pi/2$ implies that the two measurements on $A$ and $B$ are implemented sequentially one after another, and therefore are of the direct-cause type.  For intermediate values of $\theta$, the circuit exhibits a coherent quantum superposition of the two configurations corresponding to common cause and direct cause, respectively.     

Here we quantify the amount of direct cause relation, by  analyzing the  scenario in which the experimenter resets  system $A$ to a quantum state $\rho$,  as illustrated in Fig. \ref{fig:process_direct_common_cause}(c).     Here, the experimenter's intervention decorrelates systems $A$ and $\Lambda$, leaving them in the product state $\rho\otimes \sigma_\Lambda$, with  $\sigma_{\Lambda}:  =  \Tr_A  [  \sigma_{A \Lambda}]$. 
The effective process describing  the dependence of  the state of system $B$  on the state $\rho$ is 
\begin{align}
    \mathcal{N}  (\rho){:=}\Tr_{{\Lambda
    }}  \left[ U_{\theta}\bigg(\rho \otimes \sigma_{\Lambda}\bigg){U_{\theta}}^\dagger\right]   \qquad \forall \rho  \in  {\sf St}  (A) \, ,\label{Eq:channel_partial_swap}
\end{align}

For simplicity, in the following,  we will consider the case where the state  $\sigma_{\Lambda}$ is a convex combination of a maximally mixed state and a pure state, namely 
\begin{align}\label{sigmalambda1}
\sigma_{\Lambda}=(1-p)  \frac I d  +  p  \, \proj{\phi} \, ,  
\end{align}
for some probability $p\in  [0,1]$ and some unit vector $|\phi\rangle \in  {\mathbb C}^d.$   Note that this assumption is not restrictive in the two-dimensional case $d=2$, where every density matrix is of the form (\ref{sigmalambda1}).

\begin{figure}[!h]
    \centering
    \includegraphics[width=\columnwidth]{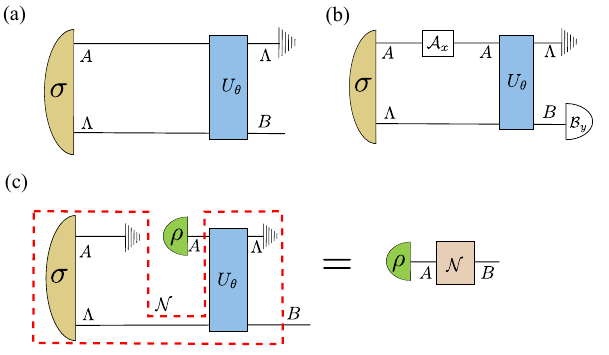}
    \caption{ {\bf Quantum setup exhibiting a coherent superposition of common cause and direct cause configurations.}     Figure \ref{fig:process_direct_common_cause}(a)  shows the original setup, consisting of the preparation of two quantum systems $A$ and $\Lambda$ in a joint state $\rho_{A\Lambda}$.  The two systems then undergo the partial swap evolution $U_\theta  = \cos \theta \,  I+  i  \, \sin \theta \,  {\tt SWAP}$, which transforms them into systems $B$ and  $\Lambda$ (here, all systems $A,B,$ and $\Lambda$ have the same Hilbert space dimension.)    Figure \ref{fig:process_direct_common_cause}(b) shows the scenario where the experimenter performs measurements on system $A$ and $B$, obtaining outcomes $x$ and $y$, respectively. The correlations arising in this scenario have been interpreted as a coherent superposition of common cause and cause-effect correlations \cite{MacLean_2017}.  Figure \ref{fig:process_direct_common_cause}(c)  shows the scenario where the experimenter probes the causal relation between system $A$ and system $B$ by resetting the state of system $A$ to a given state $\rho$.  The dependence of system $B$'s state on the state $\rho$  is characterized by the quantum process $\N$ in Eq. \eqref{Eq:channel_partial_swap}.     In this scenario, the maximum causal effect of process $\N$ provides a measure of the amount of direct cause-effect relation present between systems $A$ and $B$.      }
    \label{fig:process_direct_common_cause}
\end{figure}

The maximum  causal effect of channel $\N$ provides a measure of the amount of direct cause relation between $A$ and $B$. 
For states of the form (\ref{sigmalambda1}),  the maximum causal effect of channel has the following expression:

\begin{proposition} \label{Lem:partial_swap}
    For any state  of the form (\ref{sigmalambda1}),  the   maximum   causal effect of the channel $\N$ in Eq. (\ref{Eq:channel_partial_swap}) is
 
    \begin{align}\label{CMmaclean}
    {\CM} (\mathcal{N}) {=} \max_{F\in [0,1]}   \frac{\sqrt{a^2  +  F b^2 }  +  \sqrt{  a^2  +  F  b^2  +  2|ab| \sqrt{F(1-F)}}}{2}  \,,
    \end{align}
    with $a=  \sin^2 \theta$ and $b =  p   \sin \theta \cos \theta$. 
    
\end{proposition}
The proof is provided in Appendix~\ref{app:partial_swap}.   Note that the maximum causal effect is nonzero whenever $\sin \theta $ is nonzero: as soon as the partial swap (\ref{Eq:partial_swap}) is not equal to the identity operator, there is a transfer of information from system $A$ to system $B$, which witnesses the presence of a causal influence.  One can further observe that this causal influence is of a non-classical type,  as witnessed by a positive value of the minimum quantum causal effect:   for  the channel $\N$ in Eq. (\ref{Eq:channel_partial_swap}),
 the minimum quantum causal effect is  lower bounded as 
\begin{align}\label{minmaclean}
 \Cm  (\N)   \ge  \sin^2 \theta \, ,   \end{align}
for every state $\sigma_\Lambda$, as shown in Appendix~\ref{app:partial_swap}.  Interestingly, the minimum quantum causal effect is positive whenever the maximum causal effect is positive:  indeed, $\CM  (\N)   >0$ implies $\sin \theta \not  =  0,$ which in turn implies $\Cm (\N)  >   0$.   In other words,  causal influence in this setup is  necessarily  different from  the causal influence that arises from measuring the input system in a given basis.

\subsection{Maximum causal effect of  a coherent superpositions of quantum processes}

In general, a quantum particle can travel simultaneously on multiple alternative paths, undergoing to some noisy processes along each path  \cite{oi2003interference,aaberg2004operations,Chiribella_2019,Abbott2020communication,witnessing_latent_time}.  For example, the particle could be a single photon traveling along two arms of an interferometer,  and its polarization could undergo independent  random rotations on the two paths \cite{rubino2021experimental}. 

In this scenario,  quantum interference between noisy processes can sometime reduce the overall amount of noise affecting the particle  \cite{gisin2005error} and in some extreme cases can enable the transmission of information even when the processes on the two paths completely block information when considered in isolation \cite{Abbott2020communication,Chiribella_2019,witnessing_latent_time}.  From the point of view of causal inference, these findings show that  quantum interference between two processes with no causal effect can give rise to an overall process with positive causal effect.     In the following, we will provide a quantitative analysis of this phenomenon.

 We consider the scenario where one particle with can travel on $k$ alternative paths, corresponding to basis states $\{  |j\rangle\}_{j=1}^k$ in the Hilbert space associated to the particle's path.   For simplicity, we assume that no noise takes place on the path degree of freedom. 
 
 When the particle travels on the a well-defined path, its internal degrees of freedom ({\em e.g.}, the polarization of a photon)  undergo a noisy process. For simplicity, we assume that the same process acts on each path, and we denote the corresponding quantum channel  by $\mathcal C$.  When the  particle travels in a superposition of paths,   the joint evolution of internal degrees of freedom, labeled by $A_{\rm int}$,  and  the path degrees of freedom, labeled by $A_{\rm path}$, is described  by a quantum channel ${\cal S} :  {\sf St} (A_{\rm int}\otimes A_{\rm path}) \to  {\sf St} (A_{\rm int}\otimes A_{\rm path})$, whose general expression is provided in Ref. \cite{Chiribella_2019}.  
   In the case under consideration,  the action of the channel $\cal S$ on a  product state $\rho \otimes \sigma$ is 
    \begin{align}
 {\cal S}  (\rho  \otimes \sigma)  = {\cal C}   (\rho) \otimes \sigma_{\rm diag}    +   F \rho  F^\dag   \otimes  \sigma_{\rm off-diag}\, ,         
 \end{align}   
 with 
 \begin{align}
\nonumber  \sigma_{\rm diag}    & :  =  \sum_{j=1}^k  \langle j| \sigma |j\rangle \,  |j\rangle \langle j| \\
  \sigma_{\rm off-diag}  & :  =  \sum_{j, l:  \,  j\not =  l}  \langle j| \sigma |l\rangle \,  |j\rangle \langle l|
 \end{align}
 and 
 \begin{align}
 F  :  =  \sum_i  \overline \gamma_i  C_i  \, ,   \end{align}
 where $(C_i)_i$ are Kraus operators for channel $\cal C$ and $(\gamma_i)_i$ are complex amplitudes, called {\em vacuum amplitudes} and satisfying the normalization condition $\sum_{i}  |\gamma_i|^2  = 1$ \cite{Chiribella_2019}.    The operator $F$ is called the {\em interference operator},  and characterizes the behavior of the noisy  process when the system travels on a coherent superposition of paths.    

In the following, we will take the initial  state of the path $\sigma$ to be fixed, and we will analyze the causal influence of the state $\rho$    on the overall output state of the particle. Specifically, we will evaluate the maximum causal effect  of the effective channel $\N_\sigma:  {\sf St}  (A_{\rm int}) \to {\sf St}   (A_{\rm int} \otimes A_{\rm path})$ defined by 
\begin{align}
\nonumber \N_\sigma  (\rho)  &:  =  {\cal S}  (\rho \otimes \sigma) \\
& = {\cal C}   (\rho) \otimes \sigma_{\rm diag}    +   F \rho  F^\dag   \otimes  \sigma_{\rm off-diag}  \, .  \label{Nsigma}
\end{align}

Our main result is the following: 
\begin{proposition}\label{prop:superposition}
For every quantum channel $\cal C$,  the maximum causal effect  of the channel $\N_\sigma$ is upper bounded as 
\begin{align}\label{CMbound}
\CM  (  \N_\sigma)  \le     \CM  ({\cal C})   +  \frac{\| F^\dag F \|_{\rm Ky Fan2} \, \| \sigma_{\rm off-diag}\|_1}2 \,,  
\end{align}
where $\| F^\dag F \|_{\rm Ky Fan2}$  is the Ky Fan 2-norm, equal to the sum of the first two largest eigenvalues of the operator $F^\dag F$.   Eq. (\ref{CMbound}) holds with the equality sign when  $\CM  ({\cal C}) =  0$. 
\end{proposition}
The proof is provided in Appendix \ref{app:superposition}.  Eq. (\ref{CMbound}) shows that the interference between multiple paths  can increase the causal effect of a quantum channel $\cal C$ by an amount that is limited  by the amount of coherence in the initial state of the path (as quantified by the trace norm of the off-diagonal part of the path state) and by the amount of coherence in the noisy process (as quantified by the Ky Fan 2-norm of the modulus square of the interference operator.)  
In turn, the trace norm  is upper bounded as  
\begin{align}\label{coherencebound}
\| \sigma_{\rm off-diag} \|_1  \le   2  \left(1- \frac 1k\right) \,, 
\end{align}
and the equality sign is attained by the pure state $\sigma =  |e\rangle \langle e|$, where $|e\rangle $  is the   maximally coherent state $|e\rangle  : =  ( |1\rangle   +|2\rangle +  \dots +|k\rangle)/\sqrt k $, or any other equally weighted superposition of the basis states.  
For completeness, a proof of the bound (\ref{coherencebound}) is provided  in  Appendix \ref{app:superposition}.

 As an example, we now  consider the case where $\cal C$ is the completely depolarizing qubit  channel   ${\cal C}  (\rho)  =  I/2$, previously considered in the study of quantum communication with superposition of paths \cite{Abbott2020communication,Chiribella_2019,witnessing_latent_time}.      For the Kraus representation of channel $\cal C$, we choose the Pauli operators $(\frac I2, \frac X2,  \frac Y2,  \frac Z2)$ and for the vacuum amplitudes we adopt the uniform choice $(\frac 12, \frac 12, \frac 12,  \frac 12)$, giving rise to the interference operator $F  = \frac{  I  +  X+  Y+  Z}4$.   In this case, the modulus square of $F$ is  
 \begin{align}
 F^\dag F   =  \frac I  4  + \frac{ X+Y+Z}8 \, , 
 \end{align}
 and its eigenvalues are $(2 \pm  \sqrt 3)/8$.  Hence, the Ky Fan $2$-norm is $\|  F^\dag F\|_{{\rm Ky Fan}2}  \equiv \Tr[F^\dag F]   =  1/2$.   Since the completely depolarizing channel has zero $\CM$, Eq. (\ref{CMbound}) holds with the equality sign:   the effective channel $\N_\sigma$ arising from the superposition of paths has maximum causal effect given by 
\begin{align}
\CM (\N_\sigma)   =  \frac{  \|  \sigma_{\rm off-diag}\|_1}4  \, .
\end{align}
In particular, choosing $\sigma$ to be the maximally coherent state, we obtain 
\begin{align}\label{benchmark2}
\CM (\N_{|e\rangle \langle e|})   =  \frac{1-\frac 1k}{2}    \, .
\end{align}
It is also interesting to consider the case of completely depolarizing multiqubit channels.  Assuming the same Kraus representation and the same vacuum amplitudes for each qubit, the interference operator takes the product form $F^{\otimes m}$. The modulus square $(F^\dag F)^{\otimes m}$ has eigenvalues $\left\{[(2+ \sqrt 3)/8]^j  \, [(2- \sqrt 3)/8]^{m-j}~|~  0\le j\le m   \right\}$, and therefore the Ky Fan $2$-norm is $\|  F^{\otimes m}\|_{\rm Ky Fan 2}  =[(2+ \sqrt 3)/8]^{m-1}/2$.   
When the path is initialized in the maximally coherent state $\sigma  = |e\rangle \langle e|$,  the maximum causal effect of the effective process $\N^{(m)}_{|e\rangle \langle e|}$ is 
\begin{align} \label{benchmark2.1}
\CM  (  \N^{(m)}_{|e\rangle \langle e|})  =  \left[\frac{2+ \sqrt 3}8\right]^{m-1}  \,  \frac{   1-\frac 1k }2  \, .
\end{align}
Note that the causal effect is exponentially suppressed when the number of qubits is increased, due to the exponential decrease of the Ky Fan $2$-norm of the interference operator, meaning that the interference effects become negligible for large $m$.  

We conclude with a lower bound on the minimum causal effect for quantum channels of the discard-and-reprepare form:  
\begin{proposition}\label{prop:superposition1}
For every quantum channel $\cal C$ of the discard-and-reprepare form ${\cal C}  (\rho)  =  \rho_0\,\forall \rho  \in  {\sf St}  (A)$,  the mininum causal effect  of the channel $\N_\sigma$ is lower  bounded as 
\begin{align}\label{Cmbound}
\CM  (  \N_\sigma)  \ge   \min {\rm eigv}  (F^\dag F) \, \|  \sigma_{\rm off-diag}\|_1    \,,  
\end{align}
where $\min {\rm eigv}  (F^\dag F)    $ denotes the minimum eigenvalue of the operator $F^\dag F$.
\end{proposition}
The proof is provided in Appendix \ref{app:superposition1}. A consequence of the above proposition is that the superposition of identical discard-and-reprepare channels has a positive causal effect whenever the interference operator is invertible.

   \section{Variational algorithm for estimating the maximum quantum causal effect} \label{sec:algorithm}

\begin{figure*}
    \centering
    \includegraphics[width=\linewidth]{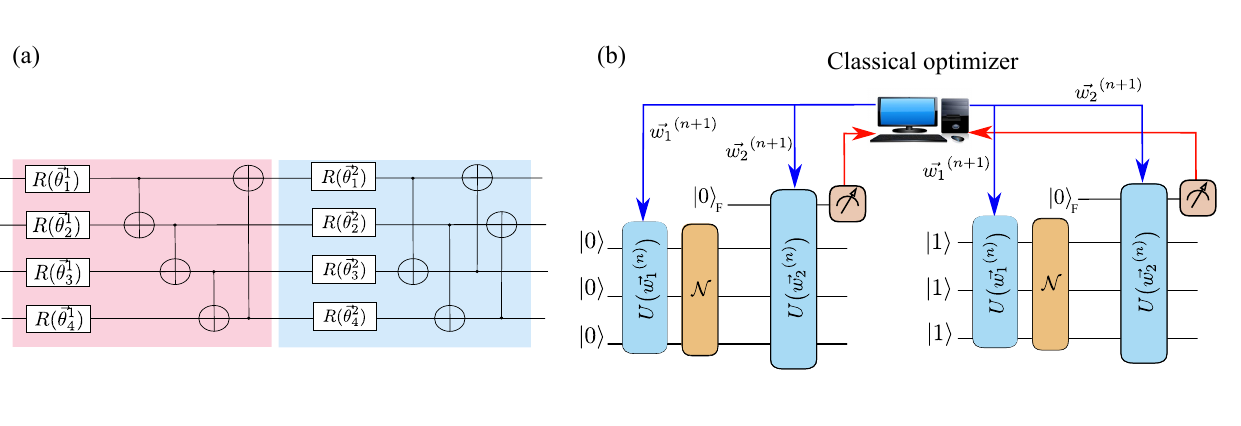}
    \caption{{\bf Variational circuit  for estimating the maximum causal effect.} 
     Fig.(a) shows our variational circuit in the case of 4 qubits.   The circuit consists of  two layers,  $\ell =1$ and $\ell =  2$,    represented by red and blue shades, respectively. Each layer includes single-qubit rotations on all the four qubits and controlled $\tt NOT$ gates on pairs of qubits.  At layer $\ell\in  \{1,2\}$, the single-qubit rotation on qubit $i\in  \{1,2,3,4\}$ is specified by a vector $\vec{\theta}_i^\ell$ of three Euler  angles. The overall unitary operator implemented by the circuit is specified by the vector $\vec w  =  ( \vec{\theta}_i^\ell )_{i\in  \{1,2,3,4\},  \ell \in \{1,2\}}$.   Fig.(b) illustrates the optimization performed by our algorithm, in the case of a quantum channel $\N$ with  three-qubit  input and three-qubit output.   At the $n$-th iteration, a pair of orthogonal states $\ket{\psi_{(n)}}=U(\vec{w_1}^{(n)})\ket{0}^{\otimes3}$ and $\ket{\psi_{(n)}^{\perp}}=U(\vec{w_1}^{(n)})\ket{1}^{\otimes3}$ is generated by applying the 3-qubit unitary gate $U(\vec{w_1}^{(n)})$ to the fixed states $|0\rangle^{\otimes 3}$  and  $|1\rangle^{\otimes 3}$.    The goal of the algorithm is to estimate the trace distance between the output states $\mathcal{N}(\psi_{(n)})$ and $\mathcal{N}(\psi_{(n)}^{\perp})$.    The trace distance is estimated by optimizing over all possible measurements on the output qubits.  The algorithm generates a measurement by applying a $4$-qubit unitary gate $U(\vec{w_2}^{(n)})$, acting on the three ouput qubits plus an auxiliary qubit $F$, initialized in the state $|0\rangle$ and measured in the computational basis $\{|0\rangle, |1\rangle\}$ after the action of the gate $U(\vec{w_2}^{(n)})$.   The algorithm computes the outcome probabilities  and feeds them to a classical optimizer, which updates the parameters $\vec{w}_1 ^{(n)}$ and $\vec{w}_2^{(n)}$ in order to maximize the total variation distance of the  probability distributions.  }
    \label{fig:entangling_unitary_algo}
\end{figure*}

The maximum causal effect can be computed analytically for channels  with sufficiently simple structure.  In general, however, it is also useful to have  techniques providing numerical estimates.   

Here we provide a variational algorithm based on Eq. (\ref{maxquantum}), which reduces the evaluation of $\CM$ to a maximization of the trace distance over pairs of pure and orthogonal input states. 
 
In turn, the algorithm uses the fact that the trace distance is equal to the maximum  of the total variation distance of the probability distributions  generated by measurements on the output system~\cite{helstrom_quantum_1976,Chen_2021,Rethinasamy_2023}.  

For simplicity, we illustrate the algorithm    for    channels $\N$ with $k$ input and $k$ output qubits.       The algorithm proceeds through a series of iterations, in which it searches for the input states  and measurement that maximize the total variation distance of the output probability distributions.

At the $n$-th iteration, a pair of orthogonal pure states  is generated by applying a parametrized unitary gate to a fixed pair of orthogonal states; specifically, the states are of the form  
\begin{align}\label{states}
|\psi^{(n)}  \rangle  :  =  U\left(\vec{w}_1^{(n)}\right)  | 0\rangle  \qquad {\rm and}  \qquad |\psi^{(n)}_\perp\rangle  :  =  U\left(\vec{w}_1^{(n)}\right)  | 1\rangle  \, ,
\end{align}
where $U\left(\vec{w}_1^{(n)}\right)$ is a parametrized $k$-qubit unitary gate specified by a vector of parameters $\vec{w}_1^{(n)}$,  called the ``weight.''      If  the variational circuit    scans all possible $k$-qubit unitary gates,  then all pair of orthogonal pure states can be obtained as  in  Eq. (\ref{states}).  

The algorithm also generates a quantum  measurement on the output qubits. The measurement is generated by applying   a parametrized $k+1$ qubit unitary gate $U\left(\vec{w}_2^{(n)}\right)$, on the $k$ output qubits and on an  auxiliary qubit  $F$,  initialized in the state $|0\rangle$, and  measured in the computational basis $\{|0\rangle,|1\rangle\}$ after the action of the gate  $U\left(\vec{w}_2^{(n)}\right)$. 
The probability of obtaining outcome $j  \in  \{0,1\}$ in the measurement, when the $k$ qubits are in the state $\rho$ is  
\begin{align}
\nonumber p(j|\rho) &= \Tr  [  \Big(I^{\otimes k}\otimes |j\rangle \langle j|\Big)  \, U\left(\vec{w}_2^{(n)}  \right) \Big(\rho  \otimes |0\rangle \langle 0|\Big)U\left(\vec{w}_2^{(n)}  \right)^\dag] \\
&=:   \Tr  [ P_j^{(n)}  \rho]  \,  ,  \end{align}
with
\begin{align}
P^{(n)}_j :  = \Big(I^{\otimes k}\otimes  \langle 0|\Big) U\left(\vec{w}_2^{(n)}  \right)^\dag \Big(I^{\otimes k}   \otimes |j\rangle\langle j|\Big) U\left(\vec{w}_2^{(n)}  \right)   \Big(I^{\otimes k}\otimes |0 \rangle\Big)  \,.
\end{align}
If the variational circuit scans all possible $(k+1)$-qubit unitaries, then this scheme scans all possible pairs of positive operators $P_0$ and $P_1$ satisfying the normalization condition   $P_0 +  P_1  =  I_k^{\otimes k}$, as shown in Appendix~\ref{App:strongly_entangling}.

The goal of the algorithm is to find a pair of input states and a measurement that maximize    the difference between the outcome probabilities after the action of channel $\N$; explicitly, the objective function to be maximized is   
\begin{align}\label{CWW}
C(  w^{(n)}_1,  w^{(n)}_2)  :  =  p\Big(1\big|  \,   \N\left( \psi^{(n)}\right)  \,  \Big)  -  p\Big(1\big| \,  \N\left(\psi^{(n)}_\perp\right) \, \Big)  \, ,  \end{align}
where we used the shorthand notations $\psi^{(n)} :  =|\psi^{(n)} \rangle \langle \psi^{(n)}|$  and $\psi^{(n)}_\perp :  =|\psi^{(n)}_\perp \rangle \langle \psi^{(n)}_\perp|$.  
Note that the maximization guarantees that  $C(  w^{(n)}_1,  w^{(n)}_2)$ is non-negative. In this case,  $C(  w^{(n)}_1,  w^{(n)}_2)$ is equal to half of the total variation distance of the  probability distributions $p\Big(j\big|  \,   \N( \psi^{(n)})  \,  \Big) $ and $p\Big(j\big| \,  \N(\psi^{(n)}_\perp) \, \Big)$. 

If the variational circuit scans all possible $k$-qubit gates and $(k+1)$-qubit gates, then the maximum of  $C(  w^{(n)}_1,  w^{(n)}_2)$ yields the maximum causal effect:  
\begin{align}\label{}
\nonumber &\max_{w^{(n)}_1,  w^{(n)}_2}   C(  w^{(n)}_1,  w^{(n)}_2)   \\  
\nonumber & =   \max_{\psi^{(n)},  \psi_\perp^{(n)}  } \max_{P_0, P_1\ge 0, P_0+P_1=  I}  \Tr[P_1^{(n)}  \N  (\psi^{(n)})]  -  \Tr[P_1^{(n)}  \N  (\psi_\perp^{(n)})]\\
\nonumber &= \frac 12 \max_{\psi^{(n)},  \psi^{(n)}_\perp  }\|  \N  (\psi^{(n)})  -  \N (\psi_\perp^{(n)})\|_1 \\
&  =     \CM  (\N) \, . \label{algorithm}  
\end{align}

The problem, of course, is that the number of parameters needed to specify an arbitrary $k$-qubit gate grows exponentially with $k$, which makes a full optimization unfeasible for large values of $k$.  
As a  feasible alternative, we adopt a variational circuit of a simpler form, proposed in Ref.~\cite{schuld_circuit-centric_2020}.     The circuit consists of  $L$ layers.    In each layer, all the qubits  undergo arbitrary local unitary rotations, followed by  $\tt CNOT$ gates. At layer $\ell$,  the $\tt CNOT$ gates are implemented on pairs of qubits that are at distance $\ell$ (modulo $k$): for all $i\in  \{1,\dots,  k\}$ a $\tt CNOT$ is applied between  the $i$-th qubit  (serving as the control) and the  $i{+}\ell(\rm{mod}\  k)$-th qubit (serving as the target.)     An illustration of the circuit in the case of $k=4$ qubits and $L=2$ layers is provided in  Fig.~\ref{fig:entangling_unitary_algo}(a).

The maximization of the objective function over $\vec{w}^{(n)}_1$ and $\vec{w}^{(n)}_2$ is carried out by using a gradient descent algorithm supported by the standard package Pennylane.   The codes used in our implementation of the algorithm  are available on Ref.~\cite{Goswami_github}.
 The algorithm terminates   when the value of objective function  does not change significantly in subsequent iterations, or when  the  number of iterations  has reached a maximum value.  Here, we set a  threshold value $\epsilon = 10^{-5}$ for  the changes in the objective function, and a maximum value  $2000$   for the number of iterations.  
After termination, the algorithm provides a lower bound to  the maximum causal effect of channel $\N$.

\begin{figure}
    \centering
    \includegraphics[width=\columnwidth]{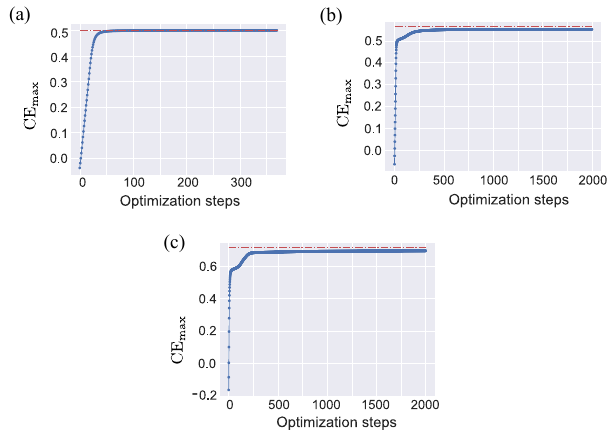}
    \caption{{\bf Benchmarking the algorithm for  a coherent superposition of  common cause and direct cause.} Subfigures \ref{fig:partial_swap_pi_4}(a),  \ref{fig:partial_swap_pi_4}(b), and \ref{fig:partial_swap_pi_4}(c) correspond to different choices of the state $\rho_\Lambda$ in  Eq. \eqref{Eq:channel_partial_swap}.  In all three cases, the angle $\theta$ is set to $\pi/4$, and  systems  $A$, $B$ and $\Lambda$ are  taken to be eight-dimensional ({\em i.e.},  each system consists of three qubits.)       The red dashed lines represent the true value of ${\CM}$, while 
    the blue dots represent subsequent iterations of our algorithm, in the scenario of Proposition~\ref{Lem:partial_swap}.   
   In Fig. (a) $p=0$, i.e., the marginal state $\rho_{\Lambda}$ at $\Lambda$ is  maximally mixed, and the true value is $0.5$ (the algorithm converged to $\approx 0.4999$).   In Fig.(b), $p=0.5$, and the true value is $0.591$ (the algorithm converged to $\approx 0.553$), and in Fig.(c), $p=1$ indicating $\rho_{\Lambda}$ is a pure state, and the true value is $0.741$ (the algorithm converged to $\approx 0.6959$). }
    \label{fig:partial_swap_pi_4}
\end{figure}

In the following, we evaluate the lower bound computed by our variational algorithm in the two examples shown in Section \ref{sec:applications}.   In the case of the common cause/direct cause setup,  the results of our algorithm can be benchmarked against  Eq. (\ref{CMmaclean}), which provides the exact value of the maximum causal effect.  The results are presented in Fig.~\ref{fig:partial_swap_pi_4}.  There, we set $d=8$ and $\theta =  \pi/4$, corresponding to a uniform  superposition of common cause and direct cause  configurations. 
  The three subfigures  show the results for three different values of the probability $p$ in Eq. (\ref{sigmalambda1}); specifically, the values are $p  =  0,  \frac 12,$ and $1$ and correspond to the values $\CM  (\N)  =  0.5,  0.591 $, and $0.741$ of the maximum causal effect, respectively.  In all these cases,  Fig.~\ref{fig:partial_swap_pi_4} shows that the numerical value found by the algorithm after the last iteration is close to the true value computed with Eq. (\ref{CMmaclean}).

\begin{figure}
    \centering
    \includegraphics[width=\columnwidth]{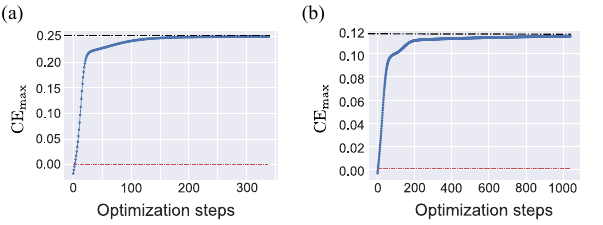}
    \caption{{\bf Benchmarking the algorithm for a coherent superposition of two completely depolarizing channels.}Signaling through coherent superposition of two completely depolarizing channels ($q=0$). We consider a 2-dimensional target system in (a) and a 4-dimensional target system in (b). We have a control qubit in $\ketbra{+}{+}$ state in both cases. The red line represents the bound on causal influence for the convex combination of the corresponding depolarization channels, which for the scenario of maximally depolarizing channels is $0$, signifying no causal influence. In contrast to the classical mixture, a coherent superposition allows more signaling, with the black dashed line representing the true value of the maximum quantum causal effect given by Proposition~\ref{prop:superposition}. For the 2-dimensional target (Fig.(a)), the algorithm converged to a ${\CM} \approx 0.2499$ (true value is $0.25$ as in Eq.~\eqref{benchmark2} with $k=2$), and for the 4-dimensional target(Fig. (b)), the algorithm converged to $\CM \approx 0.1152$ (true value is $0.1166$ as in Eq.~\eqref{benchmark2.1} with $m{=}k{=}2$  ). }
    \label{fig:coherent_sup1}
\end{figure}

Next, we considered the application of our algorithm to the estimation of $\CM$ for a coherent superposition of two completely depolarizing channels, using Eq.  (\ref{benchmark2}) for $k=2$ as a benchmark for the algorithm's estimate.  The results are shown   in Fig.~\ref{fig:coherent_sup1} for completely depolarizing channels acting on one qubit (a) and on two qubits (b). Also, in this case, the numerical values provided by the algorithm appear to be close to the exact value of the maximum causal effect provided by Eqs.~\eqref{benchmark2}, and ~\eqref{benchmark2.1}.

\section{Conclusions} \label{sec:discussion}
In this paper, we introduced the  maximum and minimum causal effects, two quantities that measure  the extent of causal influence that the input of a physical process exerts on the output.      The maximum causal effect provides a natural generalization of the ACE, a popular quantifier  that has been extensively used in classical causal inference.   In a similar way, we expect that the maximum causal effect  will find applications in the characterization of quantum causal relations,  both theoretically and experimentally.   
The minimum causal effect, instead,  quantifies the ability of physical processes to preserve the distinguishability of pairs of input states.  In quantum theory, it is closely connected with the recoverability of quantum information:  a quantum process has a high value of the minimum causal effect if and only if it is approximately correctable.

   The maximum and minimum causal effects enjoy  useful properties, including continuity and data-processing inequality.   In quantum theory, they satisfy a monogamy relation: if a quantum system $A$ has a high value of the minimum quantum effect on system $B$, then it must have a low value of the maximum causal effect on any system $B'$ that is spacelike separated with $B$.

Overall, the concepts and methods developed in this work enable a quantitative study of quantum causal relations  and have significant implications for quantum information processing and quantum technologies.

\begin{acknowledgments}
   KG thanks Siddhartha Das, Tamal Guha, and Saptarshi Roy for discussions. This work is supported by the Innovation Program for Quantum Science and Technology (Grant No. 2023ZD0300600),  by the Hong Kong Research Grant Council (RGC) through Grants No.  17307520,  SRFS2021-7S02, and R7035-21F,    and by the John Templeton Foundation through grant  62312, ``The Quantum Information Structure of Spacetime'' (qiss.fr). The opinions expressed in this publication are those of the authors and do not necessarily reflect the views of the John Templeton Foundation.  Research at the Perimeter Institute is supported by the Government of Canada through the Department of Innovation, Science and Economic Development Canada and by the Province of Ontario through the Ministry of Research, Innovation and Science.
\end{acknowledgments}

\bibliography{causal_influence}

\begin{thebibliography}{87}%
\makeatletter
\providecommand \@ifxundefined [1]{%
 \@ifx{#1\undefined}
}%
\providecommand \@ifnum [1]{%
 \ifnum #1\expandafter \@firstoftwo
 \else \expandafter \@secondoftwo
 \fi
}%
\providecommand \@ifx [1]{%
 \ifx #1\expandafter \@firstoftwo
 \else \expandafter \@secondoftwo
 \fi
}%
\providecommand \natexlab [1]{#1}%
\providecommand \enquote  [1]{``#1''}%
\providecommand \bibnamefont  [1]{#1}%
\providecommand \bibfnamefont [1]{#1}%
\providecommand \citenamefont [1]{#1}%
\providecommand \href@noop [0]{\@secondoftwo}%
\providecommand \href [0]{\begingroup \@sanitize@url \@href}%
\providecommand \@href[1]{\@@startlink{#1}\@@href}%
\providecommand \@@href[1]{\endgroup#1\@@endlink}%
\providecommand \@sanitize@url [0]{\catcode `\\12\catcode `\$12\catcode
  `\&12\catcode `\#12\catcode `\^12\catcode `\_12\catcode `\%12\relax}%
\providecommand \@@startlink[1]{}%
\providecommand \@@endlink[0]{}%
\providecommand \url  [0]{\begingroup\@sanitize@url \@url }%
\providecommand \@url [1]{\endgroup\@href {#1}{\urlprefix }}%
\providecommand \urlprefix  [0]{URL }%
\providecommand \Eprint [0]{\href }%
\providecommand \doibase [0]{http://dx.doi.org/}%
\providecommand \selectlanguage [0]{\@gobble}%
\providecommand \bibinfo  [0]{\@secondoftwo}%
\providecommand \bibfield  [0]{\@secondoftwo}%
\providecommand \translation [1]{[#1]}%
\providecommand \BibitemOpen [0]{}%
\providecommand \bibitemStop [0]{}%
\providecommand \bibitemNoStop [0]{.\EOS\space}%
\providecommand \EOS [0]{\spacefactor3000\relax}%
\providecommand \BibitemShut  [1]{\csname bibitem#1\endcsname}%
\let\auto@bib@innerbib\@empty
\bibitem [{\citenamefont {Pearl}(2009)}]{pearl_2009}%
  \BibitemOpen
  \bibfield  {author} {\bibinfo {author} {\bibfnamefont {J.}~\bibnamefont
  {Pearl}},\ }\href {\doibase 10.1017/CBO9780511803161} {\emph {\bibinfo
  {title} {Causality: {{Models}}, {{Reasoning}}, and {{Inference}}}}},\
  \bibinfo {edition} {2nd}\ ed.\ (\bibinfo  {publisher} {{Cambridge University
  Press}},\ \bibinfo {year} {2009})\BibitemShut {NoStop}%
\bibitem [{\citenamefont {Costa}\ and\ \citenamefont
  {Shrapnel}(2016)}]{costa2016}%
  \BibitemOpen
  \bibfield  {author} {\bibinfo {author} {\bibfnamefont {F.}~\bibnamefont
  {Costa}}\ and\ \bibinfo {author} {\bibfnamefont {S.}~\bibnamefont
  {Shrapnel}},\ }\href {\doibase https://doi.org/10.1088/1367-2630/18/6/063032}
  {\bibfield  {journal} {\bibinfo  {journal} {New J. Phys.}\ }\textbf {\bibinfo
  {volume} {18}},\ \bibinfo {pages} {063032} (\bibinfo {year}
  {2016})}\BibitemShut {NoStop}%
\bibitem [{\citenamefont {Allen}\ \emph {et~al.}(2017)\citenamefont {Allen},
  \citenamefont {Barrett}, \citenamefont {Horsman}, \citenamefont {Lee},\ and\
  \citenamefont {Spekkens}}]{Spekkens_cause}%
  \BibitemOpen
  \bibfield  {author} {\bibinfo {author} {\bibfnamefont {J.-M.~A.}\
  \bibnamefont {Allen}}, \bibinfo {author} {\bibfnamefont {J.}~\bibnamefont
  {Barrett}}, \bibinfo {author} {\bibfnamefont {D.~C.}\ \bibnamefont
  {Horsman}}, \bibinfo {author} {\bibfnamefont {C.~M.}\ \bibnamefont {Lee}}, \
  and\ \bibinfo {author} {\bibfnamefont {R.~W.}\ \bibnamefont {Spekkens}},\
  }\href {\doibase 10.1103/PhysRevX.7.031021} {\bibfield  {journal} {\bibinfo
  {journal} {Phys. Rev. X}\ }\textbf {\bibinfo {volume} {7}},\ \bibinfo {pages}
  {031021} (\bibinfo {year} {2017})}\BibitemShut {NoStop}%
\bibitem [{\citenamefont {Barrett}\ \emph {et~al.}(2019)\citenamefont
  {Barrett}, \citenamefont {Lorenz},\ and\ \citenamefont
  {Oreshkov}}]{barrett_2019}%
  \BibitemOpen
  \bibfield  {author} {\bibinfo {author} {\bibfnamefont {J.}~\bibnamefont
  {Barrett}}, \bibinfo {author} {\bibfnamefont {R.}~\bibnamefont {Lorenz}}, \
  and\ \bibinfo {author} {\bibfnamefont {O.}~\bibnamefont {Oreshkov}},\ }\href
  {\doibase 10.48550/ARXIV.1906.10726} {\  (\bibinfo {year} {2019}),\
  10.48550/ARXIV.1906.10726}\BibitemShut {NoStop}%
\bibitem [{\citenamefont {Barrett}\ \emph {et~al.}(2021)\citenamefont
  {Barrett}, \citenamefont {Lorenz},\ and\ \citenamefont
  {Oreshkov}}]{Barrett_2021}%
  \BibitemOpen
  \bibfield  {author} {\bibinfo {author} {\bibfnamefont {J.}~\bibnamefont
  {Barrett}}, \bibinfo {author} {\bibfnamefont {R.}~\bibnamefont {Lorenz}}, \
  and\ \bibinfo {author} {\bibfnamefont {O.}~\bibnamefont {Oreshkov}},\ }\href
  {\doibase 10.1038/s41467-020-20456-x} {\bibfield  {journal} {\bibinfo
  {journal} {Nat. Commun.}\ }\textbf {\bibinfo {volume} {12}} (\bibinfo {year}
  {2021}),\ 10.1038/s41467-020-20456-x}\BibitemShut {NoStop}%
\bibitem [{\citenamefont {Ried}\ \emph {et~al.}(2015)\citenamefont {Ried},
  \citenamefont {Agnew}, \citenamefont {Vermeyden}, \citenamefont {Janzing},
  \citenamefont {Spekkens},\ and\ \citenamefont {Resch}}]{Ried_2015}%
  \BibitemOpen
  \bibfield  {author} {\bibinfo {author} {\bibfnamefont {K.}~\bibnamefont
  {Ried}}, \bibinfo {author} {\bibfnamefont {M.}~\bibnamefont {Agnew}},
  \bibinfo {author} {\bibfnamefont {L.}~\bibnamefont {Vermeyden}}, \bibinfo
  {author} {\bibfnamefont {D.}~\bibnamefont {Janzing}}, \bibinfo {author}
  {\bibfnamefont {R.~W.}\ \bibnamefont {Spekkens}}, \ and\ \bibinfo {author}
  {\bibfnamefont {K.~J.}\ \bibnamefont {Resch}},\ }\href {\doibase
  10.1038/nphys3266} {\bibfield  {journal} {\bibinfo  {journal} {Nature
  Physics}\ }\textbf {\bibinfo {volume} {11}},\ \bibinfo {pages} {414–420}
  (\bibinfo {year} {2015})}\BibitemShut {NoStop}%
\bibitem [{\citenamefont {MacLean}\ \emph {et~al.}(2017)\citenamefont
  {MacLean}, \citenamefont {Ried}, \citenamefont {Spekkens},\ and\
  \citenamefont {Resch}}]{MacLean_2017}%
  \BibitemOpen
  \bibfield  {author} {\bibinfo {author} {\bibfnamefont {J.-P.~W.}\
  \bibnamefont {MacLean}}, \bibinfo {author} {\bibfnamefont {K.}~\bibnamefont
  {Ried}}, \bibinfo {author} {\bibfnamefont {R.~W.}\ \bibnamefont {Spekkens}},
  \ and\ \bibinfo {author} {\bibfnamefont {K.~J.}\ \bibnamefont {Resch}},\
  }\href {\doibase 10.1038/ncomms15149} {\bibfield  {journal} {\bibinfo
  {journal} {Nat. Commun.}\ }\textbf {\bibinfo {volume} {8}} (\bibinfo {year}
  {2017}),\ 10.1038/ncomms15149}\BibitemShut {NoStop}%
\bibitem [{\citenamefont {Carvacho}\ \emph {et~al.}(2017)\citenamefont
  {Carvacho}, \citenamefont {Andreoli}, \citenamefont {Santodonato},
  \citenamefont {Bentivegna}, \citenamefont {Chaves},\ and\ \citenamefont
  {Sciarrino}}]{Carvacho2017}%
  \BibitemOpen
  \bibfield  {author} {\bibinfo {author} {\bibfnamefont {G.}~\bibnamefont
  {Carvacho}}, \bibinfo {author} {\bibfnamefont {F.}~\bibnamefont {Andreoli}},
  \bibinfo {author} {\bibfnamefont {L.}~\bibnamefont {Santodonato}}, \bibinfo
  {author} {\bibfnamefont {M.}~\bibnamefont {Bentivegna}}, \bibinfo {author}
  {\bibfnamefont {R.}~\bibnamefont {Chaves}}, \ and\ \bibinfo {author}
  {\bibfnamefont {F.}~\bibnamefont {Sciarrino}},\ }\href {\doibase
  10.1038/ncomms14775} {\bibfield  {journal} {\bibinfo  {journal} {Nature
  Communications}\ }\textbf {\bibinfo {volume} {8}} (\bibinfo {year} {2017}),\
  10.1038/ncomms14775}\BibitemShut {NoStop}%
\bibitem [{\citenamefont {Chaves}\ \emph {et~al.}(2017)\citenamefont {Chaves},
  \citenamefont {Carvacho}, \citenamefont {Agresti}, \citenamefont {Di~Giulio},
  \citenamefont {Aolita}, \citenamefont {Giacomini},\ and\ \citenamefont
  {Sciarrino}}]{Chaves2017}%
  \BibitemOpen
  \bibfield  {author} {\bibinfo {author} {\bibfnamefont {R.}~\bibnamefont
  {Chaves}}, \bibinfo {author} {\bibfnamefont {G.}~\bibnamefont {Carvacho}},
  \bibinfo {author} {\bibfnamefont {I.}~\bibnamefont {Agresti}}, \bibinfo
  {author} {\bibfnamefont {V.}~\bibnamefont {Di~Giulio}}, \bibinfo {author}
  {\bibfnamefont {L.}~\bibnamefont {Aolita}}, \bibinfo {author} {\bibfnamefont
  {S.}~\bibnamefont {Giacomini}}, \ and\ \bibinfo {author} {\bibfnamefont
  {F.}~\bibnamefont {Sciarrino}},\ }\href {\doibase 10.1038/s41567-017-0008-5}
  {\bibfield  {journal} {\bibinfo  {journal} {Nature Physics}\ }\textbf
  {\bibinfo {volume} {14}},\ \bibinfo {pages} {291–296} (\bibinfo {year}
  {2017})}\BibitemShut {NoStop}%
\bibitem [{\citenamefont {Agresti}\ \emph {et~al.}(2022)\citenamefont
  {Agresti}, \citenamefont {Poderini}, \citenamefont {Polacchi}, \citenamefont
  {Miklin}, \citenamefont {Gachechiladze}, \citenamefont {Suprano},
  \citenamefont {Polino}, \citenamefont {Milani}, \citenamefont {Carvacho},
  \citenamefont {Chaves},\ and\ \citenamefont {Sciarrino}}]{Agresti2022}%
  \BibitemOpen
  \bibfield  {author} {\bibinfo {author} {\bibfnamefont {I.}~\bibnamefont
  {Agresti}}, \bibinfo {author} {\bibfnamefont {D.}~\bibnamefont {Poderini}},
  \bibinfo {author} {\bibfnamefont {B.}~\bibnamefont {Polacchi}}, \bibinfo
  {author} {\bibfnamefont {N.}~\bibnamefont {Miklin}}, \bibinfo {author}
  {\bibfnamefont {M.}~\bibnamefont {Gachechiladze}}, \bibinfo {author}
  {\bibfnamefont {A.}~\bibnamefont {Suprano}}, \bibinfo {author} {\bibfnamefont
  {E.}~\bibnamefont {Polino}}, \bibinfo {author} {\bibfnamefont
  {G.}~\bibnamefont {Milani}}, \bibinfo {author} {\bibfnamefont
  {G.}~\bibnamefont {Carvacho}}, \bibinfo {author} {\bibfnamefont
  {R.}~\bibnamefont {Chaves}}, \ and\ \bibinfo {author} {\bibfnamefont
  {F.}~\bibnamefont {Sciarrino}},\ }\href {\doibase 10.1126/sciadv.abm1515}
  {\bibfield  {journal} {\bibinfo  {journal} {Science Advances}\ }\textbf
  {\bibinfo {volume} {8}} (\bibinfo {year} {2022}),\
  10.1126/sciadv.abm1515}\BibitemShut {NoStop}%
\bibitem [{\citenamefont {Feix}\ and\ \citenamefont {Časlav
  Brukner}(2017)}]{Feix_2017}%
  \BibitemOpen
  \bibfield  {author} {\bibinfo {author} {\bibfnamefont {A.}~\bibnamefont
  {Feix}}\ and\ \bibinfo {author} {\bibnamefont {Časlav Brukner}},\ }\href
  {\doibase 10.1088/1367-2630/aa9b1a} {\bibfield  {journal} {\bibinfo
  {journal} {New J. Phys.}\ }\textbf {\bibinfo {volume} {19}},\ \bibinfo
  {pages} {123028} (\bibinfo {year} {2017})}\BibitemShut {NoStop}%
\bibitem [{\citenamefont {{Chiribella}}\ \emph {et~al.}(2013)\citenamefont
  {{Chiribella}}, \citenamefont {{D'Ariano}}, \citenamefont {{Perinotti}},\
  and\ \citenamefont {{Valiron}}}]{chiribella09}%
  \BibitemOpen
  \bibfield  {author} {\bibinfo {author} {\bibfnamefont {G.}~\bibnamefont
  {{Chiribella}}}, \bibinfo {author} {\bibfnamefont {G.~M.}\ \bibnamefont
  {{D'Ariano}}}, \bibinfo {author} {\bibfnamefont {P.}~\bibnamefont
  {{Perinotti}}}, \ and\ \bibinfo {author} {\bibfnamefont {B.}~\bibnamefont
  {{Valiron}}},\ }\href {\doibase 10.1103/PhysRevA.88.022318} {\bibfield
  {journal} {\bibinfo  {journal} {Phys. Rev.~A}\ }\textbf {\bibinfo {volume}
  {88}},\ \bibinfo {eid} {022318} (\bibinfo {year} {2013})},\ \Eprint
  {http://arxiv.org/abs/0912.0195} {arXiv:0912.0195 [quant-ph]} \BibitemShut
  {NoStop}%
\bibitem [{\citenamefont {{Oreshkov}}\ \emph {et~al.}(2012)\citenamefont
  {{Oreshkov}}, \citenamefont {{Costa}},\ and\ \citenamefont
  {{Brukner}}}]{oreshkov12}%
  \BibitemOpen
  \bibfield  {author} {\bibinfo {author} {\bibfnamefont {O.}~\bibnamefont
  {{Oreshkov}}}, \bibinfo {author} {\bibfnamefont {F.}~\bibnamefont {{Costa}}},
  \ and\ \bibinfo {author} {\bibfnamefont {{\v C}.}~\bibnamefont {{Brukner}}},\
  }\href {\doibase 10.1038/ncomms2076} {\bibfield  {journal} {\bibinfo
  {journal} {Nat. Commun.}\ }\textbf {\bibinfo {volume} {3}},\ \bibinfo {pages}
  {1092} (\bibinfo {year} {2012})},\ \Eprint {http://arxiv.org/abs/1105.4464}
  {arXiv:1105.4464 [quant-ph]} \BibitemShut {NoStop}%
\bibitem [{\citenamefont {Chiribella}\ and\ \citenamefont
  {Liu}(2022)}]{chiribella_Indefinite2022}%
  \BibitemOpen
  \bibfield  {author} {\bibinfo {author} {\bibfnamefont {G.}~\bibnamefont
  {Chiribella}}\ and\ \bibinfo {author} {\bibfnamefont {Z.}~\bibnamefont
  {Liu}},\ }\href {\doibase 10.1038/s42005-022-00967-3} {\bibfield  {journal}
  {\bibinfo  {journal} {Commun. Phys.}\ }\textbf {\bibinfo {volume} {5}}
  (\bibinfo {year} {2022}),\ 10.1038/s42005-022-00967-3}\BibitemShut {NoStop}%
\bibitem [{\citenamefont {Janzing}\ \emph {et~al.}(2013)\citenamefont
  {Janzing}, \citenamefont {Balduzzi}, \citenamefont {{Grosse-Wentrup}},\ and\
  \citenamefont {Sch{\"o}lkopf}}]{janzing_2013}%
  \BibitemOpen
  \bibfield  {author} {\bibinfo {author} {\bibfnamefont {D.}~\bibnamefont
  {Janzing}}, \bibinfo {author} {\bibfnamefont {D.}~\bibnamefont {Balduzzi}},
  \bibinfo {author} {\bibfnamefont {M.}~\bibnamefont {{Grosse-Wentrup}}}, \
  and\ \bibinfo {author} {\bibfnamefont {B.}~\bibnamefont {Sch{\"o}lkopf}},\
  }\href {\doibase 10.1214/13-AOS1145} {\bibfield  {journal} {\bibinfo
  {journal} {Ann. Statist.}\ }\textbf {\bibinfo {volume} {41}} (\bibinfo {year}
  {2013}),\ 10.1214/13-AOS1145}\BibitemShut {NoStop}%
\bibitem [{\citenamefont {Henson}\ \emph {et~al.}(2014)\citenamefont {Henson},
  \citenamefont {Lal},\ and\ \citenamefont {Pusey}}]{henson2014theory}%
  \BibitemOpen
  \bibfield  {author} {\bibinfo {author} {\bibfnamefont {J.}~\bibnamefont
  {Henson}}, \bibinfo {author} {\bibfnamefont {R.}~\bibnamefont {Lal}}, \ and\
  \bibinfo {author} {\bibfnamefont {M.~F.}\ \bibnamefont {Pusey}},\ }\href@noop
  {} {\bibfield  {journal} {\bibinfo  {journal} {New Journal of Physics}\
  }\textbf {\bibinfo {volume} {16}},\ \bibinfo {pages} {113043} (\bibinfo
  {year} {2014})}\BibitemShut {NoStop}%
\bibitem [{\citenamefont {Chiribella}\ and\ \citenamefont
  {Ebler}(2019)}]{Chiribella_quantum_speedup_2019}%
  \BibitemOpen
  \bibfield  {author} {\bibinfo {author} {\bibfnamefont {G.}~\bibnamefont
  {Chiribella}}\ and\ \bibinfo {author} {\bibfnamefont {D.}~\bibnamefont
  {Ebler}},\ }\href {\doibase 10.1038/s41467-019-09383-8} {\bibfield  {journal}
  {\bibinfo  {journal} {Nat. Commun.}\ }\textbf {\bibinfo {volume} {10}}
  (\bibinfo {year} {2019}),\ 10.1038/s41467-019-09383-8}\BibitemShut {NoStop}%
\bibitem [{\citenamefont {Chiribella}\ and\ \citenamefont
  {Swati}(2021)}]{chiribella2021fast}%
  \BibitemOpen
  \bibfield  {author} {\bibinfo {author} {\bibfnamefont {G.}~\bibnamefont
  {Chiribella}}\ and\ \bibinfo {author} {\bibnamefont {Swati}},\ }in\
  \href@noop {} {\emph {\bibinfo {booktitle} {Quantum Theory and Symmetries:
  Proceedings of the 11th International Symposium, Montreal, Canada}}}\
  (\bibinfo {organization} {Springer},\ \bibinfo {year} {2021})\ pp.\ \bibinfo
  {pages} {617--632}\BibitemShut {NoStop}%
\bibitem [{\citenamefont {Perinotti}(2021)}]{perinotti_2021}%
  \BibitemOpen
  \bibfield  {author} {\bibinfo {author} {\bibfnamefont {P.}~\bibnamefont
  {Perinotti}},\ }\href {\doibase 10.22331/q-2021-08-03-515} {\bibfield
  {journal} {\bibinfo  {journal} {Quantum}\ }\textbf {\bibinfo {volume} {5}},\
  \bibinfo {pages} {515} (\bibinfo {year} {2021})},\ \Eprint
  {http://arxiv.org/abs/2012.15213} {arxiv:2012.15213 [quant-ph]} \BibitemShut
  {NoStop}%
\bibitem [{\citenamefont {{Escol{\`a}-Farr{\`a}s}}\ and\ \citenamefont
  {Braun}(2022)}]{escola_2022}%
  \BibitemOpen
  \bibfield  {author} {\bibinfo {author} {\bibfnamefont {L.}~\bibnamefont
  {{Escol{\`a}-Farr{\`a}s}}}\ and\ \bibinfo {author} {\bibfnamefont
  {D.}~\bibnamefont {Braun}},\ }\href {\doibase 10.1103/PhysRevA.106.062415}
  {\bibfield  {journal} {\bibinfo  {journal} {Phys. Rev. A}\ }\textbf {\bibinfo
  {volume} {106}},\ \bibinfo {pages} {062415} (\bibinfo {year}
  {2022})}\BibitemShut {NoStop}%
\bibitem [{\citenamefont {Bai}\ \emph {et~al.}(2022)\citenamefont {Bai},
  \citenamefont {Wu}, \citenamefont {Zhu}, \citenamefont {Hayashi},\ and\
  \citenamefont {Chiribella}}]{Bai_2022}%
  \BibitemOpen
  \bibfield  {author} {\bibinfo {author} {\bibfnamefont {G.}~\bibnamefont
  {Bai}}, \bibinfo {author} {\bibfnamefont {Y.-D.}\ \bibnamefont {Wu}},
  \bibinfo {author} {\bibfnamefont {Y.}~\bibnamefont {Zhu}}, \bibinfo {author}
  {\bibfnamefont {M.}~\bibnamefont {Hayashi}}, \ and\ \bibinfo {author}
  {\bibfnamefont {G.}~\bibnamefont {Chiribella}},\ }\href {\doibase
  10.1038/s41534-022-00578-4} {\bibfield  {journal} {\bibinfo  {journal} {npj
  Quantum Inf.}\ }\textbf {\bibinfo {volume} {8}} (\bibinfo {year} {2022}),\
  10.1038/s41534-022-00578-4}\BibitemShut {NoStop}%
\bibitem [{\citenamefont {Yi}\ and\ \citenamefont {Bose}(2022)}]{yi_2022}%
  \BibitemOpen
  \bibfield  {author} {\bibinfo {author} {\bibfnamefont {B.}~\bibnamefont
  {Yi}}\ and\ \bibinfo {author} {\bibfnamefont {S.}~\bibnamefont {Bose}},\
  }\href {\doibase 10.1103/PhysRevLett.129.020501} {\bibfield  {journal}
  {\bibinfo  {journal} {Phys. Rev. Lett.}\ }\textbf {\bibinfo {volume} {129}},\
  \bibinfo {pages} {020501} (\bibinfo {year} {2022})}\BibitemShut {NoStop}%
\bibitem [{\citenamefont {Hutter}\ \emph {et~al.}(2023)\citenamefont {Hutter},
  \citenamefont {Chaves}, \citenamefont {Nery}, \citenamefont {Moreno},\ and\
  \citenamefont {Brod}}]{hutter2022quantifying}%
  \BibitemOpen
  \bibfield  {author} {\bibinfo {author} {\bibfnamefont {L.}~\bibnamefont
  {Hutter}}, \bibinfo {author} {\bibfnamefont {R.}~\bibnamefont {Chaves}},
  \bibinfo {author} {\bibfnamefont {R.~V.}\ \bibnamefont {Nery}}, \bibinfo
  {author} {\bibfnamefont {G.}~\bibnamefont {Moreno}}, \ and\ \bibinfo {author}
  {\bibfnamefont {D.~J.}\ \bibnamefont {Brod}},\ }\href {\doibase
  10.1103/PhysRevA.108.022222} {\bibfield  {journal} {\bibinfo  {journal}
  {Phys. Rev. A}\ }\textbf {\bibinfo {volume} {108}},\ \bibinfo {pages}
  {022222} (\bibinfo {year} {2023})}\BibitemShut {NoStop}%
\bibitem [{\citenamefont {Perinotti}\ \emph {et~al.}(2023)\citenamefont
  {Perinotti}, \citenamefont {Tosini},\ and\ \citenamefont
  {Vaglini}}]{perinotti_2023}%
  \BibitemOpen
  \bibfield  {author} {\bibinfo {author} {\bibfnamefont {P.}~\bibnamefont
  {Perinotti}}, \bibinfo {author} {\bibfnamefont {A.}~\bibnamefont {Tosini}}, \
  and\ \bibinfo {author} {\bibfnamefont {L.}~\bibnamefont {Vaglini}},\ }\href
  {\doibase 10.48550/ARXIV.2309.07771} {\  (\bibinfo {year} {2023}),\
  10.48550/ARXIV.2309.07771}\BibitemShut {NoStop}%
\bibitem [{\citenamefont {Angrist}\ \emph {et~al.}(1996)\citenamefont
  {Angrist}, \citenamefont {Imbens},\ and\ \citenamefont
  {Rubin}}]{angrist_1996}%
  \BibitemOpen
  \bibfield  {author} {\bibinfo {author} {\bibfnamefont {J.~D.}\ \bibnamefont
  {Angrist}}, \bibinfo {author} {\bibfnamefont {G.~W.}\ \bibnamefont {Imbens}},
  \ and\ \bibinfo {author} {\bibfnamefont {D.~B.}\ \bibnamefont {Rubin}},\
  }\href {\doibase 10.1080/01621459.1996.10476902} {\bibfield  {journal}
  {\bibinfo  {journal} {J. Am. Stat. Assoc.}\ }\textbf {\bibinfo {volume}
  {91}},\ \bibinfo {pages} {444} (\bibinfo {year} {1996})}\BibitemShut
  {NoStop}%
\bibitem [{\citenamefont {Balke}\ and\ \citenamefont
  {Pearl}(1997)}]{balke_1997}%
  \BibitemOpen
  \bibfield  {author} {\bibinfo {author} {\bibfnamefont {A.}~\bibnamefont
  {Balke}}\ and\ \bibinfo {author} {\bibfnamefont {J.}~\bibnamefont {Pearl}},\
  }\href {\doibase 10.1080/01621459.1997.10474074} {\bibfield  {journal}
  {\bibinfo  {journal} {J. Am. Stat. Assoc.}\ }\textbf {\bibinfo {volume}
  {92}},\ \bibinfo {pages} {1171} (\bibinfo {year} {1997})}\BibitemShut
  {NoStop}%
\bibitem [{\citenamefont {Hardy}(2001)}]{hardy2001quantum}%
  \BibitemOpen
  \bibfield  {author} {\bibinfo {author} {\bibfnamefont {L.}~\bibnamefont
  {Hardy}},\ }\href@noop {} {\enquote {\bibinfo {title} {Quantum theory from
  five reasonable axioms},}\ } (\bibinfo {year} {2001}),\ \Eprint
  {http://arxiv.org/abs/quant-ph/0101012} {arXiv:quant-ph/0101012 [quant-ph]}
  \BibitemShut {NoStop}%
\bibitem [{\citenamefont {Barrett}(2007)}]{Barrett2007}%
  \BibitemOpen
  \bibfield  {author} {\bibinfo {author} {\bibfnamefont {J.}~\bibnamefont
  {Barrett}},\ }\href {\doibase 10.1103/physreva.75.032304} {\bibfield
  {journal} {\bibinfo  {journal} {Physical Review A}\ }\textbf {\bibinfo
  {volume} {75}} (\bibinfo {year} {2007}),\
  10.1103/physreva.75.032304}\BibitemShut {NoStop}%
\bibitem [{\citenamefont {Chiribella}\ \emph {et~al.}(2010)\citenamefont
  {Chiribella}, \citenamefont {D’Ariano},\ and\ \citenamefont
  {Perinotti}}]{chiribella2010probabilistic}%
  \BibitemOpen
  \bibfield  {author} {\bibinfo {author} {\bibfnamefont {G.}~\bibnamefont
  {Chiribella}}, \bibinfo {author} {\bibfnamefont {G.~M.}\ \bibnamefont
  {D’Ariano}}, \ and\ \bibinfo {author} {\bibfnamefont {P.}~\bibnamefont
  {Perinotti}},\ }\href@noop {} {\bibfield  {journal} {\bibinfo  {journal}
  {Physical Review A}\ }\textbf {\bibinfo {volume} {81}},\ \bibinfo {pages}
  {062348} (\bibinfo {year} {2010})}\BibitemShut {NoStop}%
\bibitem [{\citenamefont {Barnum}\ and\ \citenamefont
  {Wilce}(2011)}]{Barnum2011}%
  \BibitemOpen
  \bibfield  {author} {\bibinfo {author} {\bibfnamefont {H.}~\bibnamefont
  {Barnum}}\ and\ \bibinfo {author} {\bibfnamefont {A.}~\bibnamefont {Wilce}},\
  }\href {\doibase 10.1016/j.entcs.2011.01.002} {\bibfield  {journal} {\bibinfo
   {journal} {Electronic Notes in Theoretical Computer Science}\ }\textbf
  {\bibinfo {volume} {270}},\ \bibinfo {pages} {3–15} (\bibinfo {year}
  {2011})}\BibitemShut {NoStop}%
\bibitem [{\citenamefont {Chiribella}\ \emph {et~al.}(2016)\citenamefont
  {Chiribella}, \citenamefont {D’Ariano},\ and\ \citenamefont
  {Perinotti}}]{chiribella2016quantum}%
  \BibitemOpen
  \bibfield  {author} {\bibinfo {author} {\bibfnamefont {G.}~\bibnamefont
  {Chiribella}}, \bibinfo {author} {\bibfnamefont {G.~M.}\ \bibnamefont
  {D’Ariano}}, \ and\ \bibinfo {author} {\bibfnamefont {P.}~\bibnamefont
  {Perinotti}},\ }\href@noop {} {\bibfield  {journal} {\bibinfo  {journal}
  {Quantum theory: informational foundations and foils}\ ,\ \bibinfo {pages}
  {171}} (\bibinfo {year} {2016})}\BibitemShut {NoStop}%
\bibitem [{\citenamefont {D’Ariano}\ \emph {et~al.}(2016)\citenamefont
  {D’Ariano}, \citenamefont {Chiribella},\ and\ \citenamefont
  {Perinotti}}]{DAriano2016}%
  \BibitemOpen
  \bibfield  {author} {\bibinfo {author} {\bibfnamefont {G.~M.}\ \bibnamefont
  {D’Ariano}}, \bibinfo {author} {\bibfnamefont {G.}~\bibnamefont
  {Chiribella}}, \ and\ \bibinfo {author} {\bibfnamefont {P.}~\bibnamefont
  {Perinotti}},\ }\href {\doibase 10.1017/9781107338340} {\emph {\bibinfo
  {title} {Quantum Theory from First Principles: An Informational Approach}}}\
  (\bibinfo  {publisher} {Cambridge University Press},\ \bibinfo {year}
  {2016})\BibitemShut {NoStop}%
\bibitem [{\citenamefont {Plávala}(2023)}]{Plvala2023}%
  \BibitemOpen
  \bibfield  {author} {\bibinfo {author} {\bibfnamefont {M.}~\bibnamefont
  {Plávala}},\ }\href {\doibase 10.1016/j.physrep.2023.09.001} {\bibfield
  {journal} {\bibinfo  {journal} {Physics Reports}\ }\textbf {\bibinfo {volume}
  {1033}},\ \bibinfo {pages} {1–64} (\bibinfo {year} {2023})}\BibitemShut
  {NoStop}%
\bibitem [{\citenamefont {Holevo}(1973)}]{Holevo73}%
  \BibitemOpen
  \bibfield  {author} {\bibinfo {author} {\bibfnamefont {A.}~\bibnamefont
  {Holevo}},\ }\href {http://mi.mathnet.ru/ppi903} {\bibfield  {journal}
  {\bibinfo  {journal} {Problems Inform. Transmission}\ }\textbf {\bibinfo
  {volume} {9}},\ \bibinfo {pages} {177} (\bibinfo {year} {1973})}\BibitemShut
  {NoStop}%
\bibitem [{\citenamefont {Holevo}(1998)}]{Holevo1998}%
  \BibitemOpen
  \bibfield  {author} {\bibinfo {author} {\bibfnamefont {A.}~\bibnamefont
  {Holevo}},\ }\href {\doibase 10.1109/18.651037} {\bibfield  {journal}
  {\bibinfo  {journal} {IEEE Transactions on Information Theory}\ }\textbf
  {\bibinfo {volume} {44}},\ \bibinfo {pages} {269–273} (\bibinfo {year}
  {1998})}\BibitemShut {NoStop}%
\bibitem [{\citenamefont {Schumacher}\ and\ \citenamefont
  {Westmoreland}(1997)}]{Schumacher1997}%
  \BibitemOpen
  \bibfield  {author} {\bibinfo {author} {\bibfnamefont {B.}~\bibnamefont
  {Schumacher}}\ and\ \bibinfo {author} {\bibfnamefont {M.~D.}\ \bibnamefont
  {Westmoreland}},\ }\href {\doibase 10.1103/physreva.56.131} {\bibfield
  {journal} {\bibinfo  {journal} {Physical Review A}\ }\textbf {\bibinfo
  {volume} {56}},\ \bibinfo {pages} {131–138} (\bibinfo {year}
  {1997})}\BibitemShut {NoStop}%
\bibitem [{\citenamefont {Helstrom}(1976)}]{helstrom_quantum_1976}%
  \BibitemOpen
  \bibfield  {author} {\bibinfo {author} {\bibfnamefont {C.~W.}\ \bibnamefont
  {Helstrom}},\ }\href@noop {} {\emph {\bibinfo {title} {Quantum detection and
  estimation theory}}},\ \bibinfo {series} {Mathematics in science and
  engineering}\ No.\ \bibinfo {number} {v. 123}\ (\bibinfo  {publisher}
  {Academic Press},\ \bibinfo {address} {New York},\ \bibinfo {year}
  {1976})\BibitemShut {NoStop}%
\bibitem [{\citenamefont {Fawzi}\ and\ \citenamefont
  {Renner}(2015)}]{fawzi2015quantum}%
  \BibitemOpen
  \bibfield  {author} {\bibinfo {author} {\bibfnamefont {O.}~\bibnamefont
  {Fawzi}}\ and\ \bibinfo {author} {\bibfnamefont {R.}~\bibnamefont {Renner}},\
  }\href@noop {} {\bibfield  {journal} {\bibinfo  {journal} {Communications in
  Mathematical Physics}\ }\textbf {\bibinfo {volume} {340}},\ \bibinfo {pages}
  {575} (\bibinfo {year} {2015})}\BibitemShut {NoStop}%
\bibitem [{\citenamefont {Junge}\ \emph {et~al.}(2018)\citenamefont {Junge},
  \citenamefont {Renner}, \citenamefont {Sutter}, \citenamefont {Wilde},\ and\
  \citenamefont {Winter}}]{Junge_2018}%
  \BibitemOpen
  \bibfield  {author} {\bibinfo {author} {\bibfnamefont {M.}~\bibnamefont
  {Junge}}, \bibinfo {author} {\bibfnamefont {R.}~\bibnamefont {Renner}},
  \bibinfo {author} {\bibfnamefont {D.}~\bibnamefont {Sutter}}, \bibinfo
  {author} {\bibfnamefont {M.~M.}\ \bibnamefont {Wilde}}, \ and\ \bibinfo
  {author} {\bibfnamefont {A.}~\bibnamefont {Winter}},\ }\href {\doibase
  10.1007/s00023-018-0716-0} {\bibfield  {journal} {\bibinfo  {journal}
  {Annales Henri Poincaré}\ }\textbf {\bibinfo {volume} {19}},\ \bibinfo
  {pages} {2955–2978} (\bibinfo {year} {2018})}\BibitemShut {NoStop}%
\bibitem [{\citenamefont {Blume-Kohout}\ \emph {et~al.}(2010)\citenamefont
  {Blume-Kohout}, \citenamefont {Ng}, \citenamefont {Poulin},\ and\
  \citenamefont {Viola}}]{Blume-kohut_IPS_2010}%
  \BibitemOpen
  \bibfield  {author} {\bibinfo {author} {\bibfnamefont {R.}~\bibnamefont
  {Blume-Kohout}}, \bibinfo {author} {\bibfnamefont {H.~K.}\ \bibnamefont
  {Ng}}, \bibinfo {author} {\bibfnamefont {D.}~\bibnamefont {Poulin}}, \ and\
  \bibinfo {author} {\bibfnamefont {L.}~\bibnamefont {Viola}},\ }\href
  {\doibase 10.1103/PhysRevA.82.062306} {\bibfield  {journal} {\bibinfo
  {journal} {Phys. Rev. A}\ }\textbf {\bibinfo {volume} {82}},\ \bibinfo
  {pages} {062306} (\bibinfo {year} {2010})}\BibitemShut {NoStop}%
\bibitem [{\citenamefont {Coffman}\ \emph {et~al.}(2000)\citenamefont
  {Coffman}, \citenamefont {Kundu},\ and\ \citenamefont
  {Wootters}}]{coffman2000distributed}%
  \BibitemOpen
  \bibfield  {author} {\bibinfo {author} {\bibfnamefont {V.}~\bibnamefont
  {Coffman}}, \bibinfo {author} {\bibfnamefont {J.}~\bibnamefont {Kundu}}, \
  and\ \bibinfo {author} {\bibfnamefont {W.~K.}\ \bibnamefont {Wootters}},\
  }\href@noop {} {\bibfield  {journal} {\bibinfo  {journal} {Physical Review
  A}\ }\textbf {\bibinfo {volume} {61}},\ \bibinfo {pages} {052306} (\bibinfo
  {year} {2000})}\BibitemShut {NoStop}%
\bibitem [{\citenamefont {Osborne}\ and\ \citenamefont
  {Verstraete}(2006)}]{osborne2006general}%
  \BibitemOpen
  \bibfield  {author} {\bibinfo {author} {\bibfnamefont {T.~J.}\ \bibnamefont
  {Osborne}}\ and\ \bibinfo {author} {\bibfnamefont {F.}~\bibnamefont
  {Verstraete}},\ }\href@noop {} {\bibfield  {journal} {\bibinfo  {journal}
  {Physical Review Letters}\ }\textbf {\bibinfo {volume} {96}},\ \bibinfo
  {pages} {220503} (\bibinfo {year} {2006})}\BibitemShut {NoStop}%
\bibitem [{\citenamefont {Wootters}\ and\ \citenamefont
  {Zurek}(1982)}]{wootters1982single}%
  \BibitemOpen
  \bibfield  {author} {\bibinfo {author} {\bibfnamefont {W.~K.}\ \bibnamefont
  {Wootters}}\ and\ \bibinfo {author} {\bibfnamefont {W.~H.}\ \bibnamefont
  {Zurek}},\ }\href@noop {} {\bibfield  {journal} {\bibinfo  {journal}
  {Nature}\ }\textbf {\bibinfo {volume} {299}},\ \bibinfo {pages} {802}
  (\bibinfo {year} {1982})}\BibitemShut {NoStop}%
\bibitem [{\citenamefont {Dieks}(1982)}]{dieks1982communication}%
  \BibitemOpen
  \bibfield  {author} {\bibinfo {author} {\bibfnamefont {D.}~\bibnamefont
  {Dieks}},\ }\href@noop {} {\bibfield  {journal} {\bibinfo  {journal} {Physics
  Letters A}\ }\textbf {\bibinfo {volume} {92}},\ \bibinfo {pages} {271}
  (\bibinfo {year} {1982})}\BibitemShut {NoStop}%
\bibitem [{\citenamefont {Barnum}\ \emph {et~al.}(1996)\citenamefont {Barnum},
  \citenamefont {Caves}, \citenamefont {Fuchs}, \citenamefont {Jozsa},\ and\
  \citenamefont {Schumacher}}]{barnum1996noncommuting}%
  \BibitemOpen
  \bibfield  {author} {\bibinfo {author} {\bibfnamefont {H.}~\bibnamefont
  {Barnum}}, \bibinfo {author} {\bibfnamefont {C.~M.}\ \bibnamefont {Caves}},
  \bibinfo {author} {\bibfnamefont {C.~A.}\ \bibnamefont {Fuchs}}, \bibinfo
  {author} {\bibfnamefont {R.}~\bibnamefont {Jozsa}}, \ and\ \bibinfo {author}
  {\bibfnamefont {B.}~\bibnamefont {Schumacher}},\ }\href@noop {} {\bibfield
  {journal} {\bibinfo  {journal} {Physical Review Letters}\ }\textbf {\bibinfo
  {volume} {76}},\ \bibinfo {pages} {2818} (\bibinfo {year}
  {1996})}\BibitemShut {NoStop}%
\bibitem [{\citenamefont {Barnum}\ \emph
  {et~al.}(2007{\natexlab{a}})\citenamefont {Barnum}, \citenamefont {Barrett},
  \citenamefont {Leifer},\ and\ \citenamefont {Wilce}}]{barnum2007generalized}%
  \BibitemOpen
  \bibfield  {author} {\bibinfo {author} {\bibfnamefont {H.}~\bibnamefont
  {Barnum}}, \bibinfo {author} {\bibfnamefont {J.}~\bibnamefont {Barrett}},
  \bibinfo {author} {\bibfnamefont {M.}~\bibnamefont {Leifer}}, \ and\ \bibinfo
  {author} {\bibfnamefont {A.}~\bibnamefont {Wilce}},\ }\href@noop {}
  {\bibfield  {journal} {\bibinfo  {journal} {Physical review letters}\
  }\textbf {\bibinfo {volume} {99}},\ \bibinfo {pages} {240501} (\bibinfo
  {year} {2007}{\natexlab{a}})}\BibitemShut {NoStop}%
\bibitem [{\citenamefont {Fuchs}\ and\ \citenamefont
  {Peres}(1996)}]{fuchs1996quantum}%
  \BibitemOpen
  \bibfield  {author} {\bibinfo {author} {\bibfnamefont {C.~A.}\ \bibnamefont
  {Fuchs}}\ and\ \bibinfo {author} {\bibfnamefont {A.}~\bibnamefont {Peres}},\
  }\href@noop {} {\bibfield  {journal} {\bibinfo  {journal} {Physical Review
  A}\ }\textbf {\bibinfo {volume} {53}},\ \bibinfo {pages} {2038} (\bibinfo
  {year} {1996})}\BibitemShut {NoStop}%
\bibitem [{\citenamefont {Kretschmann}\ \emph {et~al.}(2008)\citenamefont
  {Kretschmann}, \citenamefont {Schlingemann},\ and\ \citenamefont
  {Werner}}]{kretschmann2008information}%
  \BibitemOpen
  \bibfield  {author} {\bibinfo {author} {\bibfnamefont {D.}~\bibnamefont
  {Kretschmann}}, \bibinfo {author} {\bibfnamefont {D.}~\bibnamefont
  {Schlingemann}}, \ and\ \bibinfo {author} {\bibfnamefont {R.~F.}\
  \bibnamefont {Werner}},\ }\href@noop {} {\bibfield  {journal} {\bibinfo
  {journal} {IEEE transactions on information theory}\ }\textbf {\bibinfo
  {volume} {54}},\ \bibinfo {pages} {1708} (\bibinfo {year}
  {2008})}\BibitemShut {NoStop}%
\bibitem [{\citenamefont {Aharonov}\ \emph {et~al.}(1990)\citenamefont
  {Aharonov}, \citenamefont {Anandan}, \citenamefont {Popescu},\ and\
  \citenamefont {Vaidman}}]{aharonov1990superpositions}%
  \BibitemOpen
  \bibfield  {author} {\bibinfo {author} {\bibfnamefont {Y.}~\bibnamefont
  {Aharonov}}, \bibinfo {author} {\bibfnamefont {J.}~\bibnamefont {Anandan}},
  \bibinfo {author} {\bibfnamefont {S.}~\bibnamefont {Popescu}}, \ and\
  \bibinfo {author} {\bibfnamefont {L.}~\bibnamefont {Vaidman}},\ }\href@noop
  {} {\bibfield  {journal} {\bibinfo  {journal} {Physical review letters}\
  }\textbf {\bibinfo {volume} {64}},\ \bibinfo {pages} {2965} (\bibinfo {year}
  {1990})}\BibitemShut {NoStop}%
\bibitem [{\citenamefont {Oi}(2003)}]{oi2003interference}%
  \BibitemOpen
  \bibfield  {author} {\bibinfo {author} {\bibfnamefont {D.~K.}\ \bibnamefont
  {Oi}},\ }\href@noop {} {\bibfield  {journal} {\bibinfo  {journal} {Physical
  review letters}\ }\textbf {\bibinfo {volume} {91}},\ \bibinfo {pages}
  {067902} (\bibinfo {year} {2003})}\BibitemShut {NoStop}%
\bibitem [{\citenamefont {{\AA}berg}(2004)}]{aaberg2004operations}%
  \BibitemOpen
  \bibfield  {author} {\bibinfo {author} {\bibfnamefont {J.}~\bibnamefont
  {{\AA}berg}},\ }\href@noop {} {\bibfield  {journal} {\bibinfo  {journal}
  {Physical Review A—Atomic, Molecular, and Optical Physics}\ }\textbf
  {\bibinfo {volume} {70}},\ \bibinfo {pages} {012103} (\bibinfo {year}
  {2004})}\BibitemShut {NoStop}%
\bibitem [{\citenamefont {Gisin}\ \emph {et~al.}(2005)\citenamefont {Gisin},
  \citenamefont {Linden}, \citenamefont {Massar},\ and\ \citenamefont
  {Popescu}}]{gisin2005error}%
  \BibitemOpen
  \bibfield  {author} {\bibinfo {author} {\bibfnamefont {N.}~\bibnamefont
  {Gisin}}, \bibinfo {author} {\bibfnamefont {N.}~\bibnamefont {Linden}},
  \bibinfo {author} {\bibfnamefont {S.}~\bibnamefont {Massar}}, \ and\ \bibinfo
  {author} {\bibfnamefont {S.}~\bibnamefont {Popescu}},\ }\href@noop {}
  {\bibfield  {journal} {\bibinfo  {journal} {Physical Review A—Atomic,
  Molecular, and Optical Physics}\ }\textbf {\bibinfo {volume} {72}},\ \bibinfo
  {pages} {012338} (\bibinfo {year} {2005})}\BibitemShut {NoStop}%
\bibitem [{\citenamefont {Chiribella}\ and\ \citenamefont
  {Kristjánsson}(2019)}]{Chiribella_2019}%
  \BibitemOpen
  \bibfield  {author} {\bibinfo {author} {\bibfnamefont {G.}~\bibnamefont
  {Chiribella}}\ and\ \bibinfo {author} {\bibfnamefont {H.}~\bibnamefont
  {Kristjánsson}},\ }\href {\doibase 10.1098/rspa.2018.0903} {\bibfield
  {journal} {\bibinfo  {journal} {Proceedings of the Royal Society A:
  Mathematical, Physical and Engineering Sciences}\ }\textbf {\bibinfo {volume}
  {475}},\ \bibinfo {pages} {20180903} (\bibinfo {year} {2019})}\BibitemShut
  {NoStop}%
\bibitem [{\citenamefont {Abbott}\ \emph {et~al.}(2020)\citenamefont {Abbott},
  \citenamefont {Wechs}, \citenamefont {Horsman}, \citenamefont {Mhalla},\ and\
  \citenamefont {Branciard}}]{Abbott2020communication}%
  \BibitemOpen
  \bibfield  {author} {\bibinfo {author} {\bibfnamefont {A.~A.}\ \bibnamefont
  {Abbott}}, \bibinfo {author} {\bibfnamefont {J.}~\bibnamefont {Wechs}},
  \bibinfo {author} {\bibfnamefont {D.}~\bibnamefont {Horsman}}, \bibinfo
  {author} {\bibfnamefont {M.}~\bibnamefont {Mhalla}}, \ and\ \bibinfo {author}
  {\bibfnamefont {C.}~\bibnamefont {Branciard}},\ }\href {\doibase
  10.22331/q-2020-09-24-333} {\bibfield  {journal} {\bibinfo  {journal}
  {{Quantum}}\ }\textbf {\bibinfo {volume} {4}},\ \bibinfo {pages} {333}
  (\bibinfo {year} {2020})}\BibitemShut {NoStop}%
\bibitem [{\citenamefont {Kristj\'ansson}\ \emph {et~al.}(2021)\citenamefont
  {Kristj\'ansson}, \citenamefont {Mao},\ and\ \citenamefont
  {Chiribella}}]{witnessing_latent_time}%
  \BibitemOpen
  \bibfield  {author} {\bibinfo {author} {\bibfnamefont {H.}~\bibnamefont
  {Kristj\'ansson}}, \bibinfo {author} {\bibfnamefont {W.}~\bibnamefont {Mao}},
  \ and\ \bibinfo {author} {\bibfnamefont {G.}~\bibnamefont {Chiribella}},\
  }\href {\doibase 10.1103/PhysRevResearch.3.043147} {\bibfield  {journal}
  {\bibinfo  {journal} {Phys. Rev. Res.}\ }\textbf {\bibinfo {volume} {3}},\
  \bibinfo {pages} {043147} (\bibinfo {year} {2021})}\BibitemShut {NoStop}%
\bibitem [{\citenamefont {Holland}(1988)}]{Holland1988}%
  \BibitemOpen
  \bibfield  {author} {\bibinfo {author} {\bibfnamefont {P.~W.}\ \bibnamefont
  {Holland}},\ }\href {\doibase 10.2307/271055} {\bibfield  {journal} {\bibinfo
   {journal} {Sociological Methodology}\ }\textbf {\bibinfo {volume} {18}},\
  \bibinfo {pages} {449} (\bibinfo {year} {1988})}\BibitemShut {NoStop}%
\bibitem [{\citenamefont {Marvian}\ and\ \citenamefont
  {Spekkens}(2013)}]{Marvian2013}%
  \BibitemOpen
  \bibfield  {author} {\bibinfo {author} {\bibfnamefont {I.}~\bibnamefont
  {Marvian}}\ and\ \bibinfo {author} {\bibfnamefont {R.~W.}\ \bibnamefont
  {Spekkens}},\ }\href {\doibase 10.1088/1367-2630/15/3/033001} {\bibfield
  {journal} {\bibinfo  {journal} {New Journal of Physics}\ }\textbf {\bibinfo
  {volume} {15}},\ \bibinfo {pages} {033001} (\bibinfo {year}
  {2013})}\BibitemShut {NoStop}%
\bibitem [{\citenamefont {Marvian}\ and\ \citenamefont
  {Spekkens}(2014{\natexlab{a}})}]{Marvian2014}%
  \BibitemOpen
  \bibfield  {author} {\bibinfo {author} {\bibfnamefont {I.}~\bibnamefont
  {Marvian}}\ and\ \bibinfo {author} {\bibfnamefont {R.~W.}\ \bibnamefont
  {Spekkens}},\ }\href {\doibase 10.1038/ncomms4821} {\bibfield  {journal}
  {\bibinfo  {journal} {Nature Communications}\ }\textbf {\bibinfo {volume}
  {5}} (\bibinfo {year} {2014}{\natexlab{a}}),\ 10.1038/ncomms4821}\BibitemShut
  {NoStop}%
\bibitem [{\citenamefont {Marvian}\ and\ \citenamefont
  {Spekkens}(2014{\natexlab{b}})}]{Marvian2014_modes}%
  \BibitemOpen
  \bibfield  {author} {\bibinfo {author} {\bibfnamefont {I.}~\bibnamefont
  {Marvian}}\ and\ \bibinfo {author} {\bibfnamefont {R.~W.}\ \bibnamefont
  {Spekkens}},\ }\href {\doibase 10.1103/physreva.90.062110} {\bibfield
  {journal} {\bibinfo  {journal} {Physical Review A}\ }\textbf {\bibinfo
  {volume} {90}} (\bibinfo {year} {2014}{\natexlab{b}}),\
  10.1103/physreva.90.062110}\BibitemShut {NoStop}%
\bibitem [{\citenamefont {Marvian}\ \emph {et~al.}(2016)\citenamefont
  {Marvian}, \citenamefont {Spekkens},\ and\ \citenamefont
  {Zanardi}}]{Marvian2016}%
  \BibitemOpen
  \bibfield  {author} {\bibinfo {author} {\bibfnamefont {I.}~\bibnamefont
  {Marvian}}, \bibinfo {author} {\bibfnamefont {R.~W.}\ \bibnamefont
  {Spekkens}}, \ and\ \bibinfo {author} {\bibfnamefont {P.}~\bibnamefont
  {Zanardi}},\ }\href {\doibase 10.1103/physreva.93.052331} {\bibfield
  {journal} {\bibinfo  {journal} {Physical Review A}\ }\textbf {\bibinfo
  {volume} {93}} (\bibinfo {year} {2016}),\
  10.1103/physreva.93.052331}\BibitemShut {NoStop}%
\bibitem [{\citenamefont {Winter}(2017)}]{winter2017energyconstrained}%
  \BibitemOpen
  \bibfield  {author} {\bibinfo {author} {\bibfnamefont {A.}~\bibnamefont
  {Winter}},\ }\href {https://arxiv.org/abs/1712.10267} {\enquote {\bibinfo
  {title} {Energy-constrained diamond norm with applications to the uniform
  continuity of continuous variable channel capacities},}\ } (\bibinfo {year}
  {2017}),\ \Eprint {http://arxiv.org/abs/1712.10267} {arXiv:1712.10267
  [quant-ph]} \BibitemShut {NoStop}%
\bibitem [{\citenamefont {Shirokov}(2017)}]{Shirokov2017}%
  \BibitemOpen
  \bibfield  {author} {\bibinfo {author} {\bibfnamefont {M.~E.}\ \bibnamefont
  {Shirokov}},\ }\href {\doibase 10.1063/1.4987135} {\bibfield  {journal}
  {\bibinfo  {journal} {Journal of Mathematical Physics}\ }\textbf {\bibinfo
  {volume} {58}},\ \bibinfo {pages} {102202} (\bibinfo {year}
  {2017})}\BibitemShut {NoStop}%
\bibitem [{\citenamefont {Sharma}\ \emph {et~al.}(2018)\citenamefont {Sharma},
  \citenamefont {Wilde}, \citenamefont {Adhikari},\ and\ \citenamefont
  {Takeoka}}]{Sharma2018}%
  \BibitemOpen
  \bibfield  {author} {\bibinfo {author} {\bibfnamefont {K.}~\bibnamefont
  {Sharma}}, \bibinfo {author} {\bibfnamefont {M.~M.}\ \bibnamefont {Wilde}},
  \bibinfo {author} {\bibfnamefont {S.}~\bibnamefont {Adhikari}}, \ and\
  \bibinfo {author} {\bibfnamefont {M.}~\bibnamefont {Takeoka}},\ }\href
  {\doibase 10.1088/1367-2630/aac11a} {\bibfield  {journal} {\bibinfo
  {journal} {New Journal of Physics}\ }\textbf {\bibinfo {volume} {20}},\
  \bibinfo {pages} {063025} (\bibinfo {year} {2018})}\BibitemShut {NoStop}%
\bibitem [{\citenamefont {Becker}\ \emph {et~al.}(2021)\citenamefont {Becker},
  \citenamefont {Datta}, \citenamefont {Lami},\ and\ \citenamefont
  {Rouzé}}]{Becker2021}%
  \BibitemOpen
  \bibfield  {author} {\bibinfo {author} {\bibfnamefont {S.}~\bibnamefont
  {Becker}}, \bibinfo {author} {\bibfnamefont {N.}~\bibnamefont {Datta}},
  \bibinfo {author} {\bibfnamefont {L.}~\bibnamefont {Lami}}, \ and\ \bibinfo
  {author} {\bibfnamefont {C.}~\bibnamefont {Rouzé}},\ }\href {\doibase
  10.1103/physrevlett.126.190504} {\bibfield  {journal} {\bibinfo  {journal}
  {Physical Review Letters}\ }\textbf {\bibinfo {volume} {126}} (\bibinfo
  {year} {2021}),\ 10.1103/physrevlett.126.190504}\BibitemShut {NoStop}%
\bibitem [{\citenamefont {Hirche}\ \emph {et~al.}(2023)\citenamefont {Hirche},
  \citenamefont {Rouzé},\ and\ \citenamefont {França}}]{hirche2023quantum}%
  \BibitemOpen
  \bibfield  {author} {\bibinfo {author} {\bibfnamefont {C.}~\bibnamefont
  {Hirche}}, \bibinfo {author} {\bibfnamefont {C.}~\bibnamefont {Rouzé}}, \
  and\ \bibinfo {author} {\bibfnamefont {D.~S.}\ \bibnamefont {França}},\
  }\href@noop {} {\enquote {\bibinfo {title} {Quantum differential privacy: An
  information theory perspective},}\ } (\bibinfo {year} {2023}),\ \Eprint
  {http://arxiv.org/abs/2202.10717} {arXiv:2202.10717 [quant-ph]} \BibitemShut
  {NoStop}%
\bibitem [{\citenamefont {Schmid}\ \emph {et~al.}(2019)\citenamefont {Schmid},
  \citenamefont {Ried},\ and\ \citenamefont {Spekkens}}]{schmid2019initial}%
  \BibitemOpen
  \bibfield  {author} {\bibinfo {author} {\bibfnamefont {D.}~\bibnamefont
  {Schmid}}, \bibinfo {author} {\bibfnamefont {K.}~\bibnamefont {Ried}}, \ and\
  \bibinfo {author} {\bibfnamefont {R.~W.}\ \bibnamefont {Spekkens}},\
  }\href@noop {} {\bibfield  {journal} {\bibinfo  {journal} {Physical Review
  A}\ }\textbf {\bibinfo {volume} {100}},\ \bibinfo {pages} {022112} (\bibinfo
  {year} {2019})}\BibitemShut {NoStop}%
\bibitem [{\citenamefont {Fuchs}(1997)}]{Fuchs1997}%
  \BibitemOpen
  \bibfield  {author} {\bibinfo {author} {\bibfnamefont {C.~A.}\ \bibnamefont
  {Fuchs}},\ }\href {\doibase 10.1103/physrevlett.79.1162} {\bibfield
  {journal} {\bibinfo  {journal} {Physical Review Letters}\ }\textbf {\bibinfo
  {volume} {79}},\ \bibinfo {pages} {1162–1165} (\bibinfo {year}
  {1997})}\BibitemShut {NoStop}%
\bibitem [{\citenamefont {Kretschmann}\ \emph {et~al.}(2007)\citenamefont
  {Kretschmann}, \citenamefont {Schlingemann},\ and\ \citenamefont
  {Werner}}]{kretschmann_continuity_2007}%
  \BibitemOpen
  \bibfield  {author} {\bibinfo {author} {\bibfnamefont {D.}~\bibnamefont
  {Kretschmann}}, \bibinfo {author} {\bibfnamefont {D.}~\bibnamefont
  {Schlingemann}}, \ and\ \bibinfo {author} {\bibfnamefont {R.~F.}\
  \bibnamefont {Werner}},\ }\href {\doibase 10.48550/ARXIV.0710.2495} {\
  (\bibinfo {year} {2007}),\ 10.48550/ARXIV.0710.2495},\ \bibinfo {note}
  {publisher: arXiv Version Number: 1}\BibitemShut {NoStop}%
\bibitem [{\citenamefont {Hayden}\ and\ \citenamefont
  {Penington}(2020)}]{Hayden_2020}%
  \BibitemOpen
  \bibfield  {author} {\bibinfo {author} {\bibfnamefont {P.}~\bibnamefont
  {Hayden}}\ and\ \bibinfo {author} {\bibfnamefont {G.}~\bibnamefont
  {Penington}},\ }\href {\doibase 10.1007/s00220-020-03689-1} {\bibfield
  {journal} {\bibinfo  {journal} {Commun. Math. Phys.}\ }\textbf {\bibinfo
  {volume} {374}},\ \bibinfo {pages} {369–432} (\bibinfo {year}
  {2020})}\BibitemShut {NoStop}%
\bibitem [{\citenamefont {Rubino}\ \emph {et~al.}(2021)\citenamefont {Rubino},
  \citenamefont {Rozema}, \citenamefont {Ebler}, \citenamefont
  {Kristj{\'a}nsson}, \citenamefont {Salek}, \citenamefont {Allard~Gu{\'e}rin},
  \citenamefont {Abbott}, \citenamefont {Branciard}, \citenamefont {Brukner},
  \citenamefont {Chiribella} \emph {et~al.}}]{rubino2021experimental}%
  \BibitemOpen
  \bibfield  {author} {\bibinfo {author} {\bibfnamefont {G.}~\bibnamefont
  {Rubino}}, \bibinfo {author} {\bibfnamefont {L.~A.}\ \bibnamefont {Rozema}},
  \bibinfo {author} {\bibfnamefont {D.}~\bibnamefont {Ebler}}, \bibinfo
  {author} {\bibfnamefont {H.}~\bibnamefont {Kristj{\'a}nsson}}, \bibinfo
  {author} {\bibfnamefont {S.}~\bibnamefont {Salek}}, \bibinfo {author}
  {\bibfnamefont {P.}~\bibnamefont {Allard~Gu{\'e}rin}}, \bibinfo {author}
  {\bibfnamefont {A.~A.}\ \bibnamefont {Abbott}}, \bibinfo {author}
  {\bibfnamefont {C.}~\bibnamefont {Branciard}}, \bibinfo {author}
  {\bibfnamefont {{\v{C}}.}~\bibnamefont {Brukner}}, \bibinfo {author}
  {\bibfnamefont {G.}~\bibnamefont {Chiribella}},  \emph {et~al.},\ }\href@noop
  {} {\bibfield  {journal} {\bibinfo  {journal} {Physical Review Research}\
  }\textbf {\bibinfo {volume} {3}},\ \bibinfo {pages} {013093} (\bibinfo {year}
  {2021})}\BibitemShut {NoStop}%
\bibitem [{\citenamefont {Chen}\ \emph {et~al.}(2021)\citenamefont {Chen},
  \citenamefont {Song}, \citenamefont {Zhao},\ and\ \citenamefont
  {Wang}}]{Chen_2021}%
  \BibitemOpen
  \bibfield  {author} {\bibinfo {author} {\bibfnamefont {R.}~\bibnamefont
  {Chen}}, \bibinfo {author} {\bibfnamefont {Z.}~\bibnamefont {Song}}, \bibinfo
  {author} {\bibfnamefont {X.}~\bibnamefont {Zhao}}, \ and\ \bibinfo {author}
  {\bibfnamefont {X.}~\bibnamefont {Wang}},\ }\href {\doibase
  10.1088/2058-9565/ac38ba} {\bibfield  {journal} {\bibinfo  {journal} {Quantum
  Sci. Technol.}\ }\textbf {\bibinfo {volume} {7}},\ \bibinfo {pages} {015019}
  (\bibinfo {year} {2021})}\BibitemShut {NoStop}%
\bibitem [{\citenamefont {Rethinasamy}\ \emph {et~al.}(2023)\citenamefont
  {Rethinasamy}, \citenamefont {Agarwal}, \citenamefont {Sharma},\ and\
  \citenamefont {Wilde}}]{Rethinasamy_2023}%
  \BibitemOpen
  \bibfield  {author} {\bibinfo {author} {\bibfnamefont {S.}~\bibnamefont
  {Rethinasamy}}, \bibinfo {author} {\bibfnamefont {R.}~\bibnamefont
  {Agarwal}}, \bibinfo {author} {\bibfnamefont {K.}~\bibnamefont {Sharma}}, \
  and\ \bibinfo {author} {\bibfnamefont {M.~M.}\ \bibnamefont {Wilde}},\ }\href
  {\doibase 10.1103/physreva.108.012409} {\bibfield  {journal} {\bibinfo
  {journal} {Phys. Rev. A}\ }\textbf {\bibinfo {volume} {108}} (\bibinfo {year}
  {2023}),\ 10.1103/physreva.108.012409}\BibitemShut {NoStop}%
\bibitem [{\citenamefont {Schuld}\ \emph {et~al.}(2020)\citenamefont {Schuld},
  \citenamefont {Bocharov}, \citenamefont {Svore},\ and\ \citenamefont
  {Wiebe}}]{schuld_circuit-centric_2020}%
  \BibitemOpen
  \bibfield  {author} {\bibinfo {author} {\bibfnamefont {M.}~\bibnamefont
  {Schuld}}, \bibinfo {author} {\bibfnamefont {A.}~\bibnamefont {Bocharov}},
  \bibinfo {author} {\bibfnamefont {K.~M.}\ \bibnamefont {Svore}}, \ and\
  \bibinfo {author} {\bibfnamefont {N.}~\bibnamefont {Wiebe}},\ }\href
  {\doibase 10.1103/PhysRevA.101.032308} {\bibfield  {journal} {\bibinfo
  {journal} {Phys. Rev. A}\ }\textbf {\bibinfo {volume} {101}},\ \bibinfo
  {pages} {032308} (\bibinfo {year} {2020})}\BibitemShut {NoStop}%
\bibitem [{\citenamefont {Goswami}\ and\ \citenamefont
  {Chiribella}(2024)}]{Goswami_github}%
  \BibitemOpen
  \bibfield  {author} {\bibinfo {author} {\bibfnamefont {K.}~\bibnamefont
  {Goswami}}\ and\ \bibinfo {author} {\bibfnamefont {G.}~\bibnamefont
  {Chiribella}},\ }\href {https://github.com/KBG0603/Quantum-causal-effect.git}
  {\enquote {\bibinfo {title} {Quantum-causal-effect},}\ }\bibinfo
  {howpublished} {Github Page} (\bibinfo {year} {2024})\BibitemShut {NoStop}%
\bibitem [{\citenamefont {Barnum}\ \emph
  {et~al.}(2007{\natexlab{b}})\citenamefont {Barnum}, \citenamefont {Barrett},
  \citenamefont {Leifer},\ and\ \citenamefont {Wilce}}]{Barnum_2007}%
  \BibitemOpen
  \bibfield  {author} {\bibinfo {author} {\bibfnamefont {H.}~\bibnamefont
  {Barnum}}, \bibinfo {author} {\bibfnamefont {J.}~\bibnamefont {Barrett}},
  \bibinfo {author} {\bibfnamefont {M.}~\bibnamefont {Leifer}}, \ and\ \bibinfo
  {author} {\bibfnamefont {A.}~\bibnamefont {Wilce}},\ }\href {\doibase
  10.1103/physrevlett.99.240501} {\bibfield  {journal} {\bibinfo  {journal}
  {Physical Review Letters}\ }\textbf {\bibinfo {volume} {99}} (\bibinfo {year}
  {2007}{\natexlab{b}}),\ 10.1103/physrevlett.99.240501}\BibitemShut {NoStop}%
\bibitem [{\citenamefont {Kimura}\ \emph {et~al.}(2010)\citenamefont {Kimura},
  \citenamefont {Nuida},\ and\ \citenamefont {Imai}}]{Kimura2010}%
  \BibitemOpen
  \bibfield  {author} {\bibinfo {author} {\bibfnamefont {G.}~\bibnamefont
  {Kimura}}, \bibinfo {author} {\bibfnamefont {K.}~\bibnamefont {Nuida}}, \
  and\ \bibinfo {author} {\bibfnamefont {H.}~\bibnamefont {Imai}},\ }\href
  {\doibase 10.1016/s0034-4877(10)00025-x} {\bibfield  {journal} {\bibinfo
  {journal} {Reports on Mathematical Physics}\ }\textbf {\bibinfo {volume}
  {66}},\ \bibinfo {pages} {175–206} (\bibinfo {year} {2010})}\BibitemShut
  {NoStop}%
\bibitem [{\citenamefont {{Thomas M. Cover}}\ and\ \citenamefont
  {Thomas}(2006)}]{Thomas_M_Cover2006-sa}%
  \BibitemOpen
  \bibfield  {author} {\bibinfo {author} {\bibnamefont {{Thomas M. Cover}}}\
  and\ \bibinfo {author} {\bibfnamefont {J.~A.}\ \bibnamefont {Thomas}},\
  }\href@noop {} {{\selectlanguage {en}\emph {\bibinfo {title} {Elements of
  Information Theory}}}},\ \bibinfo {edition} {2nd}\ ed.\ (\bibinfo
  {publisher} {John Wiley \& Sons},\ \bibinfo {address} {Nashville, TN},\
  \bibinfo {year} {2006})\BibitemShut {NoStop}%
\bibitem [{\citenamefont {Wilde}(2013)}]{Wilde13}%
  \BibitemOpen
  \bibfield  {author} {\bibinfo {author} {\bibfnamefont {M.~M.}\ \bibnamefont
  {Wilde}},\ }\href {\doibase 10.1017/9781316809976.001} {\bibfield  {journal}
  {\bibinfo  {journal} {Quantum Information Theory}\ ,\ \bibinfo {pages}
  {xi–xii}} (\bibinfo {year} {2013})}\BibitemShut {NoStop}%
\bibitem [{\citenamefont {Jenčová}(2023)}]{jenčová2023recoverability}%
  \BibitemOpen
  \bibfield  {author} {\bibinfo {author} {\bibfnamefont {A.}~\bibnamefont
  {Jenčová}},\ }\href@noop {} {\enquote {\bibinfo {title} {Recoverability of
  quantum channels via hypothesis testing},}\ } (\bibinfo {year} {2023}),\
  \Eprint {http://arxiv.org/abs/2303.11707} {arXiv:2303.11707 [quant-ph]}
  \BibitemShut {NoStop}%
\bibitem [{\citenamefont {Frenkel}(2023)}]{Frenkel_2023}%
  \BibitemOpen
  \bibfield  {author} {\bibinfo {author} {\bibfnamefont {P.~E.}\ \bibnamefont
  {Frenkel}},\ }\href {\doibase 10.22331/q-2023-09-07-1102} {\bibfield
  {journal} {\bibinfo  {journal} {Quantum}\ }\textbf {\bibinfo {volume} {7}},\
  \bibinfo {pages} {1102} (\bibinfo {year} {2023})}\BibitemShut {NoStop}%
\bibitem [{\citenamefont {Umegaki}(1962)}]{Umegaki1962}%
  \BibitemOpen
  \bibfield  {author} {\bibinfo {author} {\bibfnamefont {H.}~\bibnamefont
  {Umegaki}},\ }\href {\doibase 10.2996/kmj/1138844604} {\bibfield  {journal}
  {\bibinfo  {journal} {Kodai Mathematical Journal}\ }\textbf {\bibinfo
  {volume} {14}} (\bibinfo {year} {1962}),\ 10.2996/kmj/1138844604}\BibitemShut
  {NoStop}%
\bibitem [{\citenamefont {Hiai}\ and\ \citenamefont {Petz}(1991)}]{Hiai1991}%
  \BibitemOpen
  \bibfield  {author} {\bibinfo {author} {\bibfnamefont {F.}~\bibnamefont
  {Hiai}}\ and\ \bibinfo {author} {\bibfnamefont {D.}~\bibnamefont {Petz}},\
  }\href {\doibase 10.1007/bf02100287} {\bibfield  {journal} {\bibinfo
  {journal} {Communications in Mathematical Physics}\ }\textbf {\bibinfo
  {volume} {143}},\ \bibinfo {pages} {99–114} (\bibinfo {year}
  {1991})}\BibitemShut {NoStop}%
\bibitem [{\citenamefont {Lindblad}(1975)}]{Lindblad1975}%
  \BibitemOpen
  \bibfield  {author} {\bibinfo {author} {\bibfnamefont {G.}~\bibnamefont
  {Lindblad}},\ }\href {\doibase 10.1007/bf01609396} {\bibfield  {journal}
  {\bibinfo  {journal} {Communications in Mathematical Physics}\ }\textbf
  {\bibinfo {volume} {40}},\ \bibinfo {pages} {147–151} (\bibinfo {year}
  {1975})}\BibitemShut {NoStop}%
\bibitem [{\citenamefont {Uhlmann}(1977)}]{Uhlmann1977}%
  \BibitemOpen
  \bibfield  {author} {\bibinfo {author} {\bibfnamefont {A.}~\bibnamefont
  {Uhlmann}},\ }\href {\doibase 10.1007/bf01609834} {\bibfield  {journal}
  {\bibinfo  {journal} {Communications in Mathematical Physics}\ }\textbf
  {\bibinfo {volume} {54}},\ \bibinfo {pages} {21–32} (\bibinfo {year}
  {1977})}\BibitemShut {NoStop}%
\bibitem [{\citenamefont {Petz}(1986)}]{Petz1986}%
  \BibitemOpen
  \bibfield  {author} {\bibinfo {author} {\bibfnamefont {D.}~\bibnamefont
  {Petz}},\ }\href {\doibase 10.1007/bf01212345} {\bibfield  {journal}
  {\bibinfo  {journal} {Communications in Mathematical Physics}\ }\textbf
  {\bibinfo {volume} {105}},\ \bibinfo {pages} {123–131} (\bibinfo {year}
  {1986})}\BibitemShut {NoStop}%
\bibitem [{\citenamefont {Petz}(1988)}]{PETZ1988}%
  \BibitemOpen
  \bibfield  {author} {\bibinfo {author} {\bibfnamefont {D.}~\bibnamefont
  {Petz}},\ }\href {\doibase 10.1093/qmath/39.1.97} {\bibfield  {journal}
  {\bibinfo  {journal} {The Quarterly Journal of Mathematics}\ }\textbf
  {\bibinfo {volume} {39}},\ \bibinfo {pages} {97–108} (\bibinfo {year}
  {1988})}\BibitemShut {NoStop}%
\bibitem [{\citenamefont {Fuchs}\ and\ \citenamefont {Van
  De~Graaf}(1999)}]{fuchs_cryptographic_1999}%
  \BibitemOpen
  \bibfield  {author} {\bibinfo {author} {\bibfnamefont {C.}~\bibnamefont
  {Fuchs}}\ and\ \bibinfo {author} {\bibfnamefont {J.}~\bibnamefont {Van
  De~Graaf}},\ }\href {\doibase 10.1109/18.761271} {\bibfield  {journal}
  {\bibinfo  {journal} {IEEE Trans. Inform. Theory}\ }\textbf {\bibinfo
  {volume} {45}},\ \bibinfo {pages} {1216} (\bibinfo {year}
  {1999})}\BibitemShut {NoStop}%
\end{thebibliography}%

\newpage
\appendix
\pagebreak

\section{Proof of Proposition~\ref{prop:classicalACE}}\label{app:classicalexpression}

The proof uses the following lemmas.
\begin{lemma}~\label{lem:orthogonal_distributions}
    For any two probability distributions $P$ and $P'$, associated with a sample space  $A=\{a\}_a$, we can find two probability distributions $Q$ and $Q'$ with maximal total variation distance, i.e. $\norm{Q-Q'}_1=2$, such that 
    \begin{align}
    P(a)-P'(a) = \frac{\norm{P-P'}_1}{2}(Q(a)-Q'(a)), \ \ \ \forall a.
\end{align}
\end{lemma}
\begin{proof}
We first construct $Q$ and $Q'$ associated with any two probability distributions $P$ and $P'$. Let us divide the sample space of interventions $A=\{a\}_a$ into two disjoint parts: $A = A_1 + A_2$, such that 
\begin{align}
    &A_1 := \{a: P(a)\geq P'(a)\}, \nonumber \\
    &A_2 := \{a: P(a) < P'(a)\}.
\end{align}
Now we define the probability distributions $Q$ and $Q'$:
\begin{align}
    &Q(a) := \begin{cases}
			\frac{2\left[P(a)-P'(a)\right]}{\norm{P - P'}_1}, & \text{if $a \in A_1$}\\
            0, & \text{if $a \in A_2$}
		 \end{cases} \nonumber \\
   &Q'(a) := \begin{cases}
			0, & \text{if $a \in A_1$}\\
            \frac{2[P'(a)-P(a)]}{\norm{P - P'}_1}. & \text{if $a \in A_2$}
		 \end{cases} \nonumber
\end{align}
One can check that both $Q$ and $Q'$ are valid probability distributions. By construction, we have  $Q(a) \geq 0$ and $Q'(a) \geq 0$ for all intervention $a$. To show the normalization, we note that because $A_1$ and $A_2$ together constitute the entire sample space, we have 
\begin{align}
    \sum_{a\in A_1} (P(a) - P'(a)) = \sum_{a\in A_2} (P'(a) - P(a)),
\end{align}
and
\begin{align}
\norm{P-P'}_1&:=\sum_a\abs{ (P(a)-P'(a))}\nonumber \\
&=\sum_{a\in A_1}(P(a) - P'(a)) + \sum_{a\in A_2}(P'(a) - P(a)) \label{Eq:TVD_0} \\
&=2\sum_{a\in A_1}(P(a) - P'(a)) \label{Eq:TVD_1} \\
&=2\sum_{a\in A_2}(P'(a) - P(a)). \label{Eq:TVD_2} 
\end{align}
Plugging Eqs.\eqref{Eq:TVD_1} and \eqref{Eq:TVD_2} in the definition of $Q(a)$ and $Q'(a)$ respectively, we see $\sum_a Q(a)=\sum_a Q'(a)=1$, in other words they are valid probability distributions. Also, Eq.~\eqref{Eq:TVD_0} immediately makes it clear that $\norm{Q-Q'}_1=2$, i.e., they have maximal total variation distance.  This construction allows us to write 
\begin{align}
    P(a)-P'(a) = \frac{\norm{P-P'}_1}{2}(Q(a)-Q'(a)), \ \ \ \forall a.
\end{align}
\end{proof}

\begin{lemma} \label{lem:classicalACEorhto_2}
    The supremum in the  right-hand-side of Eq. (\ref{almostGPT}) 
    can be restricted without loss of generality to probability distributions having the maximum total variation distance, {\em i.e.},
    \begin{align}
     \sup_{P \not  =    P'}  \frac{  \|   {\cal N}  (P)   -  {\cal N}   (P')\|_1}{\|  P  -  P'\|_1}   = \sup_{Q,Q': \norm{Q-Q'}_1=2} \frac{\norm{\N(Q) - \N(Q')}_1}{2}. 
    \end{align}
    Moreover, if the variable $A$ has a finite set of possible values, then the supremum is a maximum. 
\end{lemma}
\begin{proof}
 First, observe the following chain of equations.
 \begin{align}
     \frac{\norm{\N(P)-\N(P')}_1}{\norm{P-P'}_1}&=\frac{\sum_b \abs{\N(P)(b)-\N(P')(b)}}{\norm{P-P'}_1} \nonumber \\
     &=\frac{\sum_b \abs{\sum_a P(b|\mathrm{do}(a))(P(a)-P'(a))}}{\norm{P-P'}_1} \nonumber \\
     &=\frac{\sum_b \abs{\sum_a P(b|\mathrm{do}(a))(Q(a)-Q'(a))}}{2} \nonumber\\
     &=\frac{\sum_b \abs{\N(Q)(b)-\N(Q')(b)}}{2} \nonumber\\
     &=\frac{\norm{\N(Q) - \N(Q')}_1}{2} \label{Eq:classicalACE_orth}
      \end{align}
Here, in the second equation, we expand $\N(P)(b)$ as in Eq.~\eqref{classicalchannel}. In the third equation, we apply Lemma~\ref{lem:orthogonal_distributions}. Thus, the supremum  over all possible distributions $P$ and $P'$ is equal  to the supremum  over all distributions $Q$ and $Q'$ satisfying $\|  Q-Q'\|_1  =  2$.   

When the variable $A$ has a finite set of values, the space of probability distributions on $A$ is compact and the supremum is actually a maximum:  if  $(Q_n,  Q_n')_{n  \in {\mathbb N}}$ is a sequence of pairs of probability distributions achieving the supremum in the limit, {\em i.e.,} 
\begin{align}
\lim_{n\to \infty} \frac{\norm{\N(Q_n) - \N(Q_n')}_1}{2}  &  =     \sup_{Q,Q': \norm{Q-Q'}_1=2} \frac{\norm{\N(Q) - \N(Q')}_1}{2} \, ,
\end{align}
then   one can extract a subsequence $(Q_{n_k},  Q_{n_k}')_{k  \in {\mathbb N}}$ that converges to a pair of probability distributions $(Q_0, Q_0')$, thereby obtaining  
\begin{align}
\nonumber \frac{\norm{\N(Q_0) - \N(Q_0')}_1}{2}   &=  \lim_{k\to \infty}   \frac{\norm{\N(Q_{n_k}) - \N(Q_{n_k}')}_1}{2}  \\
  &  = \sup_{Q,Q': \norm{Q-Q'}_1=2} \frac{\norm{\N(Q) - \N(Q')}_1}{2} \, .
\end{align}
Hence, the supremum is actually a maximum. 
 \end{proof}

\medskip

\noindent \textbf{Proof of proposition~\ref{prop:classicalACE}.}   Lemma \ref{lem:classicalACEorhto_2}  reduces the supremum over all pairs of distinct probability distributions to a maximum over all pairs $  (Q, Q')$ of  probability distribution with maximum distance $\|  Q  - Q'\|_1  = 2$.    Moreover, the convexity of the norm (implied by the triangle inequality), implies that the maximization can be restricted without loss of generality to   extreme points of the probability simplex, corresponding to probability distributions of the form $Q_{a_0}   (a)   = \delta_{a,  a_0}$ for some fixed value $a_0$.   For an extreme point $Q_{a_0}$, the output probability distribution $\N  (Q_0)$ is  
\begin{align}
\nonumber [\N  (Q_0)]\,  (b)  &=   \sum_a  p(  b|  {\rm do  (a)})  \,  Q_{a_0}  (a) \\
\nonumber &=   \sum_a  p(  b|  {\rm do  (a)})  \,  \delta_{a,  a_0}  \\
\nonumber  &=  p(b|  {\rm do}(a_0))\\
  &  \equiv P_{{\rm do}  (a_0)}    (b) \,.
\end{align}
Hence,  we have 
\begin{align}
\nonumber      \sup_{P \not  =    P'}  \frac{  \|   {\cal N}  (P)   -  {\cal N}   (P')\|_1}{\|  P  -  P'\|_1}   &= \max_{Q,Q': \norm{Q-Q'}_1=2}  \frac{\norm{\N(Q) - \N(Q')}_1}{2}  
\\
   &= \max_{a_0\not  =  a_1} \frac{\norm{  P_{{\rm do}  (a_0) } - P_{{\rm do}  (a_0)} }_1}{2}.
  \end{align}
The proof is then completed by comparing this equality with the definition of ACE [Eq.~\eqref{ACEnormnonbinary}].  \qed

  \section{From the classical ACE to the maximum causal effect in general probabilistic theories}\label{app:foundations}  

Here we provide more background on the use of the ACE in classical theory, and show that the notion of maximum causal effect introduced in this work provides a natural analogue of the ACE in general probabilistic theories.

\subsection{Causation vs correlation in classical theory}
Given two  random variables $A$ and $B$,   a fundamental question  is how to determine whether there is a genuine causal influence from $A$ to $B$, and, in the affirmative case, to quantify the strength of the causal influence. 

  A cause-effect relation  $A \to B$, from Alice's to Bob's variable, generally implies  
  correlations in the joint probability distribution $P(a,b)$ that $A$ has value $a$ and $B$ has value $b$.  The converse is generally not true, as emphasized by the famous motto ``correlation does not imply causation''.  For example,  the correlation between variables $A$ and $B$ could be due to a common cause, that is, an additional random variable $\Lambda$ affecting both $A$ and $B$.  The variable $\Lambda$ is typically regarded as {\em latent}, meaning that it is not directly observable.

\begin{figure}
    \centering
\includegraphics[width=0.8\columnwidth]{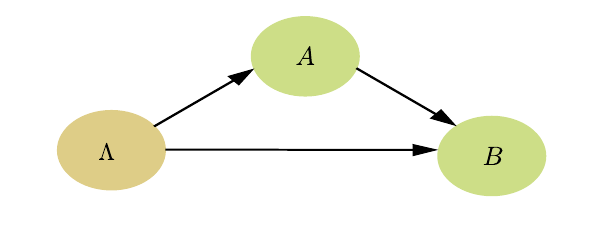}
    \caption{ {\bf  A causal model with both direct and common cause.}   The figure shows the cause-effect relations  (represented by arrows) between two observable variables, $A$ and $B$, and a latent variable $\Lambda$, not accessible to observation.  Here, variable $\Lambda$ acts as a common cause for variables $A$ and $B$, generating correlations between them. The problem is to quantify the amount of  direct causal influence from $A$ to $B$, corresponding to correlations that are not due to the common  cause $\Lambda$.    }
    \label{fig:DAG_quantum}
\end{figure}

  Consider a scenario involving observable variables $A$ and $B$, and a latent common cause $\Lambda$.  This scenario can be mathematically represented as a {\em direct acyclic graph (DAG)}, where the vertices represent random variables and the directed edges represent cause-effect relations, as illustrated in Fig.~\ref{fig:DAG_quantum}.       The joint probability distribution of variables $A$ and $B$ can be decomposed as 
\begin{align}\label{pab}
 P(a,b) = \sum_{\lambda}  P(b|a,\lambda) \, P(a|\lambda)\,  P(\lambda) \, , 
\end{align}
where  $P(\lambda)$ is it probability that the random variable $\Lambda$ takes the value $\lambda$.   

To determine whether variable $A$ affects variable $B$, one needs to consider what would happen if  Alice intervened in her variable, forcing it to assume a certain value $a$, independently of the value of variable $\Lambda$.   
This intervention,  denoted by $\mathrm{do}(a)$, overrides the causal influence from $\Lambda$ to $A$ by forcibly removing the correlations between these two variables; in other words, the probability of performing the intervention ${\rm do} (a)$ is independent of the value of $\Lambda$:
\begin{align}\label{probdo}
P(\mathrm{do}(a)|\lambda)=P(\mathrm{do}(a))  \qquad \forall a ,\ ,\forall\lambda \,.
\end{align}

\begin{figure} 
    \centering
    \includegraphics[width= 0.8\columnwidth]{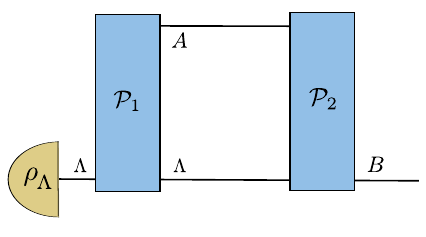}
    \caption{  {\bf Circuit associated to the causal model in Fig.~\ref{fig:DAG_quantum}.}  System $\Lambda$ is initially in a state $\rho_\Lambda$.   A first process, ${\cal P}_1$  transforms system $\Lambda$ into the composite system  $A\otimes \Lambda$, thereby implementing the causal relation $\Lambda \to A$.  Then, a second process ${\cal P}_2$ transforms system $A \otimes \Lambda$ into system $B$, thereby implementing the causal relations $A\to B$ and $\Lambda \to B$. } 
    \label{Fig:circuit}
\end{figure}

Now, if $A$ causes $B$ through a physical mechanism described by the conditional probability distribution $P(b|a,\lambda)$,  forcing $A$ to assume a chosen value before the mechanism acts should not alter the probability of obtaining the value $b$: in formula,
\begin{align}\label{conditionaldo}
p  (b|  {\rm do}  (a)   , \lambda)   =  p(b|a,\lambda)   \qquad \forall a, \, \forall b,  \,\forall \lambda\,.
\end{align}
Eqs. (\ref{probdo}) and (\ref{conditionaldo})  imply that the joint probability distribution of performing the intervention ${\rm do} (a)$ and observing value $b$ is 
\begin{align}
 P(\mathrm{do}(a),b) = \sum_{\lambda} P(b|a,\lambda) \,  P\big(\mathrm{do}(a)\big)\,  P(\lambda)  \, .   
\end{align}
The conditional probability distribution of variable $B$ given a do-intervention on variable $A$ is then 
\begin{align}\label{classicaldochannel}
P(b|  {\rm do (a)})  =  \sum_\lambda  \,      P(b| a, \lambda)  \,   P(\lambda)  \, . 
\end{align}

By definition, the probability distribution $p(b|  {\rm do} (a))$ depends only on the observable variables $a$ and $b$.

\subsection{From classical theory to general probabilistic theories: circuits and  do-interventions}

 In a general probabilistic theory,  variables are associated with physical systems, and  cause-effect relations between them are induced  by physical processes \cite{henson2014theory,Chiribella_quantum_speedup_2019}.   From this perspective, the causal structure described by a DAG is associated with a sequence of physical processes.  For example, the DAG in Fig. \ref{fig:DAG_quantum} is associated to a sequence of two processes ${\cal P}_1$ and ${\cal P}_2$, with ${\cal P}_1$ transforming system $\Lambda$ into the composite system $A\otimes \Lambda$, and process  ${\cal P}_2$ taking the composite system $A\otimes \Lambda$ into  system $B$ in output, as illustrated in Fig.~\ref{Fig:circuit}. At the beginning, system $\Lambda$ starts off in a state $\rho_\Lambda  \in  {\sf St}  (\Lambda)$.  Then,  process ${\cal P}_1$ generates system $A$ and, in general, correlates with system $\Lambda$, producing the bipartite state  
\begin{align}\label{sigmaalambda}
\sigma_{A\Lambda}: =  {\cal P}_1  (\rho_\Lambda)  \in  {\sf St}  (A\otimes \Lambda) \, .
\end{align}      
Finally, process ${\cal P}_2$ generates variable   $B$, putting it in the state 
\begin{align}\label{Eq:tau_B}
\tau_B   :  = {\cal P}_2  (\sigma_{A\Lambda})  \in {\sf St}  (B) \, .
\end{align}

 {\em A priori}, the initial state of system $\Lambda$  and the processes ${\cal P}_1$  and  ${\cal P}_2$ may be unknown, and we will not make any assumption on them, except that they are a valid state and two valid processes allowed by our theory, respectively.

To quantify the amount of causal influence from system $A$ to system $B$, we  need an analogue of the do-intervention.  In a general probabilistic theory, we define the do-intervention to be  the operation that  resets system $A$ to  a given state $\rho_0$, removing any prior  correlations between system $A$ and other physical systems.  
\begin{definition}[Do-interventions]\label{def:do}
Consider a general probabilistic theory satisfying the Causality Axiom.  For a given system $A$ in the theory,  
and a given state state $\rho_0  \in  {\sf St}  (A)$,  the general do-intervention ${\rm do}  (\rho_0):  {\sf St}  (A) \to {\sf St} (A)$  is an affine map  of the  the discard-and-reprepare form 
\begin{align}\label{dorho}
{\rm do}  (\rho_0)  :=  \rho_0  \, \Tr_A\,, 
\end{align}
where $\Tr_A$ is the unique discarding operation for system $A$ (in quantum theory, the partial trace over system $A$'s Hilbert space.) 
\end{definition}  
  
In classical theory, Definition \ref{def:do}    is consistent with the standard definition of the do-operation ${\rm do}  (a_0)$, which resets the variable $A$ to a fixed value $a_0$.    In the notation of Definition \ref{def:do}, this operation is denoted by ${\rm do}  (P_{a_0})$, where   $P_{a_0}  (a) : =  \delta_{a,a_0}$ is the probability distribution that assigns probability 1 to the value $a=a_0$ and zero to all the other values. Note that, more generally, Definition \ref{def:do} allows for  do-operations that reset variable $A$ to a random value $a$ distributed according to  a general  probability distribution $P_0(a)$.

\begin{figure}
    \centering
    \includegraphics[width=\columnwidth]{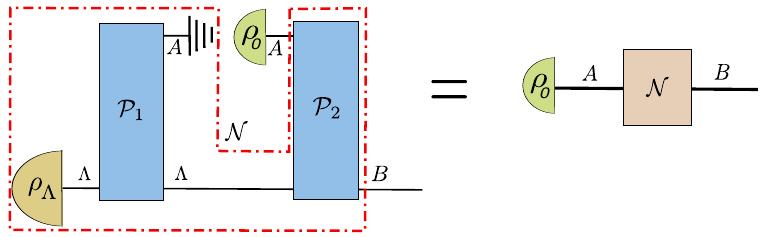}
    \caption{{\bf Inserting a do-intervention in the circuit of Fig. \ref{Fig:circuit} .}  The experimenter resets system $A$ to a given state $\rho_0$. As shown by the red dashed line on the right-hand side, the dependence of the state of the system $B$ from the state $\rho_0$ is given by the effective process $\N$.    }
    \label{fig:dointervention}
\end{figure}

In a general probabilistic theory, the do-intervention describes the act of resetting a system  to a fixed state,  breaking any  correlation that the system may  have had previously.    In particular, inserting a do-intervention   ${\rm do}  (\rho_0)$  into the circuit of Fig.~\ref{Fig:circuit} yields  a new circuit, depicted in Fig.~\ref{fig:dointervention}.   Here,  the effect of the do-intervention  ${\rm do}  (\rho_0)$ is to transform the joint  state $\sigma_{A\Lambda}$  [Eq. (\ref{sigmaalambda})]  into the product state 
\begin{align}
\nonumber \sigma'_{A\Lambda}:    &=    \Big(  {\rm do}  (\rho_0)  \otimes {\cal I}_\Lambda    \Big)  (\sigma_{A\Lambda})  \\
&  =\rho_0 \otimes \sigma_\Lambda \,,
\end{align} 
where ${\cal I}_\Lambda$ is the identity transformation on system $\Lambda$, and $\sigma_\Lambda  \in  {\sf St} (\Lambda)$ is the marginal state defined by 
\begin{align}\label{sigmalambda}
\sigma_{\Lambda} : =  \Tr_A  [  \sigma_{A\Lambda}] \,.
\end{align}
In the end, the  state of system $B$ is    
\begin{align}
\nonumber \tau_B'    &:  =  {\cal P}_2   (\sigma'_{A\Lambda})  \\
\nonumber &  =   {\cal P}_2   ( \rho_0  \otimes  \sigma_\Lambda)\\
&\equiv  {\cal N}  (\rho_0)   \, ,    \label{tauprime}  
\end{align}
where $\mathcal {N}:  \St{A} \to \St{B}$  is the process defined by  
\begin{align}\label{Eq:new_process_GPT}
 {\cal N}  (\rho)  :  =  {\cal P}_2   (  \rho  \otimes  \sigma_\Lambda) \, \qquad \forall \rho \in {\sf St} (A) \,.  
\end{align}
The process  $\cal N$ describes  the influence of the experimenter's reset operation for system $A$ on the output state of system $B$.   In the following, we will consider the situation where the  experimenter has black box access to the process ${\cal N}$, meaning that they can only prepare  the process' input  and perform measurements on  the process' output.   The goal of the experimenter is to estimate the amount of causal influence between systems $A$ and $B$ based on the observable input-output output statistics of   the process  ${\cal N}$. 

\begin{figure}
    \centering
    \includegraphics[width=\columnwidth]{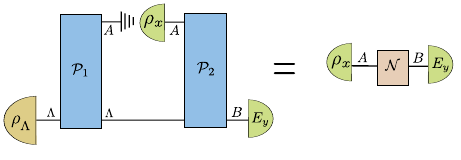}
    \caption{{\bf Probability distribution  of measurements on system $B$ following do-interventions on system $A$. }   The experimenter performs a do-intervention ${\rm do} (\rho_x)$ on system $A$, resetting it to the state $\rho_x$,  and a measurement $ (e_y)_{y=1}^N$ on system $B$.    The  probability of the outcome $y$ is $p(y|  {\rm do}  (x))   =  e_y  (\N (\rho_x))$ and can be computed in terms of the effective process $\N$.   }
    \label{fig:DAGwithmeasurement}
\end{figure}

\subsection{Maximum ACE between state preparation variables and measurement  outcome variables} \label{app:GCEMaxIsMaxACE}

Consider the scenario of Figure~\ref{fig:DAGwithmeasurement}, where an experimenter resets system $A$ to a state $\rho_x$, labeled by an index $x$ in the finite set  $X  = \{1,\dots, M\}$, and  subjects system $B$ to a measurement $(E_y)_{y \in  Y}$ with outcomes in the finite set $Y  = \{  1,
 \dots,  N\}$. In this scenario,  the state of system $B$ right before the measurement is 
\begin{align}
\tau_x   =   \mathcal{N}   (\rho_x) \  , 
\end{align} 
as one can see  from Eq. (\ref{tauprime}) by replacing $\rho_0$ with $\rho_x$ in the right-hand-side.  
Hence, the probability of observing outcome $y$ after resetting system $A$  to the state $\rho_x$   is 
\begin{align}
\nonumber p(y|  {\rm do}  (x))    &:= E_y (\tau_x)\\
  & =E_y  (   {\cal N}     (\rho_x)   ) \, . \label{measured} 
\end{align}
Then, the  ACE between the state preparation variable $X$ and the measurement outcome variable $Y$ is
\begin{align}
\nonumber  \mathrm{ACE}_{X\to Y  }   &= \max_{x,x'\in X}\frac{ \|   P_{{\rm do}  (x)} -  P_{{\rm do}  (x')} \|_1}2  \\
\nonumber   &  = \frac12   \, \max_{x,x'\in X}  \sum_{y\in Y}  \big|  p( y|  {\rm do}  (x))   -  p  (y|  {\rm do}  (x') ) \big| \\
\label{measuredACE}  &   =  \frac 12  \max_{x,x'\in  X}\sum_{y\in Y}  \big|  E_y  (   {\cal N}     (\rho_x)    -  E_y  (   {\cal N}     (\rho_{x'}) \big| \, .       \end{align}

Now, suppose that the experimenter optimizes the choice of measurement over all possible measurements allowed by the theory, searching for the measurement that maximizes the ACE.    The maximum ACE over all possible measurements is  given by 
\begin{align}
\nonumber {\rm ACE}    (  X,  {\cal N})     &:  =   \sup_{ (E_y)_{y\in  Y}  }   {\rm ACE}_{X\to Y}  \\
 \nonumber  & =
 \frac 12  \, \sup_{ (E_y)_{y\in  Y}  }  \max_{x,x'\in  X}\sum_{y\in Y}  \big|  E_y  (   {\cal N}     (\rho_x)    -  E_y  (   {\cal N}     (\rho_{x'}) \big| \\
  &  =   \max_{x,x'\in  X}  \frac{\|  {\cal N}     (\rho_x)    -  {\cal N}     (\rho_{x'}) \|_1}2\, , 
\end{align}
where the  supremum  is over all possible measurements $(E_y)_{y\in  Y}$  on system $B$, with all possible numbers of outcomes $N$, and  last equality follows from the definition of the operational norm in Eq. \eqref{Eq:operational_norm2}.   

Finally, suppose that the experimenter optimizes the choice of state preparations in order to maximize the causal effect.  The supremum over all possible state preparations is given by 
\begin{align}
\nonumber {\rm ACE}    ({\cal N})     &:  =   \sup_{ (\rho_x)_{x\in  X}  }   {\rm ACE}  (X,  {\cal N})    \\
 \nonumber   &  =  \sup_{ (\rho_x)_{x\in  X}  }    \max_{x,x'\in  X}  \frac{\|  {\cal N}     (\rho_x)    -  {\cal N}     (\rho_{x'}) \|_1}2\\
   &   =   \sup_{\rho,  \rho' \in  {\sf St}  (A)} \frac{\|  {\cal N}     (\rho)    -  {\cal N}     (\rho') \|_1}2 \, ,   \label{supACE}
\end{align}
where  the last equality follows by restricting the optimization over all indexed sets of input states $(\rho_x)_{x\in X}$ to an optimization over  pairs of states $(\rho,  \rho')$.

\begin{definition}
We call the quantity ${\rm ACE}    ({\cal N})$  in Eq. (\ref{supACE})   the {\em maximum classical ACE} induced by the process $\N :  {\sf St} (A) \to {\sf St}  (B) $.
\end{definition}
The maximum classical ACE quantifies the maximum strengths of the causal influence between  state preparation variables  and  measurement outcome variables.   

When the state space ${\sf St}  (A)$ is finite-dimensional, 
it is  standard to assume that ${\sf St}  (A)$  is also compact \cite{chiribella2010probabilistic,Barnum_2007,chiribella2016quantum,Plvala2023}.   In this case, the supremum is actually a maximum: 
\begin{proposition}\label{prop:supmax}
When the state space ${\sf St}  (A)$ is compact, one has 
 \begin{align}
{\rm ACE}    ({\cal N})   =  \max_{\rho,  \rho' \in  {\sf St}  (A)} \frac{\|  {\cal N}     (\rho)    -  {\cal N}     (\rho') \|_1}2 \, ,   \label{maximumACE}
 \end{align}
\end{proposition}

\proof  Let $(\rho_n,  \rho_n')_{n\in  \mathbb N}$ be a sequence of  pairs of (perfectly distinguishable) states achieving the supremum,  namely  
\begin{align}
\lim_{n\to \infty}  \frac{ \|  \N  ( \rho_n-  \rho_n' ) \|_1}{ 2}   =  \sup_{\rho, \rho'  \in {\sf St}  (A)} \frac{ \|  \N  ( \rho-  \rho' ) \|_1}{ 2}  \, .    
 \end{align}
Compactness of the state space ${\sf St}  (A)$ implies that one can extract a  subsequence   $(\rho_{n_k},  \rho_{n_k}')_{k\in  \mathbb N}$   that converges to a pair of   states $(\rho_0,  \rho_0')$. Hence, we have   
\begin{align}
\nonumber  \frac{ \|  \N  ( \rho_{0}-  \rho_{0}' ) \|_1}{ 2} &=   \lim_{k\to \infty}  \frac{ \|  \N  ( \rho_{n_k}-  \rho_{n_k}' ) \|_1}{ 2}  \\
&=\sup_{\rho, \,\rho' \in  {\sf St}  (A)} \frac{ \|  \N  ( \rho-  \rho' ) \|_1}{ 2}  \, .    
 \end{align}
\qed

\medskip

We conclude by showing that the maximum classical ACE is a lower bound to the  maximum causal effect $\CM$ introduced in Definition \ref{def:GCEmax}:  
\begin{proposition}\label{prop:ACECE}
 For every probabilistic theory, every pair of systems $A$ and $B$, and every process $\N  :  {\sf St}  (A) \to  {\sf St}  (B)$ one has the lower bound 
 \begin{align}\label{ACECEbound}
 \CM (\N)  \ge  {\rm ACE}  (\N)  \, .   
 \end{align}
  The equality   $\CM (\N) =  {\rm ACE}  (\N)$ holds if  the state space of system $A$ is generated by perfectly distinguishable pure states,  that is, if  every element $\delta$  of the vector space ${\sf St}_{\mathbb R}  (S)$ can be decomposed as    $\delta  =   a   \,    \sigma    -   b  \,    \sigma'$,
where $a$ and $b$ are nonnegative coefficients, and $\sigma$ and $\sigma'$ are perfectly distinguishable states.  \end{proposition}

{\bf Proof.} Since the operational norm is upper bounded by 2 for every pair of states $\rho$ and $\rho'$ \cite{chiribella2010probabilistic},  
  one has  the bound 
\begin{align}
\nonumber \CM  (\N)   &:  = \sup_{\rho \not  =  \rho'}   \frac{ \| \N  (\rho)  -  \N  (\rho')\|_1 }{\|\rho   -  \rho'\|_1} \\
 \nonumber &\ge  \sup_{\rho \not  =  \rho'}   \frac{ \| \N  (\rho)  -  \N  (\rho')\|_1 }{2}\\
    & \equiv  {\rm ACE}  (\N)\, .   \label{poiu} 
\end{align}  

If the state space ${\sf St}  (A)$ is generated by perfectly distinguishable states,  then the vector $\delta  =  \rho-\rho'$ can decomposed as  $\rho-  \rho'  =  a\ , \sigma  -  b\,  \sigma'$ where $\sigma ,\sigma'$ are perfectly distinguishable states, and $a, b$  are nonnegative numbers.   Note that the normalization of the states $\rho, \rho', \sigma,$ and $\sigma'$ implies 
\begin{align}  
\nonumber a-b   & =     a \,\Tr [\sigma]   -b\,  \Tr[\sigma']  \\
\nonumber &  =  \Tr  [a \sigma  -b  \sigma']\\
\nonumber  &  =\Tr_A  [\rho-  \rho']  \\
  \nonumber  &  = 1-1  \\
  &  =  0  \, ,
  \end{align}
  that is, $a= b$. 
  Note also that $a\not  =  0$ whenever $\rho \not =  \rho'$.   
Hence, for every $\rho  \not  =  \rho'$ we have the relation  
\begin{align}
\nonumber \frac{ \|  \N (\rho) -\N  (\rho') \|_1}{ \|  \rho  -  \rho'\|_1}  &  = \frac{ \|  \N (\rho-\rho')   \|_1}{ \|  \rho  -  \rho'\|_1}   \\
 \nonumber  &=  \frac{ \|  a\,  \N  ( \sigma-  \sigma' ) \|_1}{  \|  a\,  (\sigma  -  \sigma'  )\|_1}  \\
 \nonumber  &=\frac{ \|  \N  ( \sigma-  \sigma' ) \|_1}{  \|   \sigma  -  \sigma'  \|_1}\\
  &  \le  \sup_{\sigma~{\rm p.d.}~\sigma'} \frac{ \|  \N  ( \sigma-  \sigma' ) \|_1}{  \|   \sigma  -  \sigma'  \|_1}  \, , 
\end{align}
 where the supremum is over all pairs of perfectly distinguishable $\sigma$ and $\sigma'$.  Taking the supremum of the left-hand-side over all pairs of distinct  states $\rho \not =  \rho'  $, we then obtain the bound 
 \begin{align}
\nonumber  \CM  (\N)       & \le  \sup_{\sigma~{\rm p.d.}~\sigma'} \frac{ \|  \N  ( \sigma-  \sigma' ) \|_1}{  \|   \sigma  -  \sigma'  \|_1} \\
\label{edc}   &    =   \sup_{\sigma~{\rm p.d.}~\sigma'} \frac{ \|  \N  ( \sigma-  \sigma' ) \|_1}{ 2}  \\
 \nonumber  &\le   \sup_{\rho, \rho'  \in  {\sf St}  (A)} \frac{ \|  \N  ( \rho-  \rho' ) \|_1}{ 2} \\
  &  \equiv   {\rm ACE}  (\N)  \, .\label{rfv}
 \end{align}
 Combining Eqs.  (\ref{poiu}) and (\ref{rfv}) then yields the equality $ \CM (\N) =  {\rm ACE}  (\N)$.
\qed


\section{Properties of the maximum causal effect in general physical theories}
\subsection{Proof of Proposition~\ref{prop:propertiesmax}} \label{app:basicproperties_Max}

\begin{prop*}
 For every probabilistic theory satisfying the Causality Axiom, every pair of systems $A$ and $B$ and every process ${\cal N}:  {\sf St }  (A)  \to {\sf St}  (B)$, the maximum causal effect has the  following properties: 
 \begin{enumerate}
 \item {\em Range:}  $0\le  \CM   ({\cal N}) \le 1$,
 \item {\em Faithfulness:}  $\CM   ({\cal N})=0$ if and only if $\cal N$ is a discard-and-reprepare process, that is, if and only if ${\cal N}   =  \rho_0 \Tr_A$, 
 
 \item {\em Convexity:}   $ \CM    (\sum_i \, p_i \, {\cal N}_i) \le  \sum_i \,  p_i  \, \CM    ( {\cal N}_i)$ for every collection of processes $({\cal N}_i)_i$ and for every probability distribution $(p_i)_i$.
 \item {\em Data-processing inequality:}   $  \CM    ( {\cal  B}  \circ {\cal N}   \circ {\cal A})  \le  \CM    (  {\cal N} )$ for arbitrary processes ${\cal A}  : {\sf St}  ( A')  \to  {\sf St}  ( A),$  ${\cal N}  : {\sf St}  ( A)  \to  {\sf St}  ( B)$,  and ${\cal B}  : {\sf St}  ( B)  \to  {\sf St}  ( B')$, and arbitrary systems $A,A',B,$ and $B'$. 
 \end{enumerate}
\end{prop*}

\begin{proof}

\textbf{1. Range:} 
The upper bound follows from the data-processing inequality of the operational norm~\cite{chiribella2010probabilistic}. The lower bound follows from the positivity of the operational norm with equality when the process $\N$ is a discard-and-reprepare process. 
\medskip

\textbf{2. Faithfulness:}
The `if part' is trivial. For only if, first note from the property of the operational norm, a process $\N: \mathsf{St}(A) \to \mathsf{St}(B)$ has vanishing maximum causal effect: $\CM   ({\cal N})=0$ if and only if $\N (\rho) = \N(\sigma) = \rho_0$ for all $\rho, \sigma \in \mathsf{St}(A)$, and $\rho_0 \in \mathsf{St}(A)$ is a fixed state.  Now consider a discard-and-reprepare process $\mathcal{E}(\rho)=\rho_0 \Tr_{A}(\rho)=\rho_0$. Then we have $\mathcal{N}(\rho) = \mathcal{E}(\rho)$ for all $\rho \in \mathsf{St}(A)$, which implies $\mathcal{N}=\mathcal{E}=\rho_0 \Tr_A$. 
\medskip

\medskip
\textbf{3. Convexity:}
Consider {$\rho$} and {$\sigma$} are the states maximizing $\CM    (\sum_i \, p_i \, {\cal N}_i)$. Due to convexity of the operational norm~\cite{Kimura2010, DAriano2016}: $\norm{\sum_i p_i \N _i(\rho - \sigma) }_1 \le \sum_i p_i \norm{ \N _i(\rho - \sigma) }_1$, we have
\begin{align}
\CM    \left( \sum_i \, p_i \, {\cal N}_i\right) &= \frac{\norm {\sum_i p_i \N _i (\rho - \sigma)}_1}{\norm{\rho - \sigma}_1} \nonumber \\
&\le \sum_i p_i \frac{\norm {\N _i (\rho - \sigma)}_1}{\norm{\rho - \sigma}_1} \nonumber \\
& \le \sum_i p_i \CM(\N _i).
\end{align}
 \medskip 

\textbf{4. Data-processing inequality:} First consider any two arbitrary processes $\N _1 : {\sf St}  ( A')  \to  {\sf St}  ( A),$ and $\N _2 : {\sf St}  ( A)  \to  {\sf St}  ( B),$  and arbitrary systems $A,A',$ and $B$. Let $\rho \in \mathsf{St}(A')$ and $\sigma \in \mathsf{St}(A')$ be the states that achieve the maximum for $\CM    (  {\cal N}_2 \circ {\cal N}_1 )$: 

\begin{align}
\CM(  {\cal N}_2 \circ {\cal N}_1 ) &=  \frac{ \norm{\N _2 \circ \N _1 (\rho -\sigma)}_1}{\norm{\rho -\sigma}_1} \\
&= \frac{ \norm{\N _2\circ \mathcal{N} _1 (\rho -\sigma)}_1}{\norm{\mathcal{N}_1(\rho -\sigma)}_1} \cdot \frac{ \norm{ \mathcal{N}_1 (\rho -\sigma)}_1}{\norm{\rho -\sigma}_1} \\
&\le \CM (\mathcal{N}_2) \cdot \CM (\mathcal{N}_1). \label{Eq:CE_max_cascade} 
\end{align}
Note, if $\N _1$ or $\N_2$ is a discard-and-reprepare process, then the inequality in Eq.~\eqref{Eq:CE_max_cascade} trivially holds. Otherwise, the states maximizing $\CM    (  {\cal N}_2 \circ {\cal N}_1 ) > 0$ should necessarily respect $\norm{\N _1(\rho -\sigma)}_1 > 0$, allowing us to divide by $\norm{\N _1(\rho -\sigma)}_1$ in the second equation. The last inequality holds because the pair of states $\{\N _1 (\rho), \N _1 (\sigma)\}$ and $\{\rho, \sigma\}$ are in general suboptimal to achieve $\rm{CE}_{\max}(\N _2)$ and $\rm{CE}_{\max}(\N _1)$ respectively.

Now, using the inequality in Eq.~\eqref{Eq:CE_max_cascade} recursively, we have 
\begin{align}
   \CM    ( {\cal  B}  \circ {\cal N}   \circ {\cal A})  &\le  \CM    (  {\cal B}) \cdot \CM    (  {\cal N}) \cdot \CM    (  {\cal A}) \nonumber \\
   & \le \CM    (  {\cal N}).
\end{align}
Here the last inequality is due to $\CM    (  {\cal B}) {\le} 1$ and $\CM    (  {\cal A}) {\le} 1$, with the equality when both $\cal A$ and $\cal B$ are reversible processes.
\end{proof}

\subsection{Maximum causal effect in state spaces generated by perfectly distinguishable states}\label{app:pdstuff}

Here, we prove Proposition \ref{prop:pdmax} of the main text. 

\begin{proof}
Let us start by proving Eq.  (\ref{distinguishablerhorho'}). 
To this purpose, note that Propositions \ref{prop:ACECE} and \ref{prop:supmax}  imply the inequality 
\begin{align}
\nonumber \CM  (\N)    &  =     \max_{\rho, \rho'  \in  {\sf St} (A) }  \frac{  \|  \N  (  \rho  -  \rho')\|_1}2 \\
& \ge   \max_{\sigma ~ {\rm p.d.}~   \sigma'  }  \frac{  \|  \N  (  \sigma  -  \sigma')\|_1}2 \, , 
\end{align}
where the notation $\rho ~  {\rm p.d.} ~\rho'$ means that the states $\rho$ and $\rho'$ are perfectly distinguishable. 

On the other hand, Eq. (\ref{edc}) implies the  inequality
\begin{align}
\nonumber \CM  (\N)    &  \le   \sup_{\sigma ~ {\rm p.d.}~   \sigma'  }  \frac{  \|  \N  (  \sigma  -  \sigma')\|_1}2 \, , 
\end{align}  
where the supremum is actually a maximum due to the compactness of ${\sf St}(A)$ (by the same argument as in the proof of Proposition \ref{prop:supmax}.)     
Hence,  we obtained the equality  
\begin{align}\label{app:distinguishablerhorho'}
\nonumber \CM  (\N)    & =   \max_{\sigma ~ {\rm p.d.}~   \sigma'  }  \frac{  \|  \N  (  \sigma  -  \sigma')\|_1}2 \, , 
\end{align}
which is Eq. (\ref{distinguishablerhorho'}) in the main text.  

Now,  it only remains to prove that the maximization can be restricted to pure states without loss of generality.   Let $\rho  =  \sum_i  \, p_i \, \psi_i$ and $\rho'  =  \sum_j\,  p_j'  \,  \psi_j'$ be two perfectly distinguishable states, decomposed as convex mixtures of pure states.  Note that the perfect distinguishability of $\rho$ and $\rho'$ implies that every pure state $\psi_i$ is perfectly distinguishable fro every pure state $\psi_j'$.     Hence, one has the bound  
 \begin{align}
\nonumber  \|  \N (\rho)  - \N (\rho') \|_1   & =  \left\|  \sum_{i,j}  \,  p_i\, p'_j  \,  \Big( \N (\psi_i)  - \N (\psi_j')\, \Big)\right\|_1  \\
 \nonumber   & \le  \sum_{i,j}  \,  p_i\, p'_j  \,  \left\|   \N (\psi_i)  - \N (\psi_j')  \right\|_1  \\
\nonumber  & \le \max_{i,j}    \left\|   \N (\psi_i)  - \N (\psi_j')  \right\|_1\\
 & \le \max_{\psi ~{\rm p.d.}~ \psi'  }    \left\|   \N (\psi)  - \N (\psi')  \right\|_1\, ,
   \end{align}
   where the first inequality follows from the convexity of the operational norm \cite{chiribella2010probabilistic}.  
Maximizing the left-hand-side over all pairs of states $\rho$ and $\rho'$, we then obtain the bound 
\begin{align}
\max_{\rho ~{\rm p.d.}~ \rho'} \|  \N (\rho)  - \N (\rho') \|_1  \le \max_{\psi ~{\rm p.d.}~ \psi' } \|  \N (\psi)  - \N (\psi') \|_1 \, .
\end{align}
The converse bound 
\begin{align}
\max_{\rho~{\rm p.d.}~ \rho'}  \|  \N (\rho)  - \N (\rho') \|_1  \ge \max_{\psi ~{\rm p.d.}~ \psi' } \|  \N (\psi)  - \N (\psi') \|_1
\end{align}
trivially holds,  since restricting the maximization over pairs of perfectly distinguishable pure states cannot increase the maximum.  Hence, we  obtained the equality 
\begin{align}
\max_{\rho~{\rm p.d.}~\rho'  } \|  \N (\rho)  - \N (\rho') \|_1  = \max_{\psi~{\rm p.d.}~\psi' } \|  \N (\psi)  - \N (\psi') \|_1 \, , 
\end{align}
which, combined with Eq. (\ref{app:distinguishablerhorho'}), implies Eq. (\ref{distinguishablepsipsi'}). 
 \end{proof}

\subsection{Continuity of the maximum causal effect} \label{app:CE_max_continuity}

First, we prove a more general continuity result where the state space does not follow the condition in Definition~\ref{def:generated}.

\begin{lemma} \label{Lem:continuity_gen_max}
    if two processes $\cal  N$ and $\cal N'$ are $\epsilon$-close: $ \|  {\cal N}  (\rho)   -  {\cal N'}  (\rho)\|_1 \le \epsilon \, , \forall \rho  \in  {\sf St}  (A)$,  then their maximum causal effects satisfy:  
      \begin{align}
                 \abs{\CM (\N) - \CM (N')} \le \frac{2\epsilon}{\min \{\norm{\rho-\sigma}_1, \norm{\widetilde{\rho}-\widetilde{\sigma}}_1\}},   
        \end{align}
where the state pairs $\{\rho, \sigma\} \in \St{A}$ and $\{\widetilde{\rho}, \widetilde{\sigma}\} \in \St{A}$ achieve $\CM(\N)$ and $\CM (\N')$ respectively.
\end{lemma}

\begin{proof}
 In the lemma, we have considered the states attaining the optimum causal effect for $\mathcal{N}: \mathsf{St}(A) \to \mathsf{St}(B)$ are ${\rho} \in \mathsf{St}(A)$ and ${\sigma} \in \mathsf{St}(A) $. Then we have 
\begin{align}
    &\norm{\rho-\sigma}_1 \cdot\CM(\mathcal{N})=\norm{\mathcal{N}(\rho)-\mathcal{N}(\sigma)}_1 \nonumber \\
    &\le\norm{\mathcal{N}(\rho)-\mathcal{N'}(\rho)}_1+\norm{\mathcal{N}'(\rho)-\mathcal{N}'(\sigma)}_1 \nonumber \\
    &\quad +\norm{\mathcal{N}'(\sigma)-\mathcal{N}(\sigma)}_1 \nonumber \\
    &\le 2\epsilon + \norm{\rho-\sigma}_1\cdot \CM(\N')
\end{align}
The first inequality is due to the triangle inequality of the operational norm. In the second inequality we use $\norm{\rho-\sigma}_1 \CM(\mathcal{N}')\ge \norm{\N '(\rho)-\N '(\sigma)}_1 $.
Using the similar arguments, we can have $\norm{\widetilde{\rho}-\widetilde{\sigma}}_1 \cdot \CM(\mathcal{N}')\le 2\epsilon + \norm{\widetilde{\rho}-\widetilde{\sigma}}_1\cdot \CM(\mathcal{N})$. 

Hence, we prove the lemma by rearranging the above equations and noting $\norm{\rho-\sigma}_1 \ge \min \{\norm{\rho-\sigma},\norm{\widetilde{\rho}-\widetilde{\sigma}}_1 \}$, and $\norm{\widetilde{\rho}-\widetilde{\sigma}}_1 \ge \min \{\norm{\rho-\sigma},\norm{\widetilde{\rho}-\widetilde{\sigma}}_1 \}$.

\end{proof}

Next, we focus on Proposition~\ref{prop:continuity_max} where the condition in Definition~\ref{def:generated} holds. For the reader's convenience, we repeat their statements here:
\begin{prop*}
   \emph{ Continuity:}  if two processes $\cal  N$ and $\cal N'$ are $\epsilon$-close,  then their maximum causal effects are $\epsilon$-close;  in formula, the condition $ \|  {\cal N}  (\rho)   -  {\cal N'}  (\rho)\|_1 \le \epsilon \, , \forall \rho  \in  {\sf St}  (A)$ 
 implies  $|\CM   ({\cal N})  -  \CM   ({\cal N'})|  \le \epsilon$   
  \end{prop*}
\begin{proof} This is a simple consequence of Statement 1 of Lemma~\ref{Lem:continuity_gen_max} and the fact established in Proposition~\ref{prop:pdmax} that $\min \{\norm{\rho-\sigma}_1, \norm{\widetilde{\rho}- \widetilde{\sigma}}_1\}=2$.
\end{proof}

\section{Properties of the minimum causal effect in general physical theories}

\subsection{Proof of Proposition~\ref{prop:propertiesmin}}\label{app:basicproperties_Min}
\begin{proof}

\textbf{1. Range:} The upper bound follows from the data-processing inequality of operational norms: $\norm{\mathcal{N}(\rho-\sigma)}_1\le \norm{\rho -\sigma}_1$ for all $\rho$ and $\sigma$~\cite{chiribella2010probabilistic}. The lower bound follows from the non-negativity of the operational norm.
\medskip

\textbf{2. Data-processing inequality:}
The bound follows from the data-processing inequality of the operational norm: $\norm{\mathcal{B}\circ \N (\rho -\sigma)}_1 \le \norm{\N (\rho - \sigma)}_1$ for all $\rho$ and $\sigma$ in $\sf {St} (A)$.

\end{proof}

\subsection{Minimum causal effect in compact state spaces  generated by perfectly distinguishable states} \label{app:property_min_perfectly_disting}

Here we prove Proposition \ref{prop:pdmin} and  Theorem \ref{theo:ddworst}.

\medskip

\noindent 

{\bf Proof of Proposition \ref{prop:pdmin}.} The proof follows the lines of the proof of Proposition \ref{prop:ACECE}. 
 First, one has  the bound 
\begin{align}
\nonumber \Cm  (\N)   &:  = \inf_{\rho \not  =  \rho'}   \frac{ \| \N  (\rho)  -  \N  (\rho')\|_1 }{\|\rho   -  \rho'\|_1} \\
 \nonumber &\le  \inf_{\rho ~{\rm p.d.} ~\rho'} \frac{ \| \N  (\rho)  -  \N  (\rho')\|_1 }{\|\rho   -  \rho'\|_1}  \\
 &=  \inf_{\rho ~{\rm p.d.} \rho'}~ \frac{ \| \N  (\rho)  -  \N  (\rho')\|_1 }{2}  \label{lkjh}
 \end{align}  

If the state space ${\sf St}  (A)$ is generated by perfectly distinguishable states,  then the vector $\delta  =  \rho-\rho'$ can decomposed as  $\rho-  \rho'  =  a\ , (\sigma  -   \sigma')$ where $\sigma ,\sigma'$ are perfectly distinguishable states, and $a$ is a real number.  
Hence, for every $\rho  \not  =  \rho'$ we have the relation  
\begin{align}
\nonumber \frac{ \|  \N (\rho) -\N  (\rho') \|_1}{ \|  \rho  -  \rho'\|_1}  &  = \frac{ \|  \N (\rho-\rho')   \|_1}{ \|  \rho  -  \rho'\|_1}   \\
 \nonumber  &=  \frac{ \|  a\,  \N  ( \sigma-  \sigma' ) \|_1}{  \|  a\,  (\sigma  -  \sigma'  )\|_1}  \\
 \nonumber  &=\frac{ \|  \N  ( \sigma-  \sigma' ) \|_1}{  2}\\
  &  \ge  \inf_{\sigma~{\rm p.d.}~\sigma'} \frac{ \|  \N  ( \sigma)-  \N(\sigma' ) \|_1}{2}  \, , 
\end{align}
 where the infimum is over all pairs of perfectly distinguishable $\sigma$ and $\sigma'$.  Minimizing the left-hand-side over all pairs of states $\rho \not =  \rho'  $, we then obtain the bound 
 \begin{align}
 \Cm  (\N)        \ge  \inf_{\sigma~{\rm p.d.}~\sigma'} \frac{ \|  \N  ( \sigma)- \N( \sigma' ) \|_1}{  2  }  \,, 
 \end{align}
which combined with  Eq. (\ref{lkjh}) yields the equality  
 \begin{align}
 \Cm  (\N)        =  \inf_{\sigma~{\rm p.d.}~\sigma'} \frac{ \|  \N  ( \sigma)- \N( \sigma' ) \|_1}{  2  }  \,.  
 \end{align}
Finally, compactness of the state space ${\sf St} (A)$ implies that the infimum is actually a minimum, by the same argument used in the proof of Proposition \ref{prop:supmax}. \qed 

\medskip  

We now move to the proof of Theorem \ref{theo:ddworst}. The proof uses a few  lemmas, provided in the following: 

\begin{lemma}\label{lem:A=tildeA}
For two systems $A$ and $B$, and a process ${\cal N} :  {\sf St}  (A) \to {\sf St}  (B) $, consider the sets  
\begin{align}
\nonumber {\sf A}   ({\cal N})  : = &\left\{  \frac{  \| p\,  {\cal N}  (\rho)  - (1-p)\, {\cal N}   (\rho')\|_1}{ \|   p\, \rho  -  (1-p)\,  \rho'\|_1}  ~\right|  ~  \rho, \rho'  \in  {\sf St} (A)  \, ,  \, \rho \not  =  \rho'\, , \\
 &  \left. \phantom{ \frac{  \| p\,  {\cal N}  (\rho)  - (1-p)\, {\cal N}   (\rho')\|_1}{ \|   p\, \rho  -  (1-p)\,  \rho'\|_1}}   \qquad \qquad     p\in [0,1] \right\}  
\end{align}
and
\begin{align}
\nonumber \widetilde{\sf A}  ({\cal N})   : = \Big\{    \|  p\,  {\cal N}  (\tilde \rho)  - (1-p)\, {\cal N}   (\tilde \rho')\|_1   ~\Big|&    ~ \tilde  \rho, \tilde \rho'  \in  {\sf St} (A) , \tilde \rho  ~{\rm p.d.} ~\tilde \rho' \,  , \\
  &   \qquad \quad  p\in  [0,1]     \Big\}  \,.
\end{align}
 If the state space ${\sf St}  (A)$ is generated by perfectly distinguishable states, then the sets  ${\sf A}  ({\cal N})$ and $\widetilde{\sf A}  ({\cal N}) $ coincide.   
\end{lemma}

{\bf Proof.}    
Since any two perfectly distinguishable states $\rho$ and $\rho'$ satisfy the equality $\| p \,\rho  -  (1-p)\rho'\|_1= 1$,  the  inclusion   $\widetilde{\sf A}  ({\cal N})  \subseteq {\sf A}  ({\cal N})$ holds trivially. Hence, we only need to prove the converse inclusion.   
 
Let us take two arbitrary states $\rho$ and $\rho'$ and an arbitrary probability $p$.  Since the vector space ${\sf St}_{\mathbb R}  (A)$ is generated by perfectly distinguishable states, there exist two perfectly distinguishable states $\tilde \rho$ and $\tilde \rho'$ and two positive coefficients  $a$ and $b$   such that $p\,  \rho  -  (1-p)\, \rho'  =  a \tilde \rho  -  b \tilde \rho'$.  Defining the probability $q  :=  a/(a+b)$, we can rewrite the previous equality as $p\,  \rho  -  (1-p)\, \rho'   =  (a+b) \,  (  q \,  \tilde \rho  -  (1-q) \,  \tilde \rho')$, which implies 
\begin{align}
\nonumber \frac{  \| p\,  {\cal N}  (\rho)  - (1-p)\, {\cal N}   (\rho')\|_1}{ \|   p\, \rho  -  (1-p)\,  \rho'\|_1} & =  \frac{  \|    {\cal N}  \big(p \,\rho    - (1-p)\, \rho'\big)\|_1}{ \|   p\, \rho  -  (1-p)\,  \rho'\|_1}  \\
\nonumber &=  \frac{  \|   (a+b)    \,  {\cal N}  \big(q\,  \tilde \rho   - (1-q)\, \tilde \rho'\big)\|_1}{ \|    (a+b)  \, \big ( q\, \tilde \rho  -  (1-q)\,  \tilde \rho'\big) \|_1}\\
\nonumber & =  \frac{  \|   {\cal N}  \big(q \, \tilde \rho   - (1-q)\,   \tilde \rho'\big)\|_1}{ \|   q\, \tilde \rho  -  (1-q)\,  \tilde \rho'\|_1}\\
& =   \| q\,  {\cal N}  (\tilde \rho)  - (1-q)\, {\cal N}   (\tilde \rho')\|_1  \,.   
\end{align}
Since the probability $p$ and the states $\rho$ and $\rho'$ are arbitrary, this equality implies the inclusion  ${\sf A}  ({\cal N})  \subseteq \widetilde{\sf A}  ({\cal N})$.  \qed

 \begin{lemma}\label{lem:ddworst}
  If the state space ${\sf St}  (A)$ is compact and  generated by perfectly distinguishable states, then the equality 
  \begin{align}\label{bastaa}
  {\rm DP}_{\min}  ({\cal N})   =  \min_{0\le p\le 1}  \min_{  \rho  \,  {\rm p.d.}\,  \rho'}\, \|  p\, {\cal N}  (\rho)   -  (1-p) \,    {\cal N}  (\rho')  \|_1  \, .   
  \end{align}
  holds for   every process ${\cal N} :  {\sf St}  (A)  \to  {\sf St}  (B)$ and for every system $B$. 
 \end{lemma}

 {\bf Proof.}   The definition of ${\rm DP}_{\min}  ({\cal N})$ implies 
 \begin{align}
 \nonumber  {\rm DP}_{\min}  ({\cal N})  &  :=   \min_{p\in  [0,1]} \inf_{\rho\not = \rho'}  \frac{\|   \N (   p\,\rho  +  (1-p)\,  \rho')\|_1}{\| p\,\rho  +  (1-p)\,  \rho'\|_1}  \\
 \nonumber    &   =   \inf  \big\{  x~|~  x\in   {\sf A} ({\cal N})\big\}  \\ 
  \nonumber    &   =   \inf  \big\{  x~|~  x\in   \widetilde{\sf A} ({\cal N})\big\}  \\
    &  =  \min_{0\le p\le 1}  \inf_{  \rho  \,  {\rm p.d.}\,  \rho'}\, \|  p\, {\cal N}  (\rho)   -  (1-p) \,    {\cal N}  (\rho')  \|_1  \, , 
 \end{align}  
 where the third equality follows from Lemma \ref{lem:A=tildeA}.   Since the state space ${\sf St} (A)$ is compact, the infimum is actually a minimum, by the same argument as in the proof of Proposition \ref{prop:supmax}. \qed 
 
\begin{lemma}\label{lem:oneeffect}
 For every system $S$, and for every pair of states $\rho, \rho'  \in  {\sf St} (S)$, the operational norm satisfies the equality
 \begin{align}
   \| \rho  - \rho'\|_1  =   2\,   \max_{e\in {\sf Eff}  (S)}  \, e(\rho -\rho')  \, .
 \end{align}
\end{lemma}

{\bf Proof.}  By Eq. (\ref{opnorm}),  the operational norm is $\|\rho - \rho'\|   =  \max_{e\in {\sf Eff}  (S)}  \, e (  \rho-\rho')    - \min_{f\in {\sf Eff}  (S)}  \, f (\rho-\rho')$.    Since the positive functional $f\in {\sf Eff}  (S)$ corresponds to the outcome of some measurement,    there must  exist another positive functional $f' \in {\sf Eff}  (S)$  such that, for every state $\xi  \in  {\sf St}  (S)$,  one has  $f(\xi)  + f'  (\xi)=1 $.  Using this equation for $\xi  = \rho$ and $\xi  = \rho'$, we obtain  the relation $  f (\rho-\rho') =      f'  (\rho'-\rho)$, which implies   $ \min_f  f (\rho-\rho') \ge \min_{f'}      f'  (\rho'-\rho)$.  Since the same argument applies to $f'$, in fact we have the equality  $\min_f   f (\rho-\rho') =    \min_{f'}   f'  (\rho'-\rho)$.  which implies 
\begin{align}
\nonumber \|  \rho- \rho'\|_1  &=  \max_{e\in {\sf Eff}  (S)}  \, e (  \rho-\rho')    - \min_{f\in {\sf Eff}  (S)}  \, f (\rho-\rho')  \\
\nonumber &=  \max_{e\in {\sf Eff}  (S)}  \, e (  \rho-\rho')    - \min_{f'\in {\sf Eff}  (S)}  \, f' (\rho'-\rho) \\
\nonumber &=  \max_{e\in {\sf Eff}  (S)}  \, e (  \rho-\rho')    + \max_{f\in {\sf Eff}  (S)}  \, f (\rho-\rho') \\
& =  2 \, \max_{e\in {\sf Eff}  (S)}  \, e (  \rho-\rho') \, .
\end{align}
\qed

\begin{lemma}\label{lem:erho}
 For every system $S$, and for every pair of states $\rho, \rho'  \in  {\sf St} (A)$, there exists a positive functional $  e   \in   {\sf Eff}  (S)$   such that  
 \begin{align}
 e(\rho)   &\ge \frac{  \|  \rho-\rho'\|_1}2 \label{erho} \\
 e (\rho') &\le 1-   \frac{  \|  \rho-\rho'\|_1}2 
  \, .   \label{erhoprime}
 \end{align}
\end{lemma}
{\bf Proof.}   By Lemma \ref{lem:oneeffect}, there exists a positive functional  $e \in  {\sf Eff}  (S)$  such that $  e (\rho-  \rho')  =  \|  \rho-\rho'\|_1/2 $.  Since $e$ is positive and $\rho'$ is a state, we must have $e (\rho) \ge e  (\rho-
\rho')  = \|  \rho-\rho'\|_1/2$, which implies Eq. (\ref{erho}).  Since $e (\rho)$ is a probability, and therefore is upper bounded  by 1, we have $  1-  e(\rho')     \ge   e (\rho-  \rho')   =   \|  \rho-\rho'\|_1/2$, which implies Eq.  (\ref{erhoprime}).  \qed

\begin{lemma}\label{lem:withoutp}
 For every system $A$, and for every pair of states $\rho, \rho'  \in  {\sf St} (A)$, one has the bound
 \begin{align}
 \| p\,   \rho  -  (1-p)\,  \rho' \|_1   \ge   \|  \rho-  \rho'\|  -1  \, . \end{align}
\end{lemma}

{\bf Proof.}  The definition of the operational norm [Eq.  (\ref{opnorm})] implies
\begin{align}
\nonumber \| p\,   \rho  -  (1-p)\,  \rho' \|_1   &  =  \max_e  e (p\,   \rho  -  (1-p)\,  \rho'  ) \\
\nonumber & \qquad  -  \min_f  f (p\,   \rho  -  (1-p)\,  \rho'  )  \\
\nonumber  &  =  \max_e  e (p\,   \rho  -  (1-p)\,  \rho'  )  \\
 &\qquad +  \max_f  f (  (1-p)\,  \rho' -  p\,   \rho ) \,.    \label{basta}
\end{align}  
Moreover, Lemma \ref{lem:erho} applied to the norm $\|  \rho'-\rho\|_1$ implies that there exists a functional  $e \in  {\sf Eff}  (S)$    satisfying Eqs.  (\ref{erho}) and  (\ref{erhoprime}), and a functional     $f\in  {\sf Eff}  (S)$ satisfying the conditions $f(\rho')  \ge  \| \rho'-\rho \|_1/2$ and 
 $f(\rho)  \le  1- \| \rho'-\rho \|_1/2$. Hence, equation (\ref{basta}) implies  
 \begin{align}
 \nonumber \| p\,   \rho  -  (1-p)\,  \rho' \|_1     &\ge  e  (  p \, \rho  -  (1-p) \, \rho')  +  f  ( (1-p)\, \rho'  - p \,\rho)\\
  \nonumber  &  \ge   p\,  \frac{\| \rho-\rho'\|_1 }2  - (1-p)\,  \left(1-  \frac{\| \rho-\rho'\|_1 }2\right)  \\
   \nonumber & \quad +    (1-p)\,  \frac{\| \rho'-\rho\|_1 }2  - p\,  \left(1-  \frac{\| \rho'-\rho\|_1 }2\right)\\
   &  =  \| \rho-\rho'\|_1   - 1 \, ,
 \end{align}
the last equation following from the homogeneity  property of the norm, which implies  $\|  \rho'  -\rho\|_1  =  \|  \rho- \rho'\|_1$. \qed

\medskip

\medskip
\noindent \textbf{Proof of Theorem~\ref{theo:ddworst}.}
By Lemma \ref{lem:ddworst} we have the equality 
\begin{align}\label{a}
{\rm DP}_{\min}  ({\cal N})  =  \min_p  \min_{\rho  \,  {\rm p.d.}\,  \rho'}  \|  p \, {\cal N}  (\rho)  -  (1-p)\,  {\cal N}  (\rho')\|_1\,.
\end{align}
On the other hand, Lemma \ref{lem:withoutp} implies  the bound 
\begin{align}\label{b}
\|  p \, {\cal N}  (\rho)  -  (1-p)\,  {\cal N}  (\rho')\|_1 \ge  \|  {\cal N}  (\rho)  -   {\cal N}  (\rho') \|_1-1\,, 
\end{align}
valid for every probability $p$ and every pair of states $\rho$ and $\rho'$.
Finally, the definition of $\Cm  ({\cal N})$ implies    the bound  
\begin{align}\label{c}
\|  {\cal N}  (\rho)  -   {\cal N}  (\rho') \|_1  \ge   \Cm  ({\cal N})  \,  \|  \rho  - \rho'\|_1   =  2  \,  \Cm  ({\cal N})\,,
\end{align}
where the last equality holds for perfectly distinguishable $\rho$ and $\rho'$.  
Combining Eqs.  (\ref{a}), (\ref{b}), and (\ref{c}), we obtain the bound  ${\rm DP}_{\min}  ({\cal N})    \ge 2 \,  \Cm  ({\cal N})    -1 $. \qed

\subsection{Proof of Proposition \ref{prop:continuity_minimum}} \label{app:CE_continuity_min}
Here we show that the minimum causal effect $\Cm (\N)$ is continuous in the process $\N:  {\sf St}  (A)  \to {\sf St} (B)$ whenever the state space ${\sf St} (A) $ is compact and generated by perfectly distinguishable states.  

\medskip

  \proof
Let $\rho,  \rho'  \in  {\sf St} (A)$ be two perfectly distinguishable states such  that  
\begin{align}
\Cm  (\N)   =  
\frac  {  \|  \N ( \rho'-\rho)  \|_1 }2 \, .
\end{align}
Note that we have 
\begin{align}
\nonumber &\frac  {  \|  \N ( \rho'-\rho)  \|_1 }2   \\
\nonumber &\qquad =  \frac  {  \|  \N ( \rho')  -  \N'  (\rho')   +  \N'  (\rho'- \rho)  +  \N'(\rho) -  \N (\rho)  \|_1 }2 \\
\nonumber &  
\qquad \le  \frac  {  \|  \N ( \rho')  -  \N'  (\rho') \|_1   + \| \N'  (\rho'- \rho) \|_1 +  \|\N'(\rho) -  \N (\rho)  \|_1 }2\\
&\qquad \le 
 \Cm  (  \N')   +  \epsilon  \, ,
\end{align}
where we used the inequalities $ \|  \N ( \rho')  -  \N'  (\rho') \|_1 \le \epsilon$,   $\| \N'  (\rho'- \rho) \|_1   \le  2  \, \Cm(\N')$, and $ \|  \N ( \rho)  -  \N'  (\rho) \|_1 \le \epsilon$. 

Hence, we obtained the inequality  $\Cm (\N)  \le \Cm (\N') + \epsilon$.  Applying the same argument to $\Cm(\N')$, we then obtain the inequality $\Cm (\N')  \le \Cm (\N) + \epsilon$.  
  Combining the two inequality  we then obtain $|\Cm   ({\cal N})  -  \Cm   ({\cal N'})|  \le \epsilon.$ \qed

\section{Maximum and minimum causal effects in quantum theory}

\subsection{Proof of Proposition \ref{prop:quantumeffectclassicalchan}}\label{app:quantumeffectclassicalchan}

 For every quantum state $\rho\in  {\sf St} (A)$,    define  the probability distribution $P(a):  = \langle a| \rho |a\rangle$.

For every quantum state $\rho\in  {\sf St} (A)$,   define the diagonal density matrix $\rho_{\rm diag}:  =  \sum_a \langle a| \rho |a\rangle  \,  |a\rangle \langle a|$ and note that  $\N (\rho)   =  \N  (\rho_{\rm diag})$. 

Now, for every pair of  distinct density matrices $\rho$ and $\rho'$, the condition $\rho_{\rm diag} =  \rho'_{\rm diag} $ implies $\N (\rho)  =  \N(\rho')$. Hence,  the maximization in Eq.  (\ref{maxquantum}) can be restricted without loss of generality to orthogonal pure states $\rho$ and $\rho'$  such that $\rho_{\rm diag} \not =  \rho_{\rm diag}'$.  Using this fact, we obtain 
\begin{align}
\nonumber \CM (\N)  &  =    \sup_{\rho \not  =  \rho'}  \frac{\|  \N(\rho)  - \N(\rho')\|_1}{\|  \rho  -\rho'\|_1} \\
  \nonumber & = \sup_{\rho , \rho':~  \rho_{\rm diag}\not  =  \rho_{\rm diag}'}  \frac{\|  \N(\rho)  - \N(\rho')\|_1}{\|  \rho  -\rho'\|_1} \\  
\nonumber &  = \sup_{\rho , \rho':~  \rho_{\rm diag}\not  =  \rho_{\rm diag}'}  \frac{\|  \N(\rho_{\rm diag})  - \N(\rho_{\rm diag}')\|_1}{\|  \rho  -\rho'\|_1} \\  
  \nonumber &  \le \sup_{\rho , \rho':~  \rho_{\rm diag}\not  =   \rho_{\rm diag}'}  \frac{\|  \N(\rho_{\rm diag})  - \N(\rho_{\rm diag}')\|_1}{\|  \rho_{\rm diag}  -\rho_{\rm diag}'\|_1} \\ 
\nonumber   &=    \sup_{  P  \not  = P'}\frac{  \|  \N_{\rm c}  (P) -  \N_{\rm c}  (P') \|_1  }{\|  P-  P'\|_1  }\\
&  \equiv {\rm ACE}_{X\to Y} \, ,
\end{align}
where the inequality follows from the bound $\|  \rho  - \rho'\|_1 \ge \|  \rho_{\rm diag}  - \rho'_{\rm diag} \|_1$, and the last equality follows from Proposition \ref{prop:classicalACE} upon  defining  the probability distributions $P(a):  =  \langle a|  \rho  |a\rangle$ and $P'(a):  =  \langle a|  \rho'  |a\rangle$, and the classical channel $\N_{\rm c}$ with conditional probability distribution    $p(b|{\rm do } (a)):  =q(b|a)$.  

The converse inequality trivially holds:  
\begin{align}
\nonumber {\rm ACE}_{X\to Y}& =     \sup_{  P  \not  = P'}\frac{  \|  \N_{\rm c}  (P) -  \N_{\rm c}  (P') \|_1  }{\|  P-  P'\|_1  } \\
  \nonumber &  =    \sup_{  \rho_{\rm diag}  \not  = \rho_{\rm diag}'}\frac{  \|  \N  (\rho_{\rm diag}) -  \N  (\rho'_{\rm diag}) \|_1  }{\|  \rho_{\rm diag}-  \rho'_{\rm diag}\|_1  } \\
 \nonumber  & \le   \sup_{  \rho  \not  = \rho'}\frac{  \|  \N  (\rho) -  \N  (\rho') \|_1  }{\|  \rho-  \rho'  \|_1  } \\
 &\equiv  \CM (\N) \, .
\end{align}

\qed

\subsection{Proof of  Proposition~\ref{Lem:faithful}} \label{App:faithful_proof}

The proof of Proposition~\ref{Lem:faithful} uses two basic lemmas.  The first lemma states an elementary property of the induced trace norm,  defined as $\|{\cal M}\|_1  : =  \max_{  O   \in   L( \H _A) \,  O\not  = 0}  \,  \| {\cal M}  (O)\|_1/\|O\|_1$  for a linear map ${\cal M}:  L(\H _A)  \to L(\H _B)$. 

\begin{lemma}\label{lem:norm1}
For every linear  map ${\cal M}:  L(\H _A)  \to L(\H _B)$, one has the bound $\|  {\cal M}\|_1 \le   2 \max_{\rho  \in  {\sf St}  (A)}    \|  {\cal M}  (\rho)\|_1$.  
\end{lemma}

{\bf Proof.} For every operator $O\in  L (\H _A)$,  one has 
\begin{align}
\nonumber \|  {\cal M}  (O)\|_1 &= \left\|  \frac 12 {\cal M}  (O+O^\dag) + \frac{1}{2i} {\cal M}  \left(i  O-i O^\dag \right)  \right\|_1   \\
\nonumber &\le  \frac 12 \|   {\cal M}  (O+O^\dag)\|_1  +   \frac1{2}   \|{\cal M}  (iO-iO^\dag)  \|_1   \\
&\le \max  \{\|   {\cal M}  (O+O^\dag)\|_1  \,  , \|{\cal M}  (iO-iO^\dag)  \|_1   \}\, .\label{D10}
\end{align}
Now, the operators  $O+ O^\dag$  and $iO  - iO^\dag$ are Hermitian, and their trace norm is upper bounded by $2\|O\|$, due to the triangle inequality.  Hence, we have the bounds  
\begin{align}
\max\left\{  \frac{\|   {\cal M}  (O +   O^\dag)\|_1}{  \| O+    O^\dag \|_1}  \,, \frac{\|   {\cal M}  (iO - i  O^\dag)\|_1}{  \| iO-    iO^\dag \|_1} \right\} \ge    \frac{  \|{\cal M}  (O)\|_1}{2 \, \|  O\|_1} \, ,  
\end{align}
\begin{align}
\nonumber \|{\cal M}\|_1  & =  \max_{  O   \in   L( \H _A) \, , O\not  = 0}  \,  \frac{\| {\cal M}  (O)\|_1}{\|O\|_1}\\  
\nonumber &\le 2 \max_{  O   \in   L( \H _A) \, , O\not  = 0}   \left\{  \frac{\|   {\cal M}  (O +   O^\dag)\|_1}{  \| O+    O^\dag \|_1}  \,, \frac{\|   {\cal M}  (iO - i  O^\dag)\|_1}{  \| iO-    iO^\dag \|_1} \right\} \\
&\le  2 \max_{  H   \in   L( \H _A) \,  , H^\dag =  H  \, ,  H\not  = 0}  \,   \frac{\| {\cal M}  (H)\|_1}{\|H\|_1} \,.   \label{DD}
\end{align}  
In turn, a Hermitian operator $H$ can be decomposed as $H  =  a \,  \rho_+  -  b  \, \rho_-$, where $\rho_{\pm}$ are density matrices with orthogonal support, while $a$ and $b$ are nonnegative real numbers. 
Hence,   we have  $\|H\|_1  =  a+  b$, and 
\begin{align}
\nonumber     \frac{\| {\cal M}  (H)\|_1}{\|H\|_1}    &=   \frac{\|{\cal M}   (a\rho_+  - b \rho_-)\|_1}{a+b}  \\
 \nonumber    &\le  \frac {a}{a+b}   \|{\cal M}   (\rho_+) \|_1 +  \frac{b}{a+b}   \|  {\cal M}  (  \rho_-)\|_1  \\
 &\le \max_{\rho \in {\sf St}  (A)  }  \, \|  {\cal M}  (\rho)\|_1  \, ,
\end{align}
the first inequality following from the triangle inequality for the trace norm.  
Inserting this bound into Eq. (\ref{DD})  we then obtain the desired inequality. \qed

 \medskip 
 
The second lemma states an elementary relation between the induced trace norm and the diamond norm. 

\begin{lemma}\label{lem:diamondtrace}
For any  Hermitian-preserving map $\cal M:  L(\H _A)  \to L(\H _B)$, one has the bound $\|  {\cal M}\|_{\diamond} \le   \,d_A\, \|  {\cal M}\|_1$.  
\end{lemma}

{\bf Proof.} For a Hermitian-preserving map ${\cal M}$ the diamond norm can be equivalently computed as  $\|  {\cal M}\|_{\diamond} =  \max_{|\Psi\rangle  \in  \H _A \otimes \H _A}  \,  \|   (  {\cal M} \otimes {\cal I}_A)  \,  (|\Psi\rangle\langle \Psi|)\|_1$, where the maximization is over all length-1 vectors $|\Psi\rangle$.    Let $|\Psi\rangle  \in  \H _A \otimes {\cal H_A} $ be a state such that 
\begin{align}\label{aa}
\|  {\cal M}\|_{\diamond}  =\|  ({\cal M} \otimes {\cal I}_A)  (  |\Psi\rangle \langle \Psi|)  \|_1 \,.
\end{align}
  Then, take a Schmidt decomposition
$\ket{\Psi}{=}\sum_{i=1}^{d_A} \sqrt{p_i} {\ket{\alpha_i}}  \otimes {\ket{\beta_i}}$, where $\{  |\alpha_i \rangle\}_{i=1}^{d_A}$ and  $\{  |\beta_i \rangle\}_{i=1}^{d_A}$ are orthonormal bases and $\sum_{i=1}^{d_A} p_i  =1$.   
Inserting this expression into Eq. (\ref{aa}), we obtain  
\begin{align}
    \norm{{\cal M}}_\diamond &=  \norm{\sum_{i,j}\sqrt{p_ip_j}  
    \,  {\cal M}  (\ketbra{\alpha_i}{\alpha_j} ) \otimes \ketbra{\beta_i}{\beta_j}}_1 \nonumber \\
    &\le\sum_{i,j} \sqrt{p_ip_j} \ \norm{{\cal M}\left(\ketbra{\alpha_i}{\alpha_j}\right) \otimes \ketbra{\beta_i}{\beta_j} }_1  \nonumber \\
    &  =\sum_{i,j} \sqrt{p_ip_j} \ \norm{{\cal M}\left(\ketbra{\alpha_i}{\alpha_j}\right)}_1  \,  \norm{\ketbra{\beta_i}{\beta_j}}_1  \nonumber \\
       &=   \|  {\cal M}\|_1 \,  \left(\sum_{i=1}^{d_A}  \sqrt{p_i}\right)^2 \nonumber\\ 
     & \le     d_A\,  \|  {\cal M}\|_1\, .
\end{align}
Here, the  inequality in the second line is due to the triangle inequality for the trace norm, and $\norm{\ \ketbra{\beta_i}{\beta_j}\ }_1=1$. The equality in the third line is due to the multiplicative property of the trace norm  $\|  A  \otimes B\|_1 =  \|  A\|_1 \,  \|B\|_1$,  and the inequality in the fourth line follows from the definition of the induced norm.  The equality in the fifth line follows from the relation $\|  |\alpha_i\rangle \langle \alpha_j| \|_1  = \|  |\beta_i\rangle \langle \beta_j| \|_1  = 1$.  Finally, the inequality in the sixth line follows from  the Cauchy-Schwarz inequality $  (\sum_{i=1}^{d_A}  \sqrt{p_i}  )^2\le (\sum_{i=1}^{d_A}  p_i  )  \, (\sum_{i=1}^{d_A}  1  )    =  d_A$.  The inequality is saturated when $(p_i)_{i=1}^{d_A}$ is  the uniform distribution, that is, when $|\Psi\rangle$ is maximally entangled.  \qed 

\medskip  

{\bf Proof of Proposition~\ref{Lem:faithful}.}  Let $\rho_0$ be an arbitrary state of system $A$, and define the  discard-and-reprepare channel $\mathcal{N}_0:= \sigma_0 \,  \Tr_A $, with $\sigma_0:  =  \N (\rho_0)$.   Then, for every state $\rho  \in  {\sf St}  (A)$, we have 
\begin{align}
\nonumber \|  \mathcal{N}  (\rho)-  \mathcal{N}_0  (\rho) \|_1  &  =   \|  \mathcal{N}  (\rho)-  \sigma_0 \|_1   \\
 \nonumber &  =   \|  \mathcal{N}  (\rho)-  \N  (\rho_0) \|_1   \\
   \nonumber &\le  \CM   (\N)  \,   \|  \rho-  \rho_0\|_1  \\
 &\le  2   \,  \CM   (\N)  \, .
\end{align}
Using Lemmas \ref{lem:diamondtrace}  and \ref{lem:norm1}, we then obtain 
\begin{align}  
\nonumber \| \N  -  {\cal N}_0 \|_{\diamond}  &  \le d_A  \,  \| \N  -  {\cal N}_0 \|_1 \\
\nonumber &\le  2 d_A\,  \max_{\rho \in {\sf St}   (A )} \|  (\N  -  {\cal N}_0) (\rho)\|_1  \\
&\le 4 d_A  \,  \CM  (\N)   \, .
\end{align}
 \qed

\subsection{Proof for Proposition~\ref{lem:classical_capacity}} \label{app:classical_capacity}

For the reader's convenience, we repeat here the proposition's statement: 
\begin{prop*}
For every quantum channel ${\cal N}$, the classical capacity $C(\mathcal{N})$ satisfies the lower bound 
\begin{align}\label{capacitybound}
C(\mathcal{N}) \ge 1 - h_2\left(\frac{1- \CM  (\cal N)  }2 \right)  \,, 
\end{align}
where  $h_2(x):=-x\log_2(x)-(1-x)\log_2(1-x)$ is the binary Shannon entropy.    
    
\end{prop*}

\begin{proof}
Let us assume, to achieve $\CM (\mathcal{N})$ for the quantum channel, the optimal choice of POVM is given by $\{E_b\}$ with $b=\{0,1\}$, $0 \le E_b \le I_B$ and $\sum_b E_b = I_B$. Further, due to Proposition~\ref{prop:pdmax}, we optimize over uniformly distributed pure and orthogonal states. That is for $a=\{0,1\}$, we have $p_a=1/2$ for all $a$ and states $\psi_a\equiv \proj{\psi_a}$ at the system $A$ where $\Tr{\psi_a \psi_{a'}}=\delta_{aa'}$. This $\CM$-achieving setup allows us to define mutual information $I(A;B)_{\CM}$, and we have~\cite{Thomas_M_Cover2006-sa}

\begin{align}
    &I(A;B)_{\CM} {= }H(A) - H(B|A)_{\CM} {=} 1 - H(B|A)_{\CM}\\
    &H(B|A)_{\CM} {=} \frac{1}{2} \sum_{a,b} \Tr{E_b {\cal{N}}(\psi_a)}\log _2(\Tr{E_b {\cal{N}}(\psi_a)}). 
\end{align}
In the second equality in the first line, we note that $H(A)=1$ due to two uniformly distributed orthogonal state preparation at system $A$. Now, due to the definition of $\CM$, the total success probability of distinguishing two output states of the channel $\N$ for all uniformly distributed input states and all possible measurements at the output is given by $P_{\N}=\frac{1+\CM(\mathcal{N})}{2}$. In other words, $P_e = 1- P_{\N}$ gives the probability of error in the same distinguishability task. Now we can use Fano's inequality on $H(B|A)_{\CM}$ to note

\begin{align}
 &H(B|A)_{\CM} {\le} h_2(P_e) + P_e \log_2 (d-1) {=} h_2(1-P_{\N}), {\rm and } \nonumber \\
 &I(A;B)_{\CM} \ge 1- h_2(1-P_{\N})=h_2\left(\frac{1-\CM (\N)}{2}\right).
\end{align}
where $h_2(x):=-x\log_2(x)-(1-x)\log_2(1-x)$ is the binary Shannon entropy. Finally, we prove our theorem by noting the classical capacity $C(\N)$ of the channel $\N$ is more than $I(A;B)_{\CM}$\cite{Wilde13}, i.e.,
\begin{align}
 C(\N)\ge I(A;B)_{\CM} \ge  h_2\left(\frac{1-\CM (\N)}{2}\right). 
\end{align}

\end{proof}

\subsection{Proof of Theorem \ref{theo:correctability}}\label{app:correctability}

The proof is based on the universal recovery channel constructed by Junge {\em et al.} in Ref.~\cite{Junge_2018} and on an integral representation of the quantum relative entropy recently developed by Jen\v cova~\cite{jenčová2023recoverability}, building on work by   Frenkel~\cite{Frenkel_2023}.

Given two states $\rho,  \rho' \in  {\sf St}  (A)$, the quantum relative entropy is $D  (  \rho\|  \rho'):  =  \Tr[  \rho  (  \log \rho -  \log \rho')]$~\cite{Umegaki1962, Hiai1991}, where we take the logarithm to be in base 2.   A well-known fact~\cite{Lindblad1975, Uhlmann1977, Petz1986, PETZ1988}, is that the relative entropy is non-increasing under the application of quantum channels, namely $D  (  \N  (\rho)\|  \N (\rho') )  \le  D  (\rho\|  \rho')$.    

\begin{lemma}\label{lem:dp}
For every quantum channel $\N  :  {\sf St}  (A) \to {\sf St}  (B)$, the bound
\begin{align}
\nonumber   &    D (  \rho\| \rho')-  D(\N (\rho)\|  \N (\rho')) \\
& \qquad \le \frac{1-  {\rm DP}_{\min}  (\N)}2  \,  ( \ln \lambda  -\ln \mu   +\lambda  -  \mu )\, ,
\end{align}  
 holds for every pair of quantum states $\rho$ and $\rho'$ satisfying the condition $\mu  \, \rho'  \le \rho  \le \lambda \, \rho'$ for two nonnegative numbers $\mu$ and $\lambda$.
\end{lemma}

{\bf Proof.}  In Corollary 1 of Ref.~\cite{jenčová2023recoverability}, Jen\v cova  showed that, if two quantum states $\rho$ and $\rho'$ satisfy the condition $\mu  \, \rho'  \le \rho \le \lambda \,\rho'$ for nonnegative $\mu$ and $\lambda$,  then the following integral representation holds:  
\begin{align}
D (  \rho\| \rho')   =  \int_{\mu}^{\lambda}   {\rm d} s ~  \frac { \Tr  \left[  \Big( \rho  -  s  \, \rho' \Big)_- \right]   }s  +  \ln \lambda  + 1-\lambda   \, ,
\end{align}
where $\ln$ denotes the natural logarithm and, for a Hermitian operator $H$, $H_-$ denotes the negative part of $H$. Using the relation $\Tr [  H_-]  =  (\|  H\|_1   -  \Tr[H])/2$, we then obtain  
\begin{align}
D (  \rho\| \rho')   =  \int_{\mu}^{\lambda}   {\rm d} s ~  \frac {\|  \rho  -  s  \, \rho'\|_1   +  s  -1   }{2s} +  \ln  \lambda +1 -\lambda \, .
\end{align}
Using this equation, we obtain the bound 
\begin{align}
\nonumber   &    D (  \rho\| \rho')-  D(\N (\rho)\|  \N (\rho')) \\
\nonumber & \quad =  \int_{\mu}^{\lambda}   {\rm d} s ~  \frac {\|  \rho  -  s  \, \rho'\|_1   -  \| \N( \rho)  -  s  \,\N( \rho')\|_1    }{2s} \\
\nonumber &\quad \le \frac{1-  {\rm DP}_{\min}  (\N)}2 \,  \int_{\mu}^{\lambda}   {\rm d} s ~  \frac {\|  \rho  -  s  \, \rho'\|_1   }{s}\\
\nonumber &  \quad \le \frac{1-  {\rm DP}_{\min}  (\N)}2 \,  \int_{\mu}^{\lambda}   {\rm d} s ~  \frac {1 + s    }{s}\\
&\quad =  \frac{1-  {\rm DP}_{\min}  (\N)}2    \,  ( \ln \lambda  -\ln \mu   +\lambda  -  \mu )\, ,
\end{align}  
where the second inequality comes from the triangle inequality $\|  \rho  -  s  \, \rho'\|_1  \le \|\rho \|_1  +  \|  s  \,  \rho'\|_1  =  1+s$.    \qed  

\medskip  
{\bf Proof of Theorem \ref{theo:correctability}.}  Let $\rho$ be an arbitrary quantum state, and let $I_A/d_A$  be the maximally mixed state of system $A$.  Then, the states $ \sigma: = (\rho  +  I_A/d_A)/2$ and  $\chi:=  I_A/d_A$ satisfy the condition  $\chi/2  \le \sigma \le   (d_A+1)\, \chi/2$.  Applying Lemma \ref{lem:dp} to the states $\sigma$ and $\chi$, we obtain the bound   
\begin{align}
\nonumber   &    D (  \sigma\| \chi)-  D(\N (\sigma)\| \N  (\chi) ) \\
& \qquad \le \frac{1-  {\rm DP}_{\min}  (\N)}2  \,  \left( \ln (d_A+1)   +\frac{d_A}2  \right)\, .
\end{align} 
Now, a key result by Junge,  Renner,  Sutter,  Wilde, and  Winter~\cite{Junge_2018} guarantees the existence of a recovery channel $\cal R$  satisfying the conditions  
\begin{align}\label{uno}
{\cal R}  \circ  \N (\chi)  =  \chi
\end{align}
and 
\begin{align}\label{due}
-\log F   \Big({\cal R}  \circ  \N (\sigma),  \sigma  \Big)  \le   D (  \sigma\| \chi)  -  D(\N (\sigma)  \|  \N  (\chi) )   \, . 
\end{align}
Using the elementary inequality $ -\log  x  \ge  (1-x)/\ln 2$, we then obtain  
\begin{align}
\frac{1  - F   \Big({\cal R}  \circ  \N (\sigma),  \sigma  \Big) }{ \ln 2 }  \le   \,  D (  \sigma\| \chi)  -  D(\N (\sigma)  \|  \N  (\chi) )   \, . 
\end{align}  
Then, the Fuchs-van de Graaf inequality~\cite{fuchs_cryptographic_1999}  implies  
\begin{align}\label{last}
\|  {\cal R}  \circ  \N (\sigma) -   \sigma\|_1   \le  2  \sqrt{\ln 2} \, \sqrt{   D (  \sigma\| \chi)  -  D(\N (\sigma)  \|  \N  (\chi) ) }  \, . 
\end{align}  

Now, using Eq. (\ref{uno}), we obtain the relation ${\cal R}  \circ \N  (\sigma)  = {\cal R}  \circ \N  (\rho+ \chi)/2   =  ({\cal R}  \circ \N  (\rho)  +  \chi  )/2 $, which implies 
\begin{align}
\|  {\cal R}  \circ  \N (\sigma) -   \sigma\|_1   =  \frac{\|  {\cal R}  \circ  \N (\rho) -   \rho\|_1}{2}
\end{align}
Inserting this equality into Eq. (\ref{last}) and using Lemma \ref{lem:dp}, we obtain the bound
\begin{align}\label{last2}
\|  {\cal R}  \circ  \N (\rho) -   \rho\|_1   \le  4  \sqrt{\ln 2} \,  \sqrt{  \frac{1-  {\rm DP}_{\min}  (\N)}2} \,  \sqrt{ \ln (d_A+1)   +  \frac{d_A}2   } .
\end{align}  
\qed

\subsection{Proof of Proposition~\ref{prop:max_min_duality}}\label{App:max_min_duality}

\begin{proof}
\textbf{Statement 1:} From Proposition~\ref{Lem:faithful}, $\rm{CE}_{\max}(\N) \le \epsilon$ implies existence of a discard-and-reprepare channel $\mathcal{E}$ such that $\norm{\N - \mathcal{E}}_{\diamond} \le 4\epsilon d_A$. Then  from Theorem 3 of Kretschmann \emph{et al.}~\cite{kretschmann2008information}, there exists a recovery channel $\mathcal{R}:\St{E'}\to \St{A}$ acting on $\N ^c : \St{A}\to \St{E'}$ such that $\norm{\mathcal{R}\circ \N ^c - \mathcal{I}_A}_\diamond \le 4 \sqrt{\epsilon d_A}$. Now from the continuity of $\rm{CE}_{\min}$ in Proposition~\ref{prop:continuity_max}, and the fact $\rm{CE}_{\min}(\mathcal{I}_A)=1$, we have $\rm{CE}_{\min}(\mathcal{R}\circ \N ^c)\ge 1 - 4\sqrt{\epsilon d_A}$. Finally, Statement 1 follows from the data-processing inequality of  $\rm{CE}_{\min}$ as in Proposition~\ref{prop:propertiesmin}.

\textbf{Statement 2:} From Corollary~\ref{cor:correctability}, we note $\mathrm{CE}_{\min}(\N) \ge 1 -\epsilon$ implies there exists a recovery channel $\mathcal{R}:\St{B} \to \St{A}$ such that $\norm{\mathcal{R}\circ \N - \mathcal{I}_A}_\diamond \le f(\epsilon, d_A)$, where $f(\epsilon,d_A)=8d_A\sqrt{\epsilon \ln 2}\sqrt{d_A/2 + \ln (d_A+1)}$. Now from Theorem 3 of Kretschmann \emph{et al.}~\cite{kretschmann2008information}, there exists a discard and reprepare channel $\mathcal{E}:\St{A}\to \St{E'}$ such that $\norm{\N^c - \mathcal{E}}_\diamond \le 2 \sqrt{f(\epsilon,d_A)}$. Statement 2 now follows from continuity of $\mathrm{CE}_{\max}$ in Proposition~\ref{prop:continuity_max}, and from the fact that $\mathrm{CE}_{\max}(\mathcal{E})=0$.
\end{proof}

\section{Proof of Proposition~\ref{Lem:partial_swap}} \label{app:partial_swap}

{\bf Proof.}
First note the action of the channel $\mathcal{N}$ on a generic matrix $M$ is 
\begin{align}
    &\mathcal{N} (M)=\cos ^2(\theta) \,  \Tr[M] \,  \sigma_{\Lambda} + \sin ^2(\theta) \, M \nonumber \\
    &+ i\sin(\theta)\cos(\theta) \,\Tr_{\Lambda}\bigg[I\bigg(M \otimes \sigma_{\Lambda}\bigg)  {\tt SWAP}\bigg] \nonumber \\
    &- i\sin(\theta) \cos(\theta) \,  \Tr_{\Lambda}\bigg[  {\tt SWAP}\bigg(M \otimes \sigma_{\Lambda}\bigg)I\bigg], \nonumber \\
    &=\cos ^2(\theta)  \,\Tr [M] \,  \sigma_{\Lambda}  {+} \sin ^2(\theta) \,  M   {-}  i{\sin(\theta)\cos(\theta)} \, \left[\sigma_{\Lambda},M  \right] \, ,\label{commutator}  
\end{align}
where $[A,B]:=AB-BA$ is the commutator of two generic $d\times d$ matrices $A$ and $B$.

Hence, we have 
\begin{align}
\CM(\mathcal{N})&=\frac 12 \max_{\psi, \psi_\perp} \norm{\mathcal{N}(\psi )-\mathcal{N}(\psi_\perp )}_1  \nonumber \\
  & =\frac 12 \max_{U:  U^\dag   U  =I }  \norm{\mathcal{N} \left(   U |0\rangle\langle 0|  U^\dag \right)-\mathcal{N}\left(  U  |1 \rangle  \langle 1|  U^\dag \right)}_1  \nonumber \\
  & =\frac 12 \max_{U:  U^\dag   U  =I } \norm{\mathcal{N} \left(   U Z  U^\dag \right)}_1  \nonumber  \, ,
\end{align}
where the maximization is over all unitary operators $U$ and we used the notation 
\begin{align}Z :  =  |0\rangle \langle 0|  -  |1\rangle \langle 1|\,.
\end{align}

Using Eq. (\ref{commutator}) with $M=  UZU^\dag$, we then obtain 
\begin{align}
\CM (\N)  &=\frac 12 \max_{U:  \, U^\dag U  =  I } \Big\|\sin^2(\theta)\,   UZ   U^\dag \nonumber \\
 & \qquad -i \, {\sin(\theta)\cos(\theta)} \left[\sigma_{\Lambda},  U  Z U^\dag  \right] \Big\|_1
\end{align}
For states of the form   
$\sigma_{\Lambda}=(1-p)\frac Id + p \,  \phi$,  with $\phi: = \proj{\phi}$,    the above expression reduces to   \begin{align}
\CM(\mathcal{N})
&=\frac 12 \max_{U:  \, U^\dag U  =  I }\Big \|a \,  UZ   U^\dag -i  b   \left[  \phi,   U  Z U^\dag    \right] \Big\|_1 
\end{align} 
with $a  :  =  \sin^2 \theta$ and $b:  =    p  \sin \theta \, \cos \theta$. 
  Using the unitary invariance of the norm, we then obtain 
\begin{align}
 \nonumber\CM  (\N) &=\frac 12 \max_{U:  \, U^\dag U  =  I } \Big\|a   \, Z   -i b  \, \left[U^\dag\phi  U ,Z\right] \Big\|_1 \\
  &  = \frac 12  \max_{\psi}   \, \Big\|  a  \,  Z    -ib  \, \left[\psi   ,Z\right] \Big\|_1
 \, , 
 \end{align}
where the maximization is over all possible  pure states $\psi  := \proj{\psi}$. 

Now, for a given pure state $\psi$, define the operator $A_\psi:  =  a \, Z  -  i  b\,  [\psi, Z]$, we then have 
\begin{align}
\nonumber \CM(\N)  &   = \frac 12 \max_\psi  \,  \|  A_\psi\|_1  \\
&  = \frac 12  \max_\psi \,  \Tr \left[\sqrt{A_\psi^\dag A_\psi } \,\right] \,.  
\label{sqrtapsiapsi}
\end{align}
Note that one has  
\begin{align}\label{apsiapsi}
A_\psi^\dag A_\psi  =  a^2  \,  P   +    b^2 [Z,\psi]\,  [\psi, Z]  +  i  ab  \,  [P  ,  \psi]\, ,
\end{align}
where $P$ is the projector defined by 
\begin{align}P : =  |0\rangle \langle 0| +  |1\rangle \langle 1| \,.
\end{align}
Using the relation $  Z  =  ZP  =  P Z$,  we then obtain
\begin{align}
\nonumber [\psi, Z] & = |\psi \rangle \langle \psi|  Z  -  Z  |\psi \rangle \langle \psi|  \\
\nonumber  & = |\psi \rangle \langle \psi| P Z  -  Z P |\psi \rangle \langle \psi|  \\
    & = \sqrt F   ~  \Big( |\psi\rangle\langle \xi|Z  -  Z|\xi \rangle \langle \psi|  \Big)\, ,
\end{align}
having defined  $|\xi\rangle :  = P|\psi\rangle/\|  P|\psi\rangle\|$ and $F =  \|  P  |\psi\rangle   \|^2  =  \langle \psi| P |\psi \rangle$.  

To continue, we expand $|\psi\rangle$ as 
\begin{align}
|\psi\rangle =  \sqrt F  \,  |\xi \rangle  +  \sqrt{1-F} \,  |\gamma\rangle  \, ,    \label{psidecomp}
\end{align}
where $|\gamma\rangle$ is a unit vector  satisfying $P| \gamma \rangle =  0$. Furthermore, we expand $|\xi\rangle$ as 
\begin{align}\label{xidecomp}
|\xi\rangle = \alpha \,  Z  |\xi \rangle   +  \sqrt{1- \alpha^2}  \,  Z|\xi_\perp \rangle \, , 
\end{align}
with $\alpha  :=  \langle \xi  |  Z|\xi\rangle$,  and $|\xi_\perp\rangle$  is a suitable unit vector, orthogonal to $|\xi \rangle$ and satisfying $P  |\xi_\perp \rangle =  |\xi_\perp \rangle$.    
Combining the above decompositions, we have 
\begin{align}
|\psi \rangle   =  \alpha   \sqrt F  \,   Z |\xi\rangle   +  \sqrt{1-  \alpha^2  F}  \,  |\eta \rangle        
\end{align}
with 
\begin{align}\label{eta}
|\eta\rangle :  = \frac{  \sqrt{(1-\alpha^2)  F}  \,  Z|\xi_\perp \rangle  +   \sqrt{1- F}  \,  |\gamma \rangle    }{\sqrt{1-  \alpha^2  F}} \, .
\end{align}
 With this notation,  we  have  
\begin{align}
  [\psi, Z]    =   \sqrt F  \,   \sqrt{  1- \alpha^2 F}   \, \Big(      |\eta \rangle \langle \xi|  Z  -  Z|\xi\rangle \langle \eta| \Big)
\end{align}
and 
\begin{align}\label{above}
  [Z,\psi]   [\psi, Z]     =   F  \, \left( 1- \alpha^2 F\right)    \,  Q \, ,
\end{align}
with 
\begin{align}
Q:  = Z|\xi\rangle \langle \xi|  Z  + |\eta\rangle \langle \eta| \,.
\end{align}
Moreover, we have 
\begin{align}
\nonumber [  P, \psi ]  &  =    \sqrt F   \,  (  |\xi  \rangle \langle \psi|  -  |\psi  \rangle \langle \xi|  )\\
&   =     \sqrt {F(1-F)}   \,  (  |\xi  \rangle \langle \gamma|  -  |\gamma  \rangle \langle \xi|  )   \,. \label{above1}
\end{align}

 Inserting  Eqs. (\ref{above}) and (\ref{above1})  into Eq. (\ref{apsiapsi}), we then obtain
\begin{align}
   \nonumber A_\psi^\dag A_\psi &  =   a^2  \,  P    +  b^2   F  \, \left( 1-  \alpha^2 F   \right)    \,  Q \\
   \nonumber & \qquad +  i ab  \sqrt {F(1-F)}   \,  (  |\xi  \rangle \langle \gamma|  -  |\gamma  \rangle \langle \xi|  )    \\
   \nonumber &  =  \Big[a^2  +   b^2  F (1-  \alpha^2  F)\Big]  \,  Z  |\xi  \rangle \langle \xi |Z   +   \\  
  \nonumber  & \qquad +  a^2  Z  |\xi_\perp  \rangle \langle \xi_\perp |Z   +  b^2 F  (1-\alpha^2 F)  \, |\eta\rangle \langle \eta| \\
   &  \qquad +   i ab  \sqrt {F(1-F)}   \,  (  |\xi  \rangle \langle \gamma|  -  |\gamma  \rangle \langle \xi|  ) \, , 
\end{align}
where the second equality follows from the definition of $Q$ and from  the expansion   $P  =  Z |\xi  \rangle \langle \xi |Z  +  Z  |\xi_\perp  \rangle \langle \xi_\perp |Z .$  

Hence, we have 
\begin{align}
   \sqrt{A_\psi^\dag A_\psi} &  =  \sqrt{a^2  +   b^2 F  (1-\alpha^2  F  )}  \,  Z  |\xi  \rangle \langle \xi |Z   +   \sqrt X  
    \, ,  \label{commcomm}
   \end{align}
   with 
   \begin{align}
    \nonumber  X  &:  =  a^2  Z  |\xi_\perp  \rangle \langle \xi_\perp |Z   +  b^2 F  (1-\alpha^2  F)  \, |\eta\rangle \langle \eta|\\\
  &  \qquad + i ab  \sqrt {F(1-F)}   \,  (  |\xi  \rangle \langle \gamma|  -  |\gamma  \rangle \langle \xi|  )\, .   \label{operatorX}  
   \end{align}
The operator $X$ can be written as a matrix with respect to the orthonormal  basis $\{  Z |\xi_\perp\rangle,  |\gamma\rangle\}$.    In this basis, one has 
\begin{align}
 \nonumber 
   Z  |\xi_\perp  \rangle \langle \xi_\perp |Z         & = \begin{pmatrix}  1  &  0  \\
 0  &    0  \end{pmatrix}\\
\nonumber      |\eta\rangle \langle \eta| & = \frac{\begin{pmatrix} 
(1-\alpha^2) F  &  \sqrt{  (1-\alpha^2)  F (1-F)} 
 \\
 \sqrt{  (1-\alpha^2)  F (1-F)} 
  &    1-  F  \end{pmatrix} }{(1-\alpha^2  F) }\\
   |\xi  \rangle \langle \gamma|  -  |\gamma  \rangle \langle \xi|     
   &  = 
 \begin{pmatrix}  0  &  \alpha  \\
 -\alpha  &    0  \end{pmatrix}\,,
\end{align}
and therefore 
\begin{widetext}
\begin{align}
X   =     \begin{pmatrix}  a^2   +   b^2  F  (1-\alpha^2)  F  &    b^2  F  \sqrt{(1-\alpha^2)  F(1-F)}   +  i ab \sqrt{F (1-F)}  \, \alpha  \\
b^2  F  \sqrt{(1-\alpha^2)  F(1-F)}  -    i ab \sqrt{F (1-F)}  \, \alpha   &    b^2  F  (1-F)     \end{pmatrix}
\end{align}
\end{widetext}
Note that $X$ satisfies the relations 
\begin{align}
\nonumber  \Tr[X]  & =  a^2   +  b^2  F   (1-\alpha^2 F)\\
\nonumber  \det (X)   &  =  a^2  b^2  F (1-  F)    (1-\alpha^2) \, ,
\end{align}
where the last equality follows from Eq. (\ref{eta}).
Denoting by $x_+$ and $x_-$ the eigenvalues of $X$,  we then have
\begin{align}
    \nonumber  \left(  \Tr[\sqrt X] \right)^2     & =  (\sqrt{x_+}  +  \sqrt{x_-})^2 \\
  \nonumber  &    =   \Tr  [  X]   +  2  \sqrt{\det (X)}  \\
    &  =  a^2  +     b^2  F  (1- \alpha^2 F)     + 2  |ab | \sqrt{   F   (1-F)  (1-\alpha^2)} \, ,   
\end{align}
which, combined with Eq. (\ref{commcomm}),  implies
   \begin{align}
 \nonumber  & \Tr [\sqrt{  A_\psi^\dag A_\psi}] \\
 \nonumber &  =  \sqrt{a^2  +    b^2   F  (1-\alpha^2  F)}    +  \Tr[\sqrt X]  \\
 \nonumber &=  \sqrt{a^2  +  F  b^2} \\
\nonumber &\qquad + \sqrt{ a^2  +     b^2 F (1-\alpha^2F)    + 2  ab  \sqrt{   F    (1-F) (1-\alpha^2)} } \\
&\le   \sqrt{a^2  +  F b^2 }  +  \sqrt{  a^2  +  F  b^2  +  2ab \sqrt{F(1-F)}} \,,
\end{align}
where the equality holds if and only if $\alpha  =  0$. 
Maximizing the right-hand-side over $F$ and using Eq. (\ref{sqrtapsiapsi}), we then obtain  the bound
\begin{align}\label{f25}
\CM (\N)     \le \max_{F\in  [0,1]}  \frac{\sqrt{a^2  +  F b^2 }  +  \sqrt{  a^2  +  F  b^2  +  2ab \sqrt{F(1-F)}}}2 \, .
\end{align}
The bound can be attained  with the following construction: 
\begin{enumerate}
\item Pick a unit vector $|\gamma\rangle$ such that $\langle \gamma|\phi\rangle =  \sqrt{1-F_*}$, where $F_*$ is the value of $F$ that maximizes the right-hand-side of Eq. (\ref{f25}).  
\item Define the unit vector $|\xi  \rangle  :  =  \frac{|\phi\rangle -  \sqrt{1- F_*}    |\gamma \rangle }{\sqrt F_*} $  
\item Pick a unit vector $|\xi_\perp\rangle$ orthogonal to both $|\xi\rangle$ and $|\phi\rangle$.  
\item Define $|\psi\rangle :  = \frac{  |\xi\rangle +  |\xi_\perp\rangle }{\sqrt 2}$ and $|\psi_\perp\rangle :  = \frac{  |\xi\rangle -  |\xi_\perp\rangle}{\sqrt 2}$.
\end{enumerate}
Setting $\widetilde Z  :  = |\psi\rangle \langle \psi|  -  |\psi_\perp\rangle \langle \psi_\perp| $,   we then have  
\begin{align}
\nonumber    &\frac12  \| \N (|\psi\rangle \langle \psi|  -  |\psi_\perp\rangle \langle \psi_\perp|) \|_1  \\
\nonumber &  =  \frac 12  \|   \N (\widetilde Z)  \|_1  \\
\nonumber 
&  =  \frac12    \Big\|    \sin^2 \theta    \,  \widetilde Z -  i \sin \theta \cos \theta  \left[  p \, \phi  +  (1-p) \frac I d,  \widetilde Z  \right] \Big\|_1 \\  
\nonumber 
&  = \frac12    \Big\|    a \, \widetilde Z   - i   b  \left[     \phi ,  \widetilde Z  \right]\Big\|_1  \\
\nonumber  &     =  \frac 12  \| A \|_1 \qquad A:  =  a \,   \widetilde Z   -  i  b\,  [  \phi  ,  \widetilde Z]  \\
\nonumber &  = \frac 12  \, \Tr  [  \sqrt{A^\dag A}]\\
&  =  \frac 12\,  \Tr [  \sqrt{   a^2  \,  \widetilde P  +  b^2 F_*  \, \widetilde Q}] \, , 
 \end{align}
with 
\begin{align}
\widetilde P :  &= |\psi\rangle \langle \psi|  +  |\psi_\perp\rangle \langle \psi_\perp| =   |\xi\rangle \langle \xi|  +  |\xi_\perp\rangle \langle \xi_\perp| \\
\widetilde Q:  &= |\phi\rangle \langle \phi|  +  |\xi_\perp\rangle\langle \xi_\perp| \, .
\end{align}
We then have  
\begin{align}
\nonumber    \sqrt{   a^2  \,  \widetilde P  +  b^2 F_*  \, \widetilde Q}   &   =    \sqrt{  a^2  +  b^2  F_*} \,   |\xi_\perp\rangle\langle \xi_\perp|  \\
\nonumber & \qquad  +  \sqrt{ a^2  |\xi\rangle \langle \xi|   +  b^2 F_*  \, |\phi\rangle \langle \phi|  }\\
&  = \sqrt{  a^2  +  b^2  F_*} \,   |\xi_\perp\rangle\langle \xi_\perp|     +  \sqrt{\widetilde X} \, ,
\end{align}
with 
\begin{align}
\widetilde X :  = a^2  |\xi\rangle \langle \xi|   +  b^2 F_*  \, |\phi\rangle \langle \phi| 
\end{align}
The evaluation of the trace of the operator $\widetilde X$ is identical to the evaluation of the trace of the operator $X$ in Eq. (\ref{operatorX}), with $\alpha =  0$.  Hence, we obtain  
\begin{align}
\nonumber &  \frac12  \| \N (|\psi\rangle \langle \psi|  -  |\psi_\perp\rangle \langle \psi_\perp|) \|_1  \\
\nonumber &  =  \frac{\Tr [\sqrt{   a^2  \,  \widetilde P  +  b^2 F_*  \, \widetilde Q} ]}2  \\
\nonumber &  = \frac{ \sqrt{  a^2  +  b^2  F_*}  + \Tr [\sqrt{\widetilde X}] }2 \\
&  =  \frac{\sqrt{  a^2  +  b^2  F_*}  +  \sqrt{a^2  +  F_* b^2 }  +  \sqrt{  a^2  +  F_*  b^2  +  2ab \sqrt{F_*(1-F_*)}}}2 \,,
\end{align}
thereby proving the achievability of  the bound  (\ref{f25}).
\qed  

\section{Proof of  Eq. (\ref{minmaclean})} \label{app:minmaclean}

Using Eq.  (\ref{commutator}), the minimum causal effect can be written as 
\begin{align}
\nonumber \Cm(\N) &  =   \min_{\rho,  \rho_\perp}  \frac{  \|  \N  (\rho  -\rho_\perp )\|_1}2  \\
\nonumber &  =  \min_{\Delta}   \frac{\|    \sin ^2(\theta) \,  \Delta    {-}  i{\sin(\theta)\cos(\theta)} \, \left[\sigma_{\Lambda},\Delta \right]\|_1}2 \, ,
\end{align}
with $\Delta  :=  \rho - \rho_\perp$.  On the other hand, the trace-norm of any Hermitian operator $H$ is lower bounded as  
\begin{align}
\|  H\|_1  \ge  \Tr  [ H (P_+  -  P_-) ] \, ,
\end{align}
where $P_+$ and $P_-$ are two orthogonal projectors.  In particular, choosing $P_+$ and $P_-$ to be the projectors on the positive and negative parts of $\Delta$, respectively, yields the bound  
\begin{align}
\nonumber \Cm (\N)  &   \ge  \frac 12  \min_{\Delta}        \, \sin^2\theta \,  \Tr[\Delta  (P_+  -  P_-)]\\
\nonumber &   \qquad  - i{\sin(\theta)\cos(\theta)} \, \Tr\left\{\left[\sigma_{\Lambda},\Delta \right]   \,  (P_+  - P_-)\right\}    \\
\nonumber 
&   =  \sin^2\theta \, ,
\end{align}
where we used the relations  
\begin{align}
\Tr [  \Delta (P_+  -  P_-)]  &= \Tr  [| \Delta|]    =  \|  \Delta\|_1  =  2  
\end{align}
and 
\begin{align}
\nonumber \Tr\left\{\left[\sigma_{\Lambda},\Delta \right]   \,  (P_+  - P_-)\right\}    &   =  \Tr  [ \sigma_\Lambda   |\Delta|   -   |\Delta|  \sigma_\Lambda   ]  =  0 \, .
\end{align}

\section{Proof of Proposition~\ref{prop:superposition}} \label{app:superposition}

Applying the definition of  maximum causal effect to the  quantum channel $\N_\sigma$ in  Eq. (\ref{Nsigma}), we obtain
\begin{align}
\nonumber &\CM (\N_\sigma)  \\
\nonumber &=  \max_{|\psi\rangle  ,  |\psi_\perp\rangle : \, \langle \psi |\psi_\perp \rangle   = 0}    \frac{  \|   \N_\sigma  (|\psi\rangle \langle \psi| -|\psi_\perp\rangle \langle \psi_\perp|  )\,  \|_1  }2\\
 \nonumber &    =  \max_{|\psi\rangle  ,  |\psi_\perp\rangle : \, \langle \psi |\psi_\perp \rangle   = 0}   \left\| \frac{    {\cal C} (|\psi\rangle \langle \psi| -|\psi_\perp\rangle \langle \psi_\perp|  )    \otimes \sigma_{\rm diag}  }2  \right.  \\
 \nonumber &   \qquad   \qquad \qquad \qquad + \left.  \frac{    F (|\psi\rangle \langle \psi| -|\psi_\perp\rangle \langle \psi_\perp|  )F^\dag \otimes \sigma_{\rm off-diag}   }2 \right\|_1 \\
  \nonumber &\le  \max_{|\psi\rangle  ,  |\psi_\perp\rangle : \, \langle \psi |\psi_\perp \rangle   = 0}   \left\| \frac{    {\cal C} (|\psi\rangle \langle \psi| -|\psi_\perp\rangle \langle \psi_\perp|  )  \otimes \sigma_{\rm diag}  }2  \right\|_1 \\
 \nonumber &    \qquad + \max_{|\psi\rangle  ,  |\psi_\perp\rangle : \, \langle \psi |\psi_\perp \rangle   = 0}   \left\|  \frac{    F (|\psi\rangle \langle \psi| -|\psi_\perp\rangle \langle \psi_\perp|  )F^\dag    }2 \otimes \sigma_{\rm off-diag} \right\|_1 \\
 \nonumber &=    \max_{|\psi\rangle  ,  |\psi_\perp\rangle : \, \langle \psi |\psi_\perp \rangle   = 0}   \left\| \frac{    {\cal C} (|\psi\rangle \langle \psi| -|\psi_\perp\rangle \langle \psi_\perp|  )   }2  \right\|_1 \|  \sigma_{\rm diag}   \|_1\\
 \nonumber &    \qquad + \max_{|\psi\rangle  ,  |\psi_\perp\rangle : \, \langle \psi |\psi_\perp \rangle   = 0}   \left\|  \frac{    F (|\psi\rangle \langle \psi| -|\psi_\perp\rangle \langle \psi_\perp|  )F^\dag    }2  \right\|_1 \|  \sigma_{\rm off-diag} \|_1\\
 \nonumber &\le  \CM (\cal C)    \\
 \nonumber &    \qquad + \max_{|\psi\rangle  ,  |\psi_\perp\rangle : \, \langle \psi |\psi_\perp \rangle   = 0}   \frac{   \|  F (|\psi\rangle \langle \psi| )F^\dag\|_1  + \|   F(\psi_\perp\rangle \langle \psi_\perp|)   F^\dag \|_1   }2 \\
 \nonumber &  \qquad  \qquad \qquad \qquad \qquad  \qquad \qquad \qquad \qquad \times \|  \sigma_{\rm off-diag}\|_1 \\
 \nonumber &=    \CM({\cal C}) \\
 \nonumber &    \qquad + \max_{|\psi\rangle  ,  |\psi_\perp\rangle : \, \langle \psi |\psi_\perp \rangle   = 0}   \frac{      \langle \psi| F^\dag  F |\psi\rangle     +  \langle \psi_\perp| F^\dag    F(\psi_\perp\rangle   }2   \|  \sigma_{\rm off-diag}\|_1 \\
 &  = \CM({\cal C})   +  \frac{  \|  F^\dag F\|_{\rm Ky Fan 2} \, \|\sigma_{\rm off-diag} \|_1}2 \, ,   \label{wsx}   \end{align}
where the last equality follows from the definition of the Ky Fan $2$-norm.   

Now, suppose that the original channel $\cal C$ has zero $\CM$. In this case, the bound (\ref{wsx})  is attained by choosing $|\psi\rangle$ and $|\psi_\perp \rangle$ to be eigenvectors of $F^\dag F$ corresponding to the two largest eigenvalues.  
 \qed 
\medskip 

To conclude, we prove the bound $\|  \sigma_{\rm off-diag}\|_1 \le 2   (1-1/k)$, valid for every quantum state $\sigma$ in a $k$-dimensional Hilbert space.  To this purpose, we use the relation 
\begin{align}
\sigma_{\rm diag}   =  \frac 1k  \,    \sum_{j=0}^{k-1}   U^j \sigma  U^{-j}  \, ,
\end{align}
where $U$ is the unitary  operator defined by  $U: =   \sum_l    \exp [ 2\pi  i   jl/k]  \,  |l\rangle \langle l|$.

Hence, we have 
\begin{align}
\nonumber \|  \sigma_{\rm off-diag}\|_1  & = \|  \sigma  -  \sigma_{\rm diag}\|_1 \\
\nonumber &  =  \left\|    \frac{  k-1}k  \sigma  -  \frac 1k \sum_{j=1}^{k-1}  U^j  \sigma U^{-j}  \right\|_1\\
\nonumber &  \le \left\|  \frac{  k-1}k  \sigma \right\|_1  +  \frac 1k \sum_{j=1}^{k-1}    \left\|   U^j  \sigma U^{-j} \right\|_1\\
&   =  2   \left( 1  -\frac 1k   \right) \, .
\end{align}
The equality is attained if and only if the states $U^j\sigma U^{-j}$ have orthogonal supports, which happens if and only if $\sigma$ is pure and maximally coherent.  

\section{Proof of Proposition \ref{prop:superposition1}}\label{app:superposition1}  

Applying the definition of minimum  causal effect to the  quantum channel $\N_\sigma$ in  Eq. (\ref{Nsigma}), we obtain
\begin{align}
\nonumber &\Cm (\N_\sigma)  \\
\nonumber &=  \min_{\rho  , \,  \rho_\perp :\, \Tr[\rho \rho_\perp]  =  0}    \frac{  \|   \N_\sigma  (\rho-\rho_\perp  )\,  \|_1  }2\\
 \nonumber &    =  \min_{\rho  , \,  \rho_\perp :\, \Tr[\rho \rho_\perp]  =  0}    \left\| \frac{    {\cal C} (\rho-\rho_\perp)    \otimes \sigma_{\rm diag}  }2  \right.  \\
 \nonumber &   \qquad   \qquad \qquad \qquad + \left.  \frac{    F (\rho-\rho_\perp )F^\dag \otimes \sigma_{\rm off-diag}   }2 \right\|_1 \\
  \nonumber &=  \min_{\rho  , \,  \rho_\perp :\, \Tr[\rho \rho_\perp]  =  0}       \left\|  \frac{    F (\rho-\rho_\perp )F^\dag    }2 \otimes \sigma_{\rm off-diag} \right\|_1 \\
 \nonumber &=   
 \min_{\rho  , \,  \rho_\perp :\, \Tr[\rho \rho_\perp]  =  0}  
 \frac{\|  F (\rho-\rho_\perp)F^\dag \|_1}  2  \,
\|  \sigma_{\rm off-diag} \|_1\\
 \nonumber &=   
 \min_{\rho  , \,  \rho_\perp :\, \Tr[\rho \rho_\perp]  =  0}  
 \frac{\Tr[\sqrt{F (\rho-\rho_\perp)F^\dag F  
 (\rho-\rho_\perp)F^\dag}]}  2  \,
\|  \sigma_{\rm off-diag} \|_1\\
 \nonumber &\ge  \sqrt{\min{\rm eigv}  (F^\dag F)}\, \|  \sigma_{\rm off-diag} \|_1 \, \\
 \nonumber & \qquad  \qquad  \times \min_{\rho  , \,  \rho_\perp :\, \Tr[\rho \rho_\perp]  =  0}    \frac{    \Tr[\sqrt{F (\rho-\rho_\perp)^2 F^\dag}  ]  }2  \\
    \nonumber &=  \sqrt{\min{\rm eigv}  (F^\dag F)}\, \|  \sigma_{\rm off-diag} \|_1  \,  \\
    \nonumber  & \qquad \qquad \times \min_{\rho  , \,  \rho_\perp :\, \Tr[\rho \rho_\perp]  =  0}    \frac{    \Tr[ \sqrt{(\rho-\rho_\perp) F^\dag  F  (\rho-\rho_\perp)}  ]}2  \\
    \nonumber  &\ge   {\min{\rm eigv}  (F^\dag F)}\,  \|  \sigma_{\rm off-diag} \|_1  \,  \min_{\rho  , \,  \rho_\perp :\, \Tr[\rho \rho_\perp]  =  0}    \frac{    \Tr[ \sqrt{(\rho-\rho_\perp)^2}  ]}2\\
     &=   {\min{\rm eigv}  (F^\dag F)}\,  \|  \sigma_{\rm off-diag} \|_1    \,,
   \end{align}
   where we used the relation $F^\dag F  \ge  {\min{\rm eigv}  (F^\dag F)} \,  I$ and the operator monotonicity of the square root. 
\qed

\section{Measurement realization }\label{App:strongly_entangling}

  A basic fact used by our algorithm is that  any  two-outcome measurement on $k$ qubits can be realized by letting the $k$ qubits interact with an auxiliary qubit $F$, and by measuring the auxiliary qubit in the computational basis.  This fact is a simple consequence of Naimark's theorem. Consider a generic quantum measurement,  described by two positive operators $P_0$ and $P_1$ satisfying the normalization condition $P_0  +  P_1  =  I^{\otimes k}$.    The probability distribution of obtaining outcome $j\in  \{0,1\}$ when  this measurement is performed on an arbitrary state $\rho$ is given by  
\begin{align}
\nonumber p(j| \rho)  & =    \Tr  [ P_j \, \rho ]\\
&  =  \Tr  [  V   \rho  V^\dag  (I^{\otimes k} \otimes |j\rangle \langle j|)] \qquad \forall j\in  \{0,1\} \, , 
\end{align}
 where $V:   {\mathbb C}^{\otimes k}  \to  {\mathbb  C}^{\otimes (k+1)}$ is an isometry from the Hilbert space of $k$ qubits to the Hilbert space of $k+1$ qubits, given by 
 \begin{align}
 V:=   \sqrt{P_0} \otimes |0\rangle +  \sqrt{  P_1} \otimes |1\rangle \, .       
 \end{align}
In turn,  the isometry $V$ can be extended to a unitary gate $U$ acting on $k+1$ qubit, and satisfying the relation 
\begin{align}
V  =  U(  I^{\otimes k}  \otimes |0\rangle)\, .
\end{align}
Physically, this construction implies that every binary measurement on $k$ qubit can be implemented by preparing an auxiliary qubit in the state $|0\rangle$,  applying a unitary gate on $k+1$ qubits, and measuring the auxiliary qubit in the computational basis.

\end{document}